\documentclass[aps,pra,preprint,superscriptaddress,nofootinbib]{revtex4-2}

\PassOptionsToPackage{%
  colorlinks=true,
  linkcolor=palatinateblue,
  citecolor=brightpink,
  urlcolor=amaranth,
  linktoc=all
}{hyperref}
\usepackage{hyperref}

\usepackage{amsmath,amssymb,amsfonts}
\usepackage{graphicx}

\usepackage{amsthm,mathtools}

\AtBeginDocument{%
}

\theoremstyle{remark}  

\newtheorem{theorem}{Theorem}          
\newtheorem{lemma}[theorem]{Lemma}    
\newtheorem{proposition}[theorem]{Proposition}
\newtheorem{corollary}[theorem]{Corollary}
\newtheorem{definition}[theorem]{Definition}

\usepackage{extpfeil}

\usepackage{cancel}
\usepackage{xcolor}

\usepackage{MnSymbol}

\usepackage{enumitem}

\usepackage{empheq}

\providecommand{\rbr}[1]{\left( #1 \right)}%
\providecommand{\sqbr}[1]{\left[ #1 \right]} %
\providecommand{\mtn}[1]{{\scriptscriptstyle \! #1}}%
\providecommand{\sgn}[1]{\mathrm{sgn} \left( #1 \right) }%

\def\ra{\rightarrow}
\def\Ra{\Rightarrow}

\definecolor{amaranth}{rgb}{0.9, 0.17, 0.31}
\definecolor{palatinateblue}{rgb}{0.15, 0.23, 0.89}
\definecolor{brightpink}{rgb}{1.0, 0.0, 0.5}

\newcommand{\pbar}{\lower1.45ex\hbox{$\mathchar'26$}\mkern-7mu p}
\newcommand{\qbar}{\lower1.45ex\hbox{$\mathchar'26$}\mkern-11mu q}

\newcommand{\bbar}{\lower0.1ex\hbox{$\mathchar'26$}\mkern -8.5mu b}

\begin{document}


\title{Catalytic quantum thermodynamics \\beyond additivity and reduced-state monotones}
\author{Ali Can G\"unhan}
\thanks{Corresponding author}
\email{alicangunhan@mersin.edu.tr}
\affiliation{Department of Physics, Mersin University, 33343 Mersin, T\"urkiye}

\author{Onur Pusuluk}
\email{onur.pusuluk@gmail.com}
\affiliation{Faculty of Engineering and Natural Sciences, Kadir Has University, 34083 İstanbul, T\"urkiye}

\author{Thomas Oikonomou}
\email{thomas.oikonomou1@gmail.com}
\affiliation{Institute of Physical Chemistry, National Center for Scientific Research ``Demokritos'', 15310 Athens, Greece}

\author{G.\ Baris Bagci}
\email{gbagci@mersin.edu.tr}
\affiliation{Department of Physics, Mersin University, 33343 Mersin, T\"urkiye}

\date{\today}

\begin{abstract}
 
The generalized second laws of quantum thermodynamics are usually formulated in terms of R\'enyi divergences and the associated family of generalized free energies. In catalytic thermal transformations, this framework typically certifies the existence of a suitable catalyst but does not make the catalytic contribution explicit in the resulting system-level inequalities. Here we develop a complementary formulation based on non-additive divergences, whose pseudo-additive structure yields a family of generalized free energies with an explicit catalyst-dependent correction term. For uncorrelated catalytic thermal transformations, we show that this leads to non-additive second-law relations providing nontrivial constraints on admissible catalysts when the catalyst is returned only approximately. We also analyze correlated catalytic thermal transformations and show, through explicit finite-dimensional examples, that reduced-state data are generally insufficient to characterize thermodynamic accessibility: the thermo-majorization behavior of the joint transformation can change while the system and catalyst marginals remain fixed, and even states with identical marginals and the same mutual information can exhibit different thermo-majorization accessibility.

\end{abstract}
\maketitle

\section{Introduction}
\label{sec::intro}

The generalization of the second law to small-scale systems out of equilibrium~\cite{Janzing2000, HorodeckiOppenheim2013, Brandao2013, Renes2014, 2015_PRL_EnTO, Gour2018, 2018_Quantum_eTOs, ContThermo2022} has become one of the central achievements of contemporary quantum thermodynamics~\cite{Book2018}, building on earlier order-theoretic ideas such as mixing character and mixing distance~\cite{MeadRuch1976, 1978_Ruch_etall_Thermomajorization}. In particular, the seminal work of Brand\~ao et al.~\cite{BrandaoEtAl2015} identified an infinite family of generalized free-energy conditions for state convertibility under exact catalytic thermal operations, expressed through R\'enyi divergences and the associated generalized free energies; see also Refs.~\cite{Gour2021, Farooq2024, Verhagen2025} for later refinements and rigorous formulations. These results, together with subsequent developments on correlated catalytic transformations~\cite{2015_PRL_Lostaglio_Muller, 2017_Entropy_CorrelatedCatal, PhysRevX.8.041051, 2022_PRL_CorrelatedCatal, 2025_PRL_CorrCatalAndFreeE, 2025_arXiv_CorrCatalAndFreeE}, have played a foundational role in the modern understanding of how the second law emerges beyond the thermodynamic limit~\cite{2025_EPL_Shiraishi_CorrelatedCatalyzerAndUniqueF}.

The use of catalysts as resources for otherwise inaccessible state transformations originated in quantum information theory with the entanglement-assisted transformations of Jonathan and Plenio~\cite{JonathanPlenio1999}. In the thermodynamic setting, despite the subsequent progress outlined above, two conceptual aspects of catalytic thermodynamics remain insufficiently transparent. First, in the  uncorrelated catalytic setting~\cite{BrandaoEtAl2015, 2024_RMP_Review}, the R\'enyi-based framework typically yields existential statements: if the generalized free-energy conditions are satisfied, then there exists a catalyst enabling the desired transition. While such results establish convertibility in principle, they do not directly reveal how the catalyst itself enters the thermodynamic balance. Moreover, beyond the exact setting, approximate catalysis is known to be operationally fragile: if catalytic return is only required to hold up to a trace-distance error $\varepsilon$, then for every $\varepsilon>0$, no matter how small, there exist catalyst states of sufficiently large dimension enabling arbitrary state transitions~\cite{2017_NJP_Ng_LimitsOfCatal}. Nontrivial constraints are restored only when additional physical restrictions, such as bounded catalyst dimension or finite average energy, are imposed. In this regime, approximate uncorrelated catalysis becomes a finite-resource accounting problem whose feasibility depends on the interplay between catalyst dimension, spectral structure, and the thermodynamic demand of the target transition. This motivates seeking complementary formulations in which the catalyst's athermality, i.e., its divergence from the corresponding thermal state, enters the generalized free-energy balance directly, rather than being constrained only indirectly.

Second, in correlated catalytic settings~\cite{2015_PRL_Lostaglio_Muller, 2017_Entropy_CorrelatedCatal, PhysRevX.8.041051, 2022_PRL_CorrelatedCatal, 2025_PRL_CorrCatalAndFreeE, 2025_arXiv_CorrCatalAndFreeE}, where the catalyst is recovered only marginally, and correlations are allowed to build up with the system, it is not clear to what extent any family of reduced-state thermodynamic monotones can remain complete. Recent finite-size results have begun to make correlated catalysis more operational by relating catalytic feasibility to multicopy transformation rates, by deriving sufficient finite-dimensional catalyst bounds, and by identifying genuinely state-dependent effects such as catalytic resonance~\cite{2025_PRA_FiniteSizeCatal}. However, such advances still leave open the question of whether, in the correlated regime, the relevant state dependence can ever be fully captured by reduced-state data, even when supplemented with simple scalar descriptors of correlation.

Both issues ultimately trace back to how generalized free energies behave under composition.

In this work, we address the first issue by developing a complementary formulation based on non-additive divergences and the corresponding family of generalized free energies~\cite{Tsallis1988}. Unlike the R\'enyi family, these divergences satisfy a pseudo-additive composition law~\cite{FuruichiYanagiKuriyama2004}, so that the generalized free energy of a product state contains an explicit cross term. Rather than treating this term as a technical inconvenience, we retain it and interpret it as a thermodynamically meaningful contribution. This leads to a non-additive family of second-law inequalities for catalytic thermal transformations in which the catalyst enters explicitly through the pseudo-additive correction. We then turn to the second issue and show, through explicit thermo-majorization examples, that in finite-dimensional correlated catalytic transformations the relevant thermodynamic constraints cannot, in general, be captured completely at the level of reduced states alone, even when simple scalar correlation measures are matched.

The remainder of this work is organized as follows. In Sec.~\ref{sec::RTofTD}, we introduce the resource-theoretic framework and notation underlying both the uncorrelated and correlated catalytic settings considered here. In Sec.~\ref{sec::T2ndLaw}, we derive the non-additive second laws for uncorrelated catalytic thermal transformations, and in Sec.~\ref{sec::PseudoAdd}, we show how the pseudo-additive correction yields explicit catalyst-dependent constraints in approximate finite-dimensional settings. In Sec.~\ref{sec:CorrCatTherTransBeyondRedStateMonotones}, we then examine correlated catalytic transformations through explicit thermo-majorization examples, which reveal that finite-size accessibility can depend on genuinely joint-state features that are not, in general, recoverable from reduced states alone. Finally, in Sec.~\ref{sec::discuss}, we discuss the broader implications of these results for generalized second laws and catalytic quantum thermodynamics, including their connection to a companion study~\cite{CompanionQ2Laws} in which the role of intrinsic equilibrium correlations becomes structurally unavoidable.

\section{Framework and preliminaries}
\label{sec::RTofTD}
Quantum resource theories provide a general framework to study state transformations under restricted sets of operations, and have found applications in entanglement theory, quantum thermodynamics, coherence, non-locality and quantum correlations~\cite{Horodecki2009, Brunner2014,
Lostaglio2015, LostaglioEtAl2015, Adesso_2016, Gour20151, Goold_2016, Streltsov2017}. In general, a resource theory is based on two main ingredients: a class of free operations that has zero implementation cost, and free states that can be generated and used without a cost. If a state is not free, then it is called a resource. For example, in the resource theory of coherence, incoherent states are free, while coherence is a resource~\cite{Chitambar2019}.

The resource theory then attempts to determine which quantum states are possible to obtain by manipulating an initial quantum state given that one has an arbitrary number of copies of free states and free operations at disposal. The possibility of such processes is often expressed in terms of a particular function which decreases at the end of the quantum state transitions~\cite{Chitambar2019}.

Thermodynamics in this regard can be viewed as a resource theory, since its second law relies on the non-increase of the Helmholtz free energy $F$ when temperature and volume are kept constant. In other words, a free transition between the states is possible if and only if it corresponds to a non-increasing free energy. However, thermodynamics is originally a macroscopic theory, which deals with physical systems composed of a large number of weakly interacting degrees of freedom. As a result, macroscopic thermodynamics rests on neglecting fluctuations and assumes weak correlations, and cannot be expected to remain a viable resource theory at much smaller scales, where strong fluctuations and correlations become dominant.

Such a resource-theoretical approach has recently been developed for quantum thermodynamics~\cite{Janzing2000,Brandao2013,HorodeckiOppenheim2013,BrandaoEtAl2015,Gour2018}. In that quantum resource approach to thermodynamics, the free states are thermal (i.e., Gibbs) states, while the free operations are the energy conserving transformations that preserve the thermal state. The underlying principle governing the state transitions in small-scale thermodynamics is commonly formulated in terms of an infinite family of non-increasing generalized free energies $F_\mtn{\alpha}^{\mtn{R}}$. Under the assumptions that the states are block-diagonal in the energy eigenbasis and that catalytic transitions are uncorrelated, the free transition from a state $\rho_\mtn{S}$ to $\rho'_\mtn{S}$ can be achieved with arbitrary accuracy under catalytic thermal operations if and only if the conditions $\Delta F_\mtn{\alpha}^\mtn{R} \leq 0$ hold for all $\alpha \geq 0$, where $\Delta F^\mtn{R}_\mtn{\alpha} = F^\mtn{R}_\mtn{\alpha} (\rho'_\mtn{S}) - F^\mtn{R}_\mtn{\alpha} (\rho_\mtn{S})$ and 
\begin{eqnarray}\label{eq:rule_uncorrelated_R}
F^\mtn{R}_\mtn{\alpha} (\rho_\mtn{S}) = k_\mtn{B} T D^\mtn{R}_\mtn{\alpha} (\rho_\mtn{S} \| \gamma_\mtn{S}) - k_\mtn{B} T \ln (Z_{\mtn{S}})\,,
\end{eqnarray}
with $k_\mtn{B}$  the Boltzmann constant, $T$ the temperature, and $Z_\mtn{S}$ the partition function of the corresponding thermal state; $ D^\mtn{R}_\mtn{\alpha} (\rho_\mtn{S} \| \gamma_\mtn{S})$ is the R\'enyi divergence between the state $\rho_\mtn{S} $ and the corresponding thermal state $\gamma_\mtn{S}=e^{-\beta H_\mtn{S}}/Z_\mtn{S}$ \cite{BrandaoEtAl2015}. This generalized free energy definition $F^\mtn{R}_\mtn{\alpha}$ reduces to the Helmholtz free energy $F$ in the limit $\alpha \rightarrow 1 $. Since all states considered in this work are block-diagonal in the energy eigenbasis, divergences of states are throughout understood as the corresponding classical divergences of their eigenvalue distributions in a common eigenbasis, unless stated otherwise. Between two probability vectors $p,q\in\mathcal{I}_n^{+}$, where $\mathcal{I}_n\coloneqq\{(x_1,\ldots,x_n)\in\mathbb{R}^n \,|\, x_i\ge 0,\ \sum_{i=1}^{n}x_i=1\}$ denotes the probability simplex and $\mathcal{I}_n^{+}$ its interior ($x_i>0$), the R\'enyi divergence \cite{BrandaoEtAl2015} for $\alpha\in\mathbb{R}$ reads
\begin{eqnarray}\label{eq:renyi_div1}
D^\mtn{R}_\mtn{\alpha} (p\|q) &=&
\frac{\mathrm{sgn}(\alpha)}{\alpha-1} \ln\rbr{\sum_{i=1}^{n} \rbr{\frac{p_i}{q_i}}^\alpha q_i}.
\end{eqnarray}
Its non-negativity, i.e., $D^\mtn{R}_\mtn{\alpha} (p\|q)\geq 0$, follows directly from Jensen's inequality~\cite{jensen_2018_2371297}. The divergence properties collected in this section translate, through Eq.~\eqref{eq:rule_uncorrelated_R} and its non-additive counterpart Eq.~\eqref{eq:rule_uncorrelated}, directly into properties of the generalized free energies. The relation  between the R\'enyi divergence and the R\'enyi entropy $H^\mtn{R}_\mtn{\alpha}(p)$ is given as
\begin{eqnarray}\label{eq:renyi_uniform}
D^\mtn{R}_\mtn{\alpha} (p\|\eta) = \sgn\alpha\ln (n)-H^\mtn{R}_\mtn{\alpha}(p)\,,
\end{eqnarray}
where $\eta=\lbrace \frac{1}{n},\frac{1}{n},...,\frac{1}{n}  \rbrace$ is the uniform probability distribution, and the R\'enyi entropy $H^\mtn{R}_\mtn{\alpha}(p)$ \cite{Renyi1961}, for any probability distribution $p\in \mathcal{I}^+_n$, reads
\begin{eqnarray}\label{eq:renyi_entropy2}
H^\mtn{R}_\mtn{\alpha}(p)\coloneqq\frac{\mathrm{sgn}(\alpha)}{1-\alpha}\ln\left(\sum_{i=1}^{n}p_i^\alpha\right)\,.
\end{eqnarray}
The parameter $\alpha\in\mathbb{R}$ denotes the order of the measure, $\ln(x)$ is the natural logarithm, and the function $\sgn{x}$ is $+1$ or $-1$ for $x\geq0$ and $x<0$, respectively \cite{BrandaoEtAl2015}.
The R\'enyi entropy is additive over product probability distributions $\forall\alpha\in\mathbb{R}$, and yields the Shannon entropy $H(p)\coloneqq-\sum_{i=1}^{n}p_i \ln(p_i)$ as $\alpha \rightarrow 1$.

The criteria $\Delta F_\mtn{\alpha}^\mtn{R} \leq 0$, for all $\alpha \geq 0$, rest on three assumptions. First, the states are block-diagonal in the energy eigenbasis: coherences between energy sectors decay on the decoherence time scale, assumed far shorter than the thermalization time; moreover, an impossibility proof rules out incorporating such coherences into this picture~\cite{LostaglioMuller2019,Marvian2019}. Second, the admissible transitions are of the form $\rho_\mtn{S} \otimes \sigma_\mtn{M} \rightarrow \rho_\mtn{S}' \otimes \sigma_\mtn{M}$, where $\sigma_\mtn{M}$ denotes the state of the thermal machine containing a catalyst and possibly a work bit; the work bit is a two-level system always in an energy eigenstate---ground or excited---and provides the extracted or spent work, while the machine returns to its initial state and can thus be used again. Both the initial and the final joint states are product states: the system and the machine end up uncorrelated even if correlations build up during the process. Lastly, the aforementioned criteria assume additivity of the free energy $F^\mtn{R}_\mtn{\alpha}$, i.e., 
\begin{eqnarray}\label{eq:denemelibelli}
F^\mtn{R}_\mtn{\alpha} (\rho_\mtn{S} \otimes \sigma_\mtn{M} ) =  F^\mtn{R}_\mtn{\alpha} (\rho_\mtn{S}) +  F^\mtn{R}_\mtn{\alpha} (\sigma_\mtn{M}),
\end{eqnarray}
or, equivalently, of the R\'enyi divergence over uncorrelated pairs, 
\begin{eqnarray}\label{eq:denemeli}
D_\mtn{\alpha}^\mtn{R}\Bigl(\rho_\mtn{S} \otimes \sigma_\mtn{M} \|\gamma_\mtn{S} \otimes \gamma_\mtn{M} \Bigr) = D_\mtn{\alpha}^\mtn{R}\Bigl(\rho_\mtn{S}\|\gamma_\mtn{S}\Bigr)+D_\mtn{\alpha}^\mtn{R}\Bigl(\sigma_\mtn{M}\|\gamma_\mtn{M}\Bigr),
\end{eqnarray}
with $\gamma_\mtn{M}=e^{-\beta H_\mtn{M}}/Z_\mtn{M}$ the thermal state of the machine. It is this last assumption that the next section relaxes. In what follows, unless otherwise indicated, generalized quantities without a superscript refer to the non-additive formulation. The expressions associated with the R\'enyi entropy will carry a superscript~$R$. With this notation in place, we proceed to derive the non-additive second laws governing catalytic state transformations. These take the form of monotonicity constraints expressed in terms of non-additive divergences.

\section{Non-additive second laws from non-additive divergences}
\label{sec::T2ndLaw}
We now introduce a non-additive formulation of the generalized free-energy conditions by replacing the additive R\'enyi divergence with a non-additive divergence. In this formulation, the transition from a state $\rho_\mtn{S}$ to $\rho'_\mtn{S}$ is characterized by the quantities $\Delta F_\mtn{\alpha}=F_\mtn{\alpha} (\rho'_\mtn{S}) - F_\mtn{\alpha} (\rho_\mtn{S})$, where $F_\mtn{\alpha} (\rho_\mtn{S})$ is defined as
\begin{eqnarray}\label{eq:rule_uncorrelated}
F_\mtn{\alpha} (\rho_\mtn{S}) = k_\mtn{B} T D_\mtn{\alpha} (\rho_\mtn{S} \| \gamma_\mtn{
S}) - k_\mtn{B} T \ln (Z_\mtn{S})\,,
\end{eqnarray} 
and the non-additive divergence \cite{BorlandPlastinoTsallis1998,BorlandPlastinoTsallis1999E,FuruichiYanagiKuriyama2004} for $\alpha\in\mathbb{R}$ reads
\begin{eqnarray}\label{eq:tsallis_div1_defn}
D_\mtn{\alpha} (p\|q) =
\frac{\mathrm{sgn}(\alpha)}{\alpha-1}\sqbr{\sum_{i=1}^{n}\rbr{\frac{p_i}{q_i}}^\alpha q_i-1}\,.
\end{eqnarray}
Here the $\alpha\to 1$ member is understood as the limiting Kullback--Leibler form, $D_\mtn{1}(p\|q)=\sum_i p_i\ln(p_i/q_i)$. In the appendices, and wherever the relevant Hamiltonian must be displayed explicitly, we write $F_\mtn{\alpha}(\rho,H)$; when the Hamiltonian is clear from context, we suppress the second argument, as in Eq.~\eqref{eq:rule_uncorrelated}.
This divergence satisfies the pseudo-additive composition law
\begin{widetext}
\begin{eqnarray}\label{eq:non_add_relation}
D_\mtn{\alpha}\Bigl(\rho_\mtn{S} \otimes \sigma_\mtn{M} \|\gamma_\mtn{S} \otimes \gamma_\mtn{M} \Bigr) = D_\mtn{\alpha}\Bigl(\rho_\mtn{S}\|\gamma_\mtn{S}\Bigr)+D_\mtn{\alpha}\Bigl(\sigma_\mtn{M}\|\gamma_\mtn{M}\Bigr)+\sgn{\alpha}
(\alpha-1)
D_\mtn{\alpha}\Bigl(\rho_\mtn{S}\|\gamma_\mtn{S}\Bigr)D_\mtn{\alpha}\Bigl(\sigma_\mtn{M}\|\gamma_\mtn{M}\Bigr),
\end{eqnarray}
\end{widetext}
in general. For $\alpha\neq 1$, the explicit cross term renders it non-additive; at $\alpha=1$ the cross term vanishes and additivity is restored.
By the pseudo-additivity of $D_\mtn{\alpha}$, the free-energy expression $F_\mtn{\alpha}$ likewise fails to be additive for $\alpha\neq 1$, 
\begin{widetext}
\begin{eqnarray}\label{eq:non_add_free}
F_\mtn{\alpha}(\rho_\mtn{S} \otimes \sigma_\mtn{M} ) 
&=&
F_\mtn{\alpha} (\rho_\mtn{S}) 
+
F_\mtn{\alpha} (\sigma_\mtn{M})
\nonumber\\
&&{}+
\sgn{\alpha} \frac{(\alpha-1)}{k_\mtn{B}T}
\Big[ F_\mtn{\alpha} (\rho_\mtn{S}) + k_\mtn{B} T \ln (Z_\mtn{S}) \Big]\Big[F_\mtn{\alpha} (\sigma_\mtn{M}) + k_\mtn{B} T \ln (Z_{\mtn{M} })\Big],
\end{eqnarray}
\end{widetext}
retaining an additional term depending on quantum resource and classical free energy expressions when compared with Eq. (\ref{eq:denemelibelli}) .

Under the same assumptions used in the additive catalytic framework, this non-additive formulation provides a complementary representation of the generalized second laws. The precise statement is the following.

\begin{theorem}[Non-additive second laws]\label{thm3}
A system with Hamiltonian $H_\mtn{S}$ that is in contact with a heat reservoir at inverse temperature $\beta$ can be transformed with arbitrary accuracy from a state $\rho^i_\mtn{S}$ block-diagonal in the energy eigenbasis to another state $\rho^f_\mtn{S}$, also block-diagonal in the energy eigenbasis, under catalytic thermal operations if and only if the generalized non-additive free energies defined in Eq.~\eqref{eq:rule_uncorrelated} satisfy the inequality
\begin{equation}\label{eq:Thm18_1App}
F_\mtn{\alpha}(\rho_\mtn{S}^i)\geq F_\mtn{\alpha}(\rho^f_\mtn{S}),\qquad\forall\alpha\in\mathbb{R}.
\end{equation}
\end{theorem}

The proof, given in Appendix~\ref{appE}, rests on the catalytic $d$-majorization characterization of Theorem~\ref{thm2} (Appendix~\ref{appD}), and runs in both directions: the divergence conditions $D_\mtn{\alpha}(\rho^i_\mtn{S}\|\gamma_\mtn{S})\geq D_\mtn{\alpha}(\rho^f_\mtn{S}\|\gamma_\mtn{S})$, for all $\alpha\in\mathbb{R}$, imply the existence of a channel that preserves the thermal reference, returns the catalyst exactly, and reproduces $\rho^f_\mtn{S}$ with arbitrary accuracy, while, conversely, the existence of such a channel enforces the same divergence ordering. In this sense, the existence of generalized second laws is robust against the loss of additivity of the free-energy expression: the non-additive $F_\mtn{\alpha}$ also defines a family of monotones.

This result also indicates that there may exist different representations of generalized second laws in terms of different divergence families even when one uses the same trumping relations as a point of departure (see Appendices~\ref{appC} and~\ref{appD}).

Theorem~\ref{thm3} characterizes exact uncorrelated catalytic transformations through system-level free-energy inequalities. To make the catalyst dependence under the same conditions explicit, we now rewrite them in terms of the joint initial and final product states. The resulting factorization does not define a different transformation preorder; rather, it exposes the pseudo-additive contribution that is hidden when the conditions are expressed solely in terms of the system.

\begin{theorem}[Catalyst-explicit form of the non-additive second laws]\label{thm4}
Consider a system $S$ with Hamiltonian $H_\mtn{S}$ that is in contact with a heat
reservoir at inverse temperature $\beta$, and let the catalyst $C$, with Hamiltonian
$H_\mtn{C}$, be prepared in a full-rank state $\rho^i_\mtn{C}$, so that the total
initial and final states can be written in the form
$\rho_\mtn{SC}=\rho_\mtn{S}\otimes\rho^i_\mtn{C}$. Let the states of $S$, $\rho^i_\mtn{S}$ and
$\rho^f_\mtn{S}$, be block-diagonal in the energy eigenbasis. Then the joint
generalized free-energy inequalities
\begin{widetext}
\begin{equation}\label{eq:QSecondLaws}
    F_\mtn{\alpha}(\rho^i_\mtn{S}\otimes\rho^i_\mtn{C},H_\mtn{S}+H_\mtn{C})\ge
    F_\mtn{\alpha}(\rho^f_\mtn{S}\otimes\rho^i_\mtn{C},H_\mtn{S}+H_\mtn{C}),
    \qquad\forall\alpha\in(-\infty,+\infty),
\end{equation}
\end{widetext}
hold if and only if the system-level conditions of Theorem~\ref{thm3} do.
\end{theorem}

This equivalence concerns the two families of generalized free-energy inequalities; it does not assert that the prescribed $\rho^i_\mtn{C}$ itself catalyzes the transformation. For a fixed $\rho^i_\mtn{C}$, the existence of a thermal operation implementing $\rho^i_\mtn{S}\otimes\rho^i_\mtn{C}\mapsto\rho^f_\mtn{S}\otimes\rho^i_\mtn{C}$ is governed by the thermo-majorization relation of the joint states, whereas Theorem~\ref{thm3} concerns the existence of \textit{some} catalyst under which the target system transformation can be achieved with arbitrary accuracy.

\begin{proof}
In terms of divergences, \eqref{eq:QSecondLaws} reads
\begin{widetext}
\begin{equation}\label{eq:QSecondLawsDivs}
    D_\mtn{\alpha}(\rho^i_\mtn{S}\otimes\rho^i_\mtn{C}\|\gamma_\mtn{S}\otimes\gamma_\mtn{C}) \ge
    D_\mtn{\alpha}(\rho^f_\mtn{S}\otimes\rho^i_\mtn{C}\|\gamma_\mtn{S}\otimes\gamma_\mtn{C}),\qquad\forall\alpha\in(-\infty,+\infty). 
\end{equation}
\end{widetext}
Here, $\gamma_\mtn{C}$ denotes the thermal state of $C$ at inverse temperature $\beta$.
After applying pseudo-additivity of $D_\mtn{\alpha}$~\cite{FuruichiYanagiKuriyama2004} and rearranging, we have
\begin{widetext}
\begin{equation}\label{eq:Thm4Proof1App}
    \Biggl\{1+\sgn{\alpha}(\alpha-1) D_\mtn{\alpha}\Bigl(\rho^i_\mtn{C}\|\gamma_\mtn{C}\Bigr)\Biggr\}\times
    \Biggl[D_\mtn{\alpha}\Bigl(\rho^i_\mtn{S}\|\gamma_\mtn{S}\Bigr)-D_\mtn{\alpha}\Bigl(\rho^f_\mtn{S}\|\gamma_\mtn{S}\Bigr)\Biggr]\ge 0,\qquad\forall\alpha\in(-\infty,+\infty). 
\end{equation}
\end{widetext}

Since both $\rho^i_\mtn{C}$ and $\gamma_\mtn{C}$ are of full rank,
$D_\mtn{\alpha}(\rho^i_\mtn{C}\|\gamma_\mtn{C})$ is finite for all
$\alpha\in\mathbb{R}$. Moreover, from the proof of Theorem~\ref{thm2}, the term in
$\{\cdot\}$ is non-negative for all $\alpha$ (see
Eqs.~(\ref{eq:MonotonicityDT3_1App})--(\ref{eq:MonotonicityDT3_4App}) in
particular); full-rankness of $\gamma_\mtn{C}$ in fact makes it
strictly positive. Hence,
condition~\eqref{eq:Thm4Proof1App} becomes
\begin{equation}\label{eq:Thm4Proof2App}
    D_\mtn{\alpha}\Bigl(\rho^i_\mtn{S}\|\gamma_\mtn{S}\Bigr)-D_\mtn{\alpha}\Bigl(\rho^f_\mtn{S}\|\gamma_\mtn{S}\Bigr)\ge 0,\qquad\forall\alpha\in(-\infty,+\infty). 
\end{equation}
In terms of free energies, \eqref{eq:Thm4Proof2App} is equivalent to \eqref{eq:Thm18_1App}. That is, making the catalyst explicit in the free-energy balance leaves the
non-additive second laws unchanged.
\end{proof}

The factorized inequality \eqref{eq:Thm4Proof1App} displays explicitly how the catalyst enters the non-additive free-energy balance: its athermality $D_\mtn{\alpha}(\rho^i_\mtn{C}\|\gamma_\mtn{C})$ appears as a multiplicative, dimension- and spectrum-dependent prefactor of the system divergence difference. For exact uncorrelated catalysis this prefactor is strictly positive and therefore does not alter the transformation conditions; its finite-size consequences, once the catalyst is returned only approximately, are the subject of Sec.~\ref{sec::PseudoAdd}.

Another issue of interest is the difference between the maximal amount of work that can be extracted from a system
 in contact with a heat bath at temperature $T$ and the minimal amount that must be invested to prepare it in the same state.  In the additive setting, this difference is nonnegative: the work of formation is no smaller than the work of extraction~\cite{HorodeckiOppenheim2013}. The non-additive divergence is monotonically nondecreasing in the parameter $\alpha$, such that $F_{\mtn{\alpha} = 0} (\rho_\mtn{S}) \le F_{\mtn{\alpha} \rightarrow 1} (\rho_\mtn{S}) \le F_{\mtn{\alpha} = \infty} (\rho_\mtn{S})$, where $F_{\mtn{\alpha} \rightarrow 1} (\rho_\mtn{S})$ corresponds to the Helmholtz free energy. 
The maximal extractable work and the minimal formation cost are governed by the infimum ($\alpha=0$) and the supremum ($\alpha\to\infty$) of this family over $\alpha\geq 0$, respectively (Appendix~\ref{appG}). 

Hence, in the presence of non-additivity as well, the minimal amount of work required to prepare the system in a given state is at least the maximal amount of work reliably extracted from it in that state. In the limit $\alpha \rightarrow 1$, where the Helmholtz free energy becomes the sole choice, these two work values coincide.

These work values acquire operational meaning by appending to the system a two-level weight system, the \textit{work bit}~\cite{HorodeckiOppenheim2013,BrandaoEtAl2015}: the energy gained by the qubit defines the work extracted from the system, while the energy lost by the qubit defines the work invested for the transformation to take place. The work-bit-assisted form of the non-additive second laws is then the following.

\begin{theorem}[Work-bit-assisted non-additive second laws]\label{thm5}
A system with Hamiltonian $H_\mtn{S}$ that is in contact with a heat reservoir at inverse temperature $\beta$ can be transformed with arbitrary accuracy from a state $\rho^i_\mtn{S}$ block-diagonal in the energy eigenbasis to another state $\rho^f_\mtn{S}$, also block-diagonal in the energy eigenbasis, under catalytic thermal operations assisted by a qubit $Q$ in a pure energy eigenstate $\rho_\mtn{Q}$ with Hamiltonian $H_\mtn{Q}$, if and only if the corresponding free energies satisfy the inequality
\begin{widetext}
\begin{equation}\label{eq:QSecondLawsQubit}
    F_\mtn{\alpha}(\rho^i_\mtn{S}\otimes\rho^i_\mtn{C}\otimes\rho^i_\mtn{Q},H_\mtn{S}+H_\mtn{C}+H_\mtn{Q})\ge F_\mtn{\alpha}(\rho^f_\mtn{S}\otimes\rho^i_\mtn{C}\otimes\rho^f_\mtn{Q},H_\mtn{S}+H_\mtn{C}+H_\mtn{Q}),\qquad\forall\alpha\in[0,+\infty). 
\end{equation}
\end{widetext}
\end{theorem}

As in Theorem~\ref{thm3}, the catalyst here is existential: the theorem asserts the existence of \textit{some} catalyst under which the composite transition $\rho^i_\mtn{S}\otimes|0\rangle\langle0|\mapsto\rho^f_\mtn{S}\otimes|\omega\rangle\langle\omega|$ can be achieved with arbitrary accuracy. The work bit is part of the transformed composite, not a catalyst; its state changes deterministically and records the work.

The proof is given in Appendix~\ref{appF}. In contrast with Theorems~\ref{thm3} and~\ref{thm4}, the family of conditions is here restricted to $\alpha\in[0,+\infty)$: since the work bit is prepared in a pure state, it is not of full rank and the corresponding divergences do not converge for $\alpha<0$; evaluated on the arbitrarily accurate approximants of Appendix~\ref{appC}, those conditions are then trivially satisfied, so the informative family is $\alpha\in[0,+\infty)$.

The asymmetry between the extractable work and the work of formation is quantified by the notion of \textit{work distance}, defined in the resource theory of quantum thermodynamics as follows:
\begin{equation}\label{eq:Workable_dist}
  \mathcal{D}_{\text{work}}(\rho_\mtn{S} \succ\rho_\mtn{S}') \coloneqq \inf_\mtn{\alpha\ge 0}[F^\mtn{R}_\mtn{\alpha}(\rho_\mtn{S}) - F^\mtn{R}_\mtn{\alpha}(\rho_\mtn{S}')]\,.
\end{equation}
The work distance, originally introduced in Ref.~\cite{BrandaoEtAl2015} in terms of the additive R\'enyi divergence, remains invariant under the non-additive free-energy expression (see Appendix~\ref{appG}). Hence, the same definition can be consistently maintained. This establishes the non-additive second laws at the level of marginal free-energy differences, and sets the stage for analyzing their behavior in the presence of finite-size effects and system--catalyst correlations.

\subsection{Comparison with the additive framework and the choice of divergence}\label{sec:helm-comparison}

It is instructive to relate Theorems~\ref{thm3}--\ref{thm5} to the additive formulation of the second laws, expressed through the R\'enyi free energies $F^\mtn{R}_\mtn{\alpha}$ of Refs.~\cite{HorodeckiOppenheim2013,BrandaoEtAl2015}. At the level of allowed transitions, nothing changes: for each fixed $\alpha\neq 1$, the non-additive and R\'enyi divergences are related by the strictly monotone map of Eq.~\eqref{eq:TsallisRenyiMonotone}, so that each condition in Eq.~\eqref{eq:Thm18_1App} is equivalent to its R\'enyi counterpart of the same order [Eq.~\eqref{eq:RenyiTsallisOrderEquiv}, Appendix~\ref{appCC}], and the work distance is likewise unchanged (Appendix~\ref{appG}). In the limit $\alpha\to 1$, both families collapse to the Helmholtz free energy, $F(\rho_\mtn{S})=-k_BT\ln Z_\mtn{S}+k_BT\,D_\mtn{1}(\rho_\mtn{S}\|\gamma_\mtn{S})$, the monotone underlying the single second law of macroscopic thermodynamics~\cite{HorodeckiOppenheim2013}. The non-additive family is therefore a deformation of the same Helmholtz limit, not of the transition preorder.

What distinguishes the two formulations is the composition rule, and hence the bookkeeping of the catalyst. The R\'enyi divergence is additive over products, Eq.~\eqref{eq:denemeli}, so that in exact uncorrelated catalysis the catalyst cancels identically in $\Delta F^\mtn{R}_\mtn{\alpha}$ and leaves no trace in the system-level inequality. The divergence of Eq.~\eqref{eq:tsallis_div1_defn} instead obeys the pseudo-additive law, Eq.~\eqref{eq:non_add_relation}: the catalyst enters through an explicit multiplicative cross term, cf.~Eq.~\eqref{eq:Thm4Proof1App}; it is precisely this structure that Sec.~\ref{sec::PseudoAdd} exploits---through the exact factorization of Eq.~\eqref{eq:DeltaF-exact-factorized}---to recast approximate catalysis as a finite-size feasibility problem. A second, more technical difference concerns the index range. While the additive second laws of Refs.~\cite{HorodeckiOppenheim2013,BrandaoEtAl2015} are stated only for $\alpha\ge 0$, the characterization of Theorem~\ref{thm3} runs over all $\alpha\in\mathbb{R}$, with the conditions being evaluated on arbitrarily accurate approximants of the target state (Appendix~\ref{appC}).

These observations also make explicit the physical and operational motivation for the particular choice of divergence. Mere non-additivity carries no constructive content; it asserts only that additivity fails. What the present framework requires is (i) monotonicity under the relevant channels for all $\alpha\in\mathbb{R}$ (Appendix~\ref{appD}), (ii) order equivalence with the R\'enyi family at each fixed index, ensuring no transition conditions are lost or invented (Appendix~\ref{appCC}), and (iii) closed-form composition with an explicit cross term. The Tsallis-type divergence of Eq.~\eqref{eq:tsallis_div1_defn} satisfies all three simultaneously: its pseudo-additive law couples the system and catalyst athermalities through a single multiplicative term that vanishes at $\alpha=1$. It is this structured departure from additivity---rather than the departure itself---that makes the catalytic contribution explicit in Theorems~\ref{thm4} and~\ref{thm5}, and quantifiable in Sec.~\ref{sec::PseudoAdd}, while guaranteeing that the additive framework is recovered in the Helmholtz limit.

\section{Finite-size catalytic implications of pseudo-additivity}
\label{sec::PseudoAdd}

In this section, we analyze these effects explicitly in the finite-size regime. Hereafter, unless noted otherwise, the thermal machine carries no work bit, so $M$---with states $\sigma_\mtn{M}$ and dimension $d_\mtn{M}$---is precisely the catalyst $C$ of Theorems~\ref{thm3}--\ref{thm5}. The pseudo-additive structure of non-additive free-energy expression, as encoded in the last term of Eq.~\eqref{eq:non_add_free}, has an immediate consequence for catalytic thermodynamics: in contrast with the additive cancellation exhibited in Eq.~\eqref{eq:denemelibelli}, the catalytic contribution does not disappear from the monotonicity condition even for uncorrelated catalytic transformations. This makes the thermodynamic role of the catalyst explicit and turns catalytic feasibility into a finite-resource accounting problem rather than a purely existential statement. The detailed derivations for the exact and approximate cases, together with the finite-dimensional benchmark examples discussed below, are presented in Appendix~\ref{app:pseudo-additivity-finite-size}.

For exact uncorrelated catalysis of the form
\begin{equation}
  \rho_\mtn{S} \otimes \sigma_\mtn{M} \;\longmapsto\; \rho'_\mtn{S} \otimes \sigma_\mtn{M} ,
\end{equation}
the total non-additive free-energy change factorizes into a system term multiplied by a catalyst-dependent correction (Appendix~\ref{app:exact-uncorr-cat}),

\begin{equation}\label{eq:DeltaF-exact-factorized}
  \Delta F_\mtn{\alpha}
  =
  \Delta F_\mtn{\alpha}^{\mtn{(S)}}
  \Bigl[
    1 + \mathrm{sgn}(\alpha)(\alpha-1)D_\mtn{\alpha}(\sigma_\mtn{M}\|\gamma_\mtn{M})
  \Bigr],
\end{equation}
where $\Delta F_\mtn{\alpha}^{\mtn{(S)}}=F_\mtn{\alpha}(\rho_\mtn{S}')-F_\mtn{\alpha}(\rho_\mtn{S})$ is the free-energy change of the system alone. This is the free-energy form of the factorized inequality \eqref{eq:Thm4Proof1App}: the catalyst's athermality $D_\mtn{\alpha}(\sigma_\mtn{M}\|\gamma_\mtn{M})$ enters multiplicatively and remains present in the generalized second-law relation.
For trivial catalyst Hamiltonians, however, this exact condition alone does not yet imply a nontrivial finite-size restriction. The finite-size content of the non-additive formulation becomes operative when the catalyst is allowed to be returned only approximately.

To see how it operates, we consider approximate uncorrelated catalysis,
\begin{equation}
  \rho_\mtn{S} \otimes \sigma_\mtn{M} \;\longmapsto\; \rho'_\mtn{S} \otimes \sigma'_\mtn{M},
\end{equation}
under a trace-distance constraint
\begin{equation}\label{eq:return-error-def}
  \frac12\|\sigma'_\mtn{M}-\sigma_\mtn{M}\|_1 \le \varepsilon .
\end{equation}
In this setting, pseudo-additivity yields an explicit trade-off between the target transformation on the system $S$, the allowed catalytic return error $\varepsilon$, and the finite-dimensional structure of the catalyst.

The trade-off takes a closed form in the benchmark of a trivial catalyst Hamiltonian, $H_\mtn{M}=0$, so that the reference thermal state $\gamma_\mtn{M}$ is maximally mixed, and the catalyst states are taken diagonal, $\sigma_\mtn{M}=\mathrm{diag}(\vec{p})$ and $\sigma'_\mtn{M}=\mathrm{diag}(\vec{q})$, on a $d_\mtn{M}$-dimensional space. Introducing the spectral power sums and the shorthand
\begin{widetext}
\begin{equation}\label{eq:PQA-shorthand}
P_\mtn{\alpha} \coloneqq \sum_i p_i^\alpha,
\qquad
Q_\mtn{\alpha} \coloneqq \sum_i q_i^\alpha,
\qquad
A_\mtn{\alpha} \coloneqq \frac{k_BT\big(1+(\alpha-1)\ln Z_\mtn{S}\big)}{\alpha-1},
\end{equation}
\end{widetext}
the monotonicity condition $\Delta F_\mtn{\alpha}\le 0$, expanded by pseudo-additivity (Appendix~\ref{app:approx-uncorr-cat}), is guaranteed by the sufficient condition
\begin{equation}
|Q_\mtn{\alpha}-P_\mtn{\alpha}|\;\bigl|F_\mtn{\alpha}(\rho_\mtn{S}')+A_\mtn{\alpha}\bigr|
\;\le\;
P_\mtn{\alpha}\,\bigl(-\Delta F_\mtn{\alpha}^{\mtn{(S)}}\bigr).
\label{eq:tradeoff-sufficient}
\end{equation}
Combining Eq.~\eqref{eq:tradeoff-sufficient} with the continuity bounds on $|Q_\mtn{\alpha}-P_\mtn{\alpha}|$ derived in Appendix~\ref{app:approx-uncorr-cat} yields explicit bounds on the admissible catalytic return error:\\ for $\alpha\ge 1$,
\begin{equation}
\varepsilon
\;\le\;
\frac{P_\mtn{\alpha}\,\bigl(-\Delta F_\mtn{\alpha}^{\mtn{(S)}}\bigr)}{2\alpha\,\bigl|F_\mtn{\alpha}(\rho_\mtn{S}')+A_\mtn{\alpha}\bigr|},
\label{eq:eps-bound-alpha-ge-1}
\end{equation}
while for $0<\alpha<1$,
\begin{equation}
\varepsilon
\;\le\;
\frac{1}{2}\left[
\frac{P_\mtn{\alpha}\,\bigl(-\Delta F_\mtn{\alpha}^{\mtn{(S)}}\bigr)}{d_\mtn{M}^{1-\alpha}\,\bigl|F_\mtn{\alpha}(\rho_\mtn{S}')+A_\mtn{\alpha}\bigr|}
\right]^{\!1/\alpha}.
\label{eq:eps-bound-alpha-lt-1}
\end{equation}
Equations~\eqref{eq:eps-bound-alpha-ge-1} and~\eqref{eq:eps-bound-alpha-lt-1} establish the catalyst dependence of the admissible return error $\varepsilon$: for a fixed target transition $\rho_\mtn{S}\mapsto\rho'_\mtn{S}$, the admissible $\varepsilon$ depends on the spectral structure of the catalyst through $P_\mtn{\alpha}$; the catalyst dimension $d_\mtn{M}$ enters only in the regime $\alpha<1$. Two cautionary remarks are in order. First, these bounds are sufficient, but not necessary, for approximate catalysis within the class of admissible profiles considered here. Second, the continuity bound underlying the $0<\alpha<1$ case is a worst-case estimate, not yet a direct proof of a sharp minimum catalyst dimension; the robust statement is that \textit{non-additive monotonicity conditions, combined with an explicit finite-dimensional ansatz for approximate catalyst return, generate nontrivial finite-size feasibility thresholds.} Thus, in contrast with the way this dependence appears in additive formulations, approximate catalysis is naturally recast here as a finite-size feasibility problem at the level of the generalized monotonicity condition itself.
Within the trivial-Hamiltonian benchmark, if the returned catalyst is chosen to be uniform, then the pseudo-additive correction is controlled by how the initial catalyst distribution deviates from uniformity. Physically, this benchmark catalyst is a $d_\mtn{M}$-level ancilla whose levels all have the same energy, so that its only thermodynamic resources are its dimension and the spectral profile of its state. Two simple examples, worked out in Sec.~\ref{app:finite-benchmark-examples}, illustrate this point. In the first, the trace-distance error $\varepsilon$ is distributed symmetrically over half of the catalyst levels, and to leading nontrivial order in $\varepsilon$,
\begin{equation}
\Delta F_\mtn{\alpha}
\approx
\Delta F_\mtn{\alpha}^{\mtn{(S)}}
-
2\,\varepsilon^2\,\alpha(\alpha-1)\bigl(F_\mtn{\alpha}(\rho_\mtn{S})+A_\mtn{\alpha}\bigr).
\label{eq:I32_sub_I15c}
\end{equation}
In the second, the same return-error scale is concentrated into only two levels, and
\begin{equation}
\Delta F_\mtn{\alpha}
\approx
\Delta F_\mtn{\alpha}^{\mtn{(S)}}
-
d_\mtn{M}\,\varepsilon^2\,\alpha(\alpha-1)\bigl(F_\mtn{\alpha}(\rho_\mtn{S})+A_\mtn{\alpha}\bigr).
\label{eq:I43_sub_I15c}
\end{equation}
Although both constructions satisfy the same approximate-return constraint at fixed catalyst dimension $d_\mtn{M}$, the resulting pseudo-additive corrections in the generalized monotonicity condition scale as $2\varepsilon^2$ for the distributed profile and as $d_\mtn{M}\varepsilon^2$ for the concentrated one, the latter enhancement being tempered by the positivity requirement $\varepsilon\le 1/d_\mtn{M}$. At $d_\mtn{M}=2$ the two constructions coincide, and so do the corrections. Hence, the pseudo-additive correction depends not only on the size of the catalytic deviation, but also on how that deviation is distributed across the catalyst spectrum.

It is worth noting that this profile sensitivity is not exclusive to the non-additive formulation: in additive R\'enyi-based treatments, the catalyst contribution can also distinguish between different spectral profiles at fixed return error. What is distinctive here is instead the mode of thermodynamic bookkeeping: within the pseudo-additive framework, this dependence appears directly in the generalized free-energy balance itself, rather than only through a separate additive catalyst term.

The benchmark assumption of a trivial catalyst Hamiltonian can moreover be lifted, and doing so exposes a further, sharper dependence: on the catalyst's energy gap. For a general full-rank thermal reference $\gamma_\mtn{M}=\mathrm{diag}(g_1,\dots,g_{d_\mtn{M}})$ and diagonal catalyst states $\sigma_\mtn{M}=\mathrm{diag}(\vec{p})$, $\sigma'_\mtn{M}=\mathrm{diag}(\vec{q})$, the pseudo-additive expansion closes into the exact identity (Appendix~\ref{app:gapped-catalyst})
\begin{widetext}
\begin{equation}\label{eq:DeltaF-general-identity}
\Delta F_\mtn{\alpha}
=
W_\mtn{\alpha}(\vec{p})\,\Delta F_\mtn{\alpha}^{\mtn{(S)}}
+
\bigl[W_\mtn{\alpha}(\vec{q})-W_\mtn{\alpha}(\vec{p})\bigr]\bigl[F_\mtn{\alpha}(\rho'_\mtn{S})+A_\mtn{\alpha}\bigr],
\qquad
W_\mtn{\alpha}(\vec{v})\coloneqq\sum_i v_i^{\alpha}\,g_i^{1-\alpha},
\end{equation}
\end{widetext}
valid for all $\alpha>0$; for uniform $\gamma_\mtn{M}$ one has $W_\mtn{\alpha}=d_\mtn{M}^{\alpha-1}P_\mtn{\alpha}$, and Eqs.~\eqref{eq:tradeoff-sufficient}--\eqref{eq:eps-bound-alpha-lt-1} are recovered. Monotonicity is therefore guaranteed by the sufficient condition $|W_\mtn{\alpha}(\vec{q})-W_\mtn{\alpha}(\vec{p})|\,|F_\mtn{\alpha}(\rho'_\mtn{S})+A_\mtn{\alpha}| \le W_\mtn{\alpha}(\vec{p})\,(-\Delta F_\mtn{\alpha}^{\mtn{(S)}})$, and for $\alpha\ge 1$ the admissible return error obeys
\begin{equation}\label{eq:eps-bound-gapped}
\varepsilon
\;\le\;
\frac{W_\mtn{\alpha}(\vec{p})\,\bigl(-\Delta F_\mtn{\alpha}^{\mtn{(S)}}\bigr)}{2\alpha\, g_{\mathrm{min}}^{\,1-\alpha}\,\bigl|F_\mtn{\alpha}(\rho_\mtn{S}')+A_\mtn{\alpha}\bigr|},
\qquad
g_{\mathrm{min}}\coloneqq \min_i g_i,
\end{equation}
with the $0<\alpha<1$ analog given in Appendix~\ref{app:gapped-catalyst}. For a two-level catalyst with energy gap $E_\mtn{M}$ at the bath inverse temperature $\beta$, $g_{\mathrm{min}}^{\,1-\alpha}=e^{(\alpha-1)\beta E_\mtn{M}}\bigl(1+e^{-\beta E_\mtn{M}}\bigr)^{\alpha-1}$, so that, for a thermally prepared catalyst ($W_\mtn{\alpha}(\vec p)=1$) and $\alpha>1$, the admissible error contracts exponentially in the gap, $\varepsilon_{\mathrm{max}}\propto e^{-(\alpha-1)\beta E_\mtn{M}}$ for $\beta E_\mtn{M}\gg 1$. At fixed dimension $d_\mtn{M}=2$, it is thus the Boltzmann weight of the catalyst spectrum, rather than the dimension, that sets the return-error budget: a stiff catalyst must be returned exponentially more precisely.

Taken together, these examples support a physically transparent conclusion: in approximate uncorrelated catalysis, the non-additive monotonicity conditions do more than test consistency with the second laws in principle. They quantify how catalyst dimension, spectral profile, energy gaps, and the thermodynamic demand of the target transition combine to set the admissible return error---and hence when catalytic assistance is ruled out at finite resources. These observations highlight that pseudo-additivity leads to a qualitatively different sensitivity to the structure of the finite catalytic deviations, which cannot be captured by the return error alone.

\section{Correlated catalytic thermal transformations\texorpdfstring{\\}{} beyond reduced-state monotones}
\label{sec:CorrCatTherTransBeyondRedStateMonotones}
\begin{figure*} [htb]
    \centering
    \includegraphics[width=\textwidth]{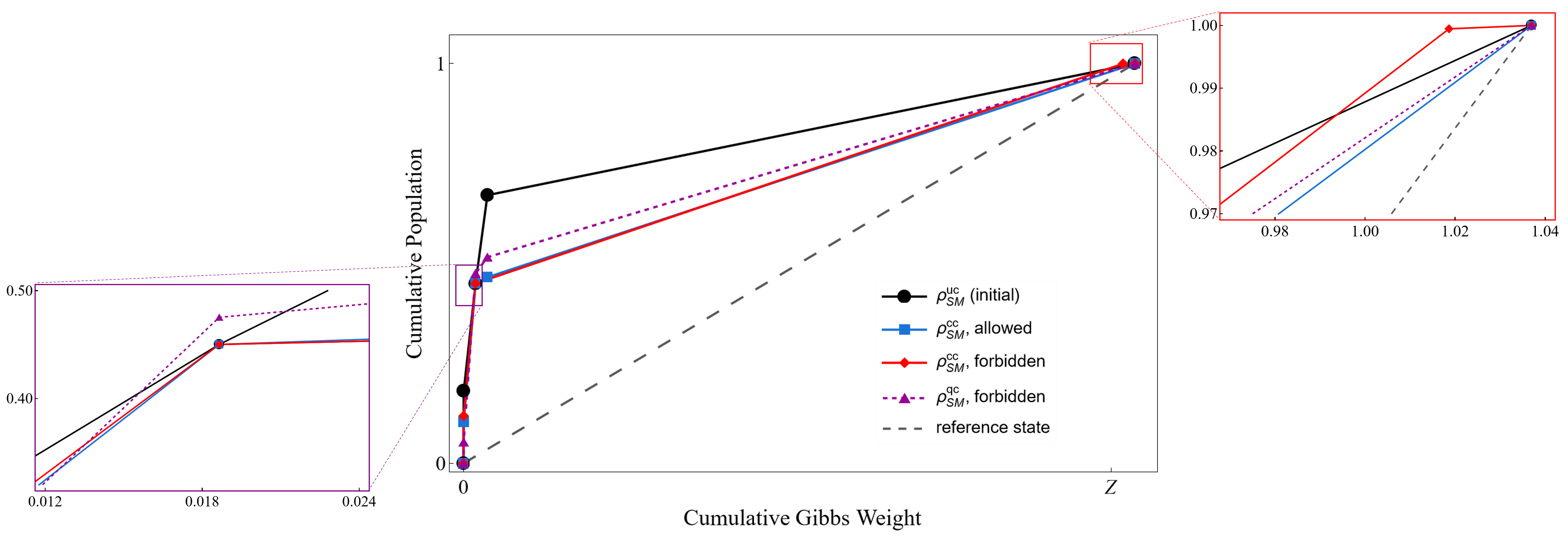}
    \caption{Thermo-majorization curves for the common initial product state $\rho_\mtn{SM}^\mtn{\mathrm{uc}}=\rho_\mtn{S}^\mtn{\beta_2}\otimes\rho_\mtn{M}^\mtn{\beta_1}$ (black), the allowed classically correlated final state $\rho_\mtn{SM}^\mtn{\mathrm{cc}}$ with $\chi= 5 \times 10^{-2}$ (blue), the forbidden classically correlated final state $\rho_\mtn{SM}^\mtn{\mathrm{cc}}$ with $\chi=6.5 \times 10^{-2}$ (red), the forbidden discordant final state $\rho_\mtn{SM}^\mtn{\mathrm{qc}}$ with $\lambda=9.47 \times 10^{-2}$ (purple dashed), and the common thermal reference state (gray dashed), for $E_g=0$, $E_e=2$, $\beta_1=0.1$, $\beta_2=0.2$, $\beta_3=1$, and $\beta_b=2$. In all three final-state examples, the initial state, the thermal reference state, and the final marginals are identical; only the correlation structure of the final joint state is changed. In particular, $\rho_\mtn{SM}^\mtn{\mathrm{qc}}(\lambda=9.47 \times 10^{-2})$ has the same reduced marginals and essentially the same mutual information as the allowed $\rho_\mtn{SM}^\mtn{\mathrm{cc}}(\chi=5 \times 10^{-2})$ state, yet thermo-majorization yields a different accessibility verdict, showing that reduced-state data and a scalar correlation measure are generally insufficient in finite-dimensional correlated catalysis.}
    \label{fig:correlated-catalysis-thermomajorization}
\end{figure*}

Correlated catalytic thermal transformations constitute a distinct regime in which the catalyst is returned only marginally, and is allowed to become correlated with the system during the transition. In the asymptotic regime of negligibly small correlations, the infinite family of generalized second laws collapses to the Helmholtz free-energy inequality~\cite{2015_PRL_Lostaglio_Muller, 2017_Entropy_CorrelatedCatal, PhysRevX.8.041051, 2022_PRL_CorrelatedCatal, 2025_PRL_CorrCatalAndFreeE, 2025_arXiv_CorrCatalAndFreeE, 2025_EPL_Shiraishi_CorrelatedCatalyzerAndUniqueF}. Here, by contrast, we focus on finite-dimensional regimes in which the generated system--catalyst correlations are not negligible and can materially affect the thermo-majorization verdict for the joint transition. The detailed constructions and calculations underlying this section are presented in Appendix~\ref{app:corrcatal}.

Our starting point is a fixed initial uncorrelated (uc) state,
\begin{equation}
\rho_\mtn{SM}^\mtn{\mathrm{uc}}=\rho_\mtn{S}^\mtn{\beta_2}\otimes \rho_\mtn{M}^\mtn{\beta_1},
\end{equation}
in which the subsystem $S$ is initially in thermal equilibrium at inverse temperature $\beta_2$ and the subsystem $M$ (the catalyst) at $\beta_1$, while the thermal operations are taken with respect to a bath at inverse temperature $\beta_b$.

Both $S$ and $M$ are two-level systems governed by the same local Hamiltonian,
\begin{equation}\label{eq:qubit-hamiltonian}
H_\mtn{S} = H_\mtn{M} = E_g |g\rangle \langle g| + E_e |e\rangle \langle e | ,
\end{equation}
with ground and excited energies $E_g$ and $E_e$, respectively. The initial state comes

together with fixed final marginals
\begin{equation}
\rho_\mtn{M}'=\rho_\mtn{M}^\mtn{\beta_1},
\qquad
\rho_\mtn{S}'=\rho_\mtn{S}^\mtn{\beta_3},
\end{equation}
so that the catalyst is locally returned to its initial thermal state while the system is transformed to a different local thermal state. Here $p_\mtn{\beta} \coloneqq e^{-E_g\beta}/(e^{-E_e\beta}+e^{-E_g\beta})$ denotes the ground-state population of a local thermal state at inverse temperature $\beta$, so that $\beta_1$, $\beta_2$, and $\beta_3$ characterize the chosen local marginals, and $\beta_b$ the bath. The numerical values used in all examples below are
\begin{widetext}
\begin{equation}
E_g = 0,
\qquad
E_e = 2,
\qquad
\beta_1 = 0.1,
\qquad
\beta_2 = 0.2,
\qquad
\beta_3 = 1,
\qquad
\beta_b = 2.
\label{eq:numerical-parameters}
\end{equation}
\end{widetext}
Correlations are quantified by the mutual information, by which we mean throughout the von Neumann mutual information, computed with the natural logarithm. The key point is that, throughout the examples, the initial state, the final marginals, and the thermal reference state used in the thermo-majorization test are all held fixed. Only the \textit{correlation structure} of the final joint state is varied, as made explicit in Eqs.~\eqref{eq:rhoC} and~\eqref{eq:rhoD} below. Since reduced-state generalized free-energy functionals based on R\'enyi (additive) and non-additive divergences are functions of the marginals alone, they depend only on the fixed local transition $\rho_\mtn{S}^\mtn{\beta_2}\mapsto\rho_\mtn{S}^\mtn{\beta_3}$ and take identical values for all final states considered here; see Eq.~\eqref{eq:renyi-system-values} in Appendix~\ref{app:corrcatal}. Any family of second laws formulated solely in terms of such functionals must therefore assign the same verdict to every one of these final states. However, the thermo-majorization analysis of the \textit{joint} transitions shows that reduced-state information is not sufficient to decide accessibility.

Our first example is a family of purely classically correlated (cc) final states $\rho_\mtn{SM}^\mtn{\mathrm{cc}}$, parameterized by $\chi$. In the ordered product basis $\{|gg\rangle,|ge\rangle,|eg\rangle,|ee\rangle\}$, with $S$ the first tensor factor, it reads
\begin{widetext}
\begin{equation}
\rho_\mtn{SM}^\mtn{\mathrm{cc}}(\chi) =
\begin{pmatrix}
p_\mtn{\beta_3} p_\mtn{\beta_1} + \chi & 0 & 0 & 0\\
0 & p_\mtn{\beta_3}(1-p_\mtn{\beta_1})-\chi & 0 & 0\\
0 & 0 & (1-p_\mtn{\beta_3})p_\mtn{\beta_1}-\chi & 0\\
0 & 0 & 0 & (1-p_\mtn{\beta_3})(1-p_\mtn{\beta_1})+\chi
\end{pmatrix},
\label{eq:rhoC}
\end{equation}
\end{widetext}
where the admissible range of $\chi$ is fixed by positivity (Appendix~\ref{app:corrcatal}). The reduced states are independent of $\chi$, but the thermo-majorization verdict is not. For the parameter choice $\chi=5 \times 10^{-2}$, the mutual information is $I(S{:}M)=$$5.17 \times 10^{-2}$, and the transition $\rho_\mtn{SM}^\mtn{\mathrm{uc}}\mapsto\rho_\mtn{SM}^\mtn{\mathrm{cc}}$ is allowed (blue curve in Fig.~\ref{fig:correlated-catalysis-thermomajorization}, which never rises above the initial-state (black) curve). For $\chi=6.5 \times 10^{-2}$, the mutual information increases to $I(S{:}M)=$$10.14 \times 10^{-2}$, but the transition becomes forbidden (red curve in Fig.~\ref{fig:correlated-catalysis-thermomajorization}, which rises above the initial-state (black) curve). Thus, even with fixed marginals, changing only the final classical correlation can change thermodynamic accessibility.

Our second example sharpens this conclusion. We consider a discordant, i.e., quantum-correlated (qc), family of final states $\rho_\mtn{SM}^\mtn{\mathrm{qc}}$, parameterized by $\lambda$ and carrying coherence within the degenerate energy block spanned by $|ge\rangle$ and $|eg\rangle$,
\begin{widetext}
\begin{equation}
\rho_\mtn{SM}^\mtn{\mathrm{qc}}(\lambda) =
\begin{pmatrix}
p_\mtn{\beta_3} p_\mtn{\beta_1} & 0 & 0 & 0\\
0 & p_\mtn{\beta_3}(1-p_\mtn{\beta_1}) & \lambda & 0\\
0 & \lambda & (1-p_\mtn{\beta_3})p_\mtn{\beta_1} & 0\\
0 & 0 & 0 & (1-p_\mtn{\beta_3})(1-p_\mtn{\beta_1})
\end{pmatrix},
\label{eq:rhoD}
\end{equation}
\end{widetext}
in the same ordered basis, with the same marginals as $\rho_\mtn{SM}^\mtn{\mathrm{cc}}$. Choosing $\lambda=9.47 \times 10^{-2}$, one obtains
$I(S{:}M)=$$5.17 \times 10^{-2}$, which matches, to the quoted numerical precision, the mutual information of the allowed classically correlated state $\rho_\mtn{SM}^\mtn{\mathrm{cc}}$ at $\chi=5 \times 10^{-2}$. Nevertheless, the transition $\rho_\mtn{SM}^\mtn{\mathrm{uc}}\mapsto\rho_\mtn{SM}^\mtn{\mathrm{qc}}$ is forbidden by thermo-majorization (purple dashed curve in Fig.~\ref{fig:correlated-catalysis-thermomajorization}, which rises above the initial-state (black) curve even though the blue curve with the same mutual information does not). Hence, even supplementing reduced-state data with a scalar correlation measure such as the mutual information is not sufficient in general: thermodynamic accessibility can depend not only on the amount of correlation, but also on its nature.

Figure~\ref{fig:correlated-catalysis-thermomajorization} collects the three final-state curves together with those of the common initial and thermal reference states. Since only the final joint state is varied, the figure makes the main conclusion visually transparent: the three final states shown all share the same reduced marginals---the allowed classically correlated one, the forbidden more strongly correlated one, and the forbidden discordant one, the latter carrying essentially the same mutual information as the allowed one.

Taken together, these examples show that, in finite-dimensional correlated catalysis, no family of second laws defined solely in terms of reduced-state monotones can provide a complete characterization of state accessibility. 
The obstruction resides in the joint state itself, not in its marginals. More strongly, the examples also show that even the total amount of correlation is not, by itself, sufficient: the detailed structure of the final system--catalyst correlations can be thermodynamically decisive. The full constructions, parameter choices, and explicit matrices are given in Appendix~\ref{app:corrcatal}.

\section{Discussion}
\label{sec::discuss}
The generalized second laws for catalytic thermal operations were established by Brand\~ao et al.~\cite{BrandaoEtAl2015} in terms of R\'enyi divergences and subsequently refined in~\cite{Gour2021, Farooq2024, Verhagen2025}. The limits of approximate catalysis, including explicit bounds on the trace-distance return error as a function of catalyst dimension and energy, were analyzed by Ng et al.~\cite{2017_NJP_Ng_LimitsOfCatal}, while recent finite-size results have made correlated catalytic constructions substantially more operational by relating catalytic feasibility to multicopy transformation rates, deriving sufficient finite-dimensional catalyst bounds, and showing that the required catalyst size can depend sharply on the particular transition through phenomena such as catalytic resonance~\cite{2025_PRA_FiniteSizeCatal}. A comprehensive classification of catalytic settings and their interrelations was given in the recent review by Lipka-Bartosik, Wilming, and Ng~\cite{2024_RMP_Review}. Against this background, the present work isolates two distinct but complementary finite-size issues that are not equally transparent in the additive formulation.

For uncorrelated catalytic transformations, the R\'enyi formulation's additivity implies the exact cancellation of the catalyst contribution in the generalized free-energy balance, so that the resulting system-level inequalities certify the existence of a suitable catalyst without making its thermodynamic role explicit. The non-additive formulation developed here does not alter the underlying transition conditions: the work distance is invariant (Appendix~\ref{appG}), and the same state conversions are certified. What changes is the bookkeeping. The pseudo-additive correction retains the catalyst in the generalized monotonicity condition, so that its athermality and spectral structure enter the balance relations directly. In particular, in approximate catalytic scenarios, this representation reveals that admissibility can depend not only on the trace-distance return error but also on how that error is distributed across the catalyst spectrum (Appendix~\ref{app:pseudo-additivity-finite-size}), a sensitivity that is also present in catalyst-level R\'enyi quantities but is invisible in the system-level R\'enyi balance.

For correlated catalytic transformations in finite-dimensional regimes where the generated system--catalyst correlations are not negligible, the recent finite-size framework of Ref.~\cite{2025_PRA_FiniteSizeCatal} provides an important complementary backdrop:
it shows that correlated-catalytic feasibility can be made finite-dimensional and quantitatively meaningful, and that the required catalyst size can depend sharply on the specific transformation, including through catalytic resonance. Our focus is different: not whether, or at what catalyst size, such transformations are feasible, but which data carry the resulting thermodynamic constraints. Through explicit thermo-majorization examples, we show that once such finite-size correlated catalysis is admitted, the reduced state of the system is generally not a sufficient carrier of the relevant thermodynamic constraints. More precisely, there exist families of joint transformations for which the system and the catalyst marginals are fixed, while thermo-majorization accessibility changes as the final system--catalyst correlation structure is varied. Consequently, no family of second laws formulated solely in terms of reduced-state functionals can provide a complete characterization of correlated catalytic thermal transformations in such regimes.

We further show that even the total amount of correlation need not be sufficient. In particular, we compare final correlated states with identical marginals and the same mutual information, but different internal correlation structure---for example, one purely classically correlated~\cite{Vedral-2001, modi2010unified, modi2012classical} and another discordant~\cite{Zurek-2002, modi2010unified, bera2017quantum}---and find that thermo-majorization can assign different accessibility verdicts to them. Thus, in the correlated catalytic regime, thermodynamic accessibility can depend not only on the amount of correlation, but also on its nature. The relevant thermodynamic information is therefore not, in general, compressible into system-only state functions, nor even into system-only data supplemented by a scalar correlation measure.

The two regimes thus depart from the additive picture in different ways. Correlated catalytic thermodynamics is therefore not merely a perturbative extension of uncorrelated catalysis; rather, it is a regime in which the thermodynamic accounting becomes intrinsically joint-state dependent. This is consistent with the observation that the work cost of quantum processes depends on the full system--environment correlations rather than on reduced-state data alone~\cite{Faist2015,FaistRenner2018}.

\paragraph*{Toward a compositional perspective.}
The contrast between the two catalytic regimes also admits a compositional reading: although the limitation manifests differently in each, both cases trace the shortcomings of additive generalized second-law formulations to how thermodynamic constraints behave under composition. It leaves a visible pseudo-additive residue even for product states; in the correlated regime, it forces the relevant constraints to depend on genuinely joint-state structure. A fuller synthesis of this idea, in which pseudo-additivity and correlation structure are treated as complementary manifestations of a common compositional principle in finite-size catalytic thermodynamics, is currently under development and will be explored in a companion work~\cite{CompanionQ2Laws}.

\section{Conclusion}
\label{sec::conc}

We have developed a complementary formulation of the generalized second laws of quantum thermodynamics based on non-additive divergences. The corresponding non-additive free energies are obtained from these divergences by a single positive affine transformation, and therefore carry the same monotonicity structure. The pseudo-additive composition law of the non-additive divergence retains an explicit catalyst-dependent correction in the generalized monotonicity conditions, making the thermodynamic role of the catalyst structurally visible in the free-energy balance, while leaving the underlying transition preorder unchanged (Sec.~\ref{sec:helm-comparison}). 

In uncorrelated catalytic thermal transformations, the pseudo-additive balance makes explicit the finite-size feasibility constraints that depend not only on the return error and catalyst dimension, but also on how the catalytic deviation is distributed across the catalyst spectrum. For correlated catalytic thermal transformations, we have shown through explicit finite-dimensional examples that no family of second laws formulated solely in terms of reduced-state quantities can provide a complete characterization of thermodynamic accessibility. Even when the system and catalyst marginals are fixed, and even when the total amount of correlation is matched, different internal correlation structures can nonetheless yield different thermo-majorization verdicts. The relevant thermodynamic constraints are therefore intrinsically joint-state dependent. Technical details and extended derivations supporting these results are provided in the Appendices.

 A companion study~\cite{CompanionQ2Laws} extends the present framework to systems governed by deformed Gibbs equilibria, where the non-multiplicativity of the deformed-exponential renders the joint equilibrium state intrinsically correlated and the joint-state character of the second laws becomes structurally unavoidable.

\bigskip
\begin{acknowledgments}
We are grateful to M. Horodecki, G. Gour, M. Tomamichel, and M. P. Müller for fruitful correspondence. We also thank the two anonymous referees, whose careful reports have substantially improved this work.

The authors used Claude (Anthropic; accessed 2026) during revision, under the authors' direction, for language editing, LaTeX preparation, numerical cross-checks, and verification of statements against cited sources; all suggestions were reviewed, verified, and approved by the authors, who take full responsibility for the content.
\end{acknowledgments}

\bigskip

\appendix
\noindent\textit{Notation.} Throughout the appendices, thermal states are denoted by $\rho^\mtn{\beta}$ with the relevant subsystem label; in the main text the equivalent objects are written as $\gamma$ (e.g., $\gamma_\mtn{S}=\rho^\mtn{\beta}_\mtn{S}$). In Appendices~\ref{appB}--\ref{appD}, $M$ and $M'$ (with or without subscripts) denote positive integers arising as denominators of rational approximations; they bear no relation to the thermal machine $M$ of the main text.

\renewcommand{\theequation}{A\arabic{equation}}
\newcommand{\appsection}[1]{%
  \refstepcounter{section}
  \section*{Appendix \Alph{section}: #1}%
}

\setcounter{equation}{0}  
\appsection{Non-additive Entropy and Divergence}\label{appA}

In this Appendix, we provide the definition and important properties of the non-additive entropy, its associated divergence and their properties.

Consider probability distributions $p=\{p_i\}_{i=1}^n$, $q=\{q_i\}_{i=1}^n$. The non-additive entropy of order $\alpha\in\mathbb R\setminus\{0,1\}$~\cite{Tsallis1988} is defined as
\begin{equation}\label{eq:TsallisEntroptApp}
    H_\mtn{\alpha}(p)\coloneqq\frac{\sgn\alpha}{1-\alpha}\left( \sum_{i=1}^n p_i^{\alpha} -1\right),
\end{equation}
where 
\begin{equation}\label{eq:sgnApp}
\sgn\alpha\coloneqq
  \begin{cases}
      +1 & \alpha\geq 0 \\
      -1 & \alpha < 0 
   \end{cases}    \,.
\end{equation}

The following relation is important
\begin{equation}\label{eq:RenyiHTsallisH}
    \ln\left[\sgn{\alpha}(1-\alpha)H_\mtn{\alpha}(p)+1\right]=\sgn{\alpha}(1-\alpha)H_\mtn{\alpha}^\mtn{R}(p).
\end{equation}\\
Here $H_\mtn{\alpha}^\mtn{R}$ is the R\'enyi entropy of order $\alpha\in\mathbb R\setminus\{0,1\}$~\cite{Renyi1961} and is defined as
\begin{equation}\label{eq:RenyisEntroptApp}
    H_\mtn{\alpha}^\mtn{R}(p)\coloneqq\frac{\sgn\alpha}{1-\alpha}\ln\left( \sum_{i=1}^n p_i^{\alpha}\right),
\end{equation}
where $\ln(\cdot)$ is the natural logarithm function.

The non-additive divergence of order $\alpha\in\mathbb{R}$~\cite{BorlandPlastinoTsallis1998, BorlandPlastinoTsallis1999E} reads
\begin{eqnarray}\label{eq:TsallisDApp2}
    D_\mtn{\alpha}(p\|q) \coloneqq -\sgn{\alpha}\sum_{i=1}^{n}p_i\ln_\mtn{\alpha}\left(\frac{q_i}{p_i}\right),
\end{eqnarray}
where $\alpha-$logarithmic function~\cite{UmarovTsallisSteinberg2008} is 
\begin{equation}\label{eq:LnAlphaApp}
    \ln_\mtn{\alpha}(x)\coloneqq\frac{x^{1-\alpha}-1}{1-\alpha}.  
\end{equation}

The following relations hold
\begin{subequations}\label{eq:RenyiDTsallisD}
\begin{eqnarray}\label{eq:RenyiDTsallisD1}
    \exp_{2-\alpha}\left[D_\mtn{\alpha}(p\|q)\right]&=&\exp\left[D_\mtn{\alpha}^\mtn{R}(p\|q)\right],\\
\label{eq:RenyiDTsallisD2}    \ln\left[(\alpha-1)D_\mtn{\alpha}(p\|q)+1\right]&=&(\alpha-1)D_\mtn{\alpha}^\mtn{R}(p\|q),
\end{eqnarray}
\end{subequations}
where $\alpha-$exponential function~\cite{UmarovTsallisSteinberg2008} is
\begin{equation}\label{eq:ExpAlphaApp}
    \exp_\mtn{\alpha}{(x)}\coloneqq\left[1+(1-\alpha)x\right]^\frac{1}{(1-\alpha)},
\end{equation}
and $D_\mtn{\alpha}^\mtn{R}$ denotes the R\'enyi divergence of order  $\alpha\in\mathbb{R}\setminus\{1\}$, whenever the expression 
\begin{equation}\label{eq:RenyiDApp}
 D_\mtn{\alpha}^\mtn{R}(p\|q)\coloneqq\frac{\sgn\alpha}{\alpha-1}\ln\left(\sum_{i=1}^n p_i^{\alpha} q_i^{1-\alpha}\right)
\end{equation}
is well-defined~\cite{Renyi1961}. The relations in Eqs.~(\ref{eq:RenyiDTsallisD}) hold for all $\alpha \neq 1$ whenever the quantity $\sum_i p_i^\alpha q_i^{1-\alpha}$ is finite and strictly positive, so that both $D_\alpha(p\|q)$ and $D_\alpha^{\mtn{R}}(p\|q)$ are well-defined. For $\alpha>1$, this requires the support condition $\mathrm{supp}(p)\subseteq\mathrm{supp}(q)$. For $\alpha<0$, these expressions are well-defined under the additional assumption that $p_i,q_i>0$ for all $i$.

We now list some limiting cases of the non-additive entropy and divergence:

$\alpha=0$:
\begin{subequations}
\begin{empheq}{alignat=2}      
\label{eq:Div_0TApp} D_\mtn{0}(p\|q)&=\lim_{\alpha\rightarrow0^+}D_\mtn{\alpha}(p\|q)&&=-\left(\sum_{i:p_i\neq0}q_i-1\right)\\
\label{eq:Ent_0TApp} H_\mtn{0}(p)&=\lim_{\alpha\rightarrow0^+}H_\mtn{\alpha}(p)&&=\:\text{rank}(p)-1.
\end{empheq}
\end{subequations}

$\alpha=1$:
\begin{subequations}
\begin{empheq}{alignat=2}
\label{eq:Div_1App}    D_\mtn{1}(p\|q)&=-\sum_{i=1}^n p_i\ln{\frac{q_i}{p_i}}&&\coloneqq D^\mtn{KL}(p\|q) \\
\label{eq:Ent_1App}    H_\mtn{1}(p)&=-\sum_{i=1}^n p_i\ln{p_i}&&\coloneqq H(p)
\end{empheq}
\end{subequations}
where $D^\mtn{KL}(p\|q)$ is the Kullback-Leibler divergence~\cite{KullbackLeibler1951}, and $H(p)$ is the Shannon entropy~\cite{Shannon1948a,Shannon1948b,Shannon1949Book}.

$\alpha\rightarrow+\infty$:
\begin{subequations}
\begin{empheq}{alignat=2}
\label{eq:Div_inftyApp}    D_\mtn{+\infty}(p\|q)&=\lim_{\alpha\rightarrow +\infty}D_\mtn{\alpha}(p\|q)&&=+\ln\left(\max_i\frac{p_i}{q_i}\right)\\ 
\label{eq:Ent_inftyTApp}    H_\mtn{+\infty}(p)&=\lim_{\alpha\rightarrow +\infty}H_\mtn{\alpha}(p)&&=0.
\end{empheq}
\end{subequations}

$\alpha\rightarrow-\infty$:
\begin{subequations}
\begin{empheq}{alignat=2}
\label{eq:Div_-inftyRApp}    D_\mtn{-\infty}(p\|q)&=\lim_{\alpha\rightarrow -\infty}D_\mtn{\alpha}(p\|q)&&=+\ln\left(\max_i\frac{q_i}{p_i}\right) \\ 
\label{eq:Ent_-inftyRApp}    H_\mtn{-\infty}(p)&=\lim_{\alpha\rightarrow -\infty}H_\mtn{\alpha}(p)&&=+\ln\left(\min_i p_i\right).
\end{empheq}
\end{subequations}
Another useful relation reads~\cite{BrandaoEtAl2015}:
\begin{widetext}
\begin{equation} \label{eq:AUsefulRelnApp}
    \alpha\,\sgn{\alpha}D_\mtn{1-\alpha}(p\|q)=(1-\alpha)\sgn{1-\alpha}D_\mtn{\alpha}(q\|p),\quad \forall\alpha\in\mathbb{R}\setminus\{0,1\}.
\end{equation}
\end{widetext}
which can be obtained by direct substitution.\\
Note also that 
    \begin{equation}\label{eq:PositiveDApp}
        D_\mtn{\alpha}(p\|q)\ge0,\qquad\forall\alpha\in\mathbb{R}.
    \end{equation}
The non-additive divergence $D_\mtn{\alpha}$ satisfies the data processing inequality
\begin{equation}\label{eq:DataPreocessApp}
    D_\mtn{\alpha}(p\|q)\geq D_\mtn{\alpha}(\Lambda(p)\|\Lambda(q)) ,\qquad\forall\alpha\in\mathbb{R},
\end{equation}
for a column-stochastic mapping $\Lambda$, $\sum_{j=1}^n\Lambda_\mtn{ji}=1$,
\begin{equation}\label{eq:ChannelApp}
    \Lambda:p\longmapsto\Lambda(p).
\end{equation}

\noindent In the range $\alpha\ge0$~\cite{FuruichiYanagiKuriyama2004}, we have 

\begin{widetext}
\begin{subequations}
\begin{align}
\label{eq:MonotoneDiv1}D_\mtn{\alpha}(\Lambda(p)\|\Lambda(q))
&\equiv-\sum_{j=1}^{n}(\Lambda(p))_j\ln_\mtn{\alpha}\frac{(\Lambda(q))_j}{(\Lambda(p))_j}\\
\label{eq:MonotoneDiv2}
&=-\sum_{j=1}^{n}\left(\sum_{i=1}^{n}\Lambda_\mtn{ji} p_i\right)\left(\ln_\mtn{\alpha}\frac{\sum_{k=1}^{n}\Lambda_\mtn{jk} q_k}{\sum_{m=1}^{n}\Lambda_\mtn{jm} p_m}\right)\\
\label{eq:MonotoneDiv3}
&\le \sum_{j=1}^{n}\left[\sum_{i=1}^{n}\Lambda_\mtn{ji} p_i\left(\ln_\mtn{\alpha}\frac{q_i}{p_i}\right)\right] = \sum_{i=1}^{n}\left[ p_i\left(\ln_\mtn{\alpha}\frac{q_i}{p_i}\right)\right]\equiv D_\mtn{\alpha}(p\|q).
\end{align}
\end{subequations}
\end{widetext}

\noindent Note that the inequality \eqref{eq:MonotoneDiv3} follows from the generalized log-sum inequality~\cite{BorlandPlastinoTsallis1998,BorlandPlastinoTsallis1999E}:
\begin{equation}\label{eq:GenLogSumApp}
    -\left(\sum_{i=1}^{n}p_i\right)\ln_\mtn{\alpha}\left(\frac{\sum_{j=1}^{n}q_j}{\sum_{k=1}^{n}p_k}\right)
\leq 
-\sum_{i=1}^{n}p_i\ln_\mtn{\alpha}\left(\frac{q_i}{p_i}\right).
\end{equation}

\noindent One can generalize the result above to the case $\alpha<0$ by denoting $\alpha'\equiv(1-\alpha)$ so that $\alpha'>1$:
\begin{widetext}
\begin{subequations}
    \begin{align}
\label{eq:ProofNegAlpha1}
(1-\alpha')\sgn{1-\alpha'}D_\mtn{\alpha'}(\Lambda(p)\|\Lambda(q))&\le(1-\alpha')\sgn{1-\alpha'}D_\mtn{\alpha'}(p\|q)\\
\label{eq:ProofNegAlpha2}
\alpha'\sgn{\alpha'}D_\mtn{(1-\alpha')}(\Lambda(p)\|\Lambda(q))&\le\alpha'\sgn{\alpha'}D_\mtn{(1-\alpha')}(p\|q)\\ 
\label{eq:ProofNegAlpha3}
D_\mtn{(1-\alpha')}(\Lambda(p)\|\Lambda(q))&\le D_\mtn{(1-\alpha')}(p\|q)\\
\label{eq:ProofNegAlpha4}
D_\mtn{\alpha}(\Lambda(p)\|\Lambda(q))&\le D_\mtn{\alpha}(p\|q)\qquad\forall\alpha<0.
    \end{align}
\end{subequations}
\end{widetext}

\noindent Introducing the uniform probability distribution of dimensions $n$ as
    \begin{equation}\label{eq:UniformProbDistApp}
        \eta_n=\{\eta_\mtn{n,i}\}_{i=1}^n,\qquad \eta_\mtn{n,i}=1/n,\quad \forall i\in\{1,\cdots, n\},
    \end{equation}
one can easily show that the non-additive entropy and divergence are related to one another as

\begin{empheq}{alignat=2}
\label{eq:EntDivT}      
H_\mtn{\alpha}(p) &+ n^{1-\alpha} D_\mtn{\alpha}(p\|\eta_\mtn{n})&&=\frac{\sgn{\alpha}}{1-\alpha}\left( n^{1-\alpha}-1\right).
\end{empheq}

\renewcommand{\theequation}{B\arabic{equation}}
\setcounter{equation}{0}  
\appsection{Some Technical Heuristics} \label{appB}

We now restate, without proof, the following lemma, which sets the condition for a channel to be written as a direct sum of two channels that act separately on disjoint partitions of the total input/output space (see Lemma 16 in~\cite{BrandaoEtAl2015} for proof).\\

\begin{lemma} \label{lemma1} 
Consider a channel $\Lambda$ acting on the space of probability distributions -the probability simplex- $\mathcal{K}$, and $b\in\mathcal{K}$,
\begin{subequations}
    \begin{empheq}{alignat=3}
\label{LChannel1App}        \Lambda&:\mathcal{K}&&\rightarrow\mathcal{K}\\
\label{LChannel2App}               &:b&&\mapsto \Lambda(b).
    \end{empheq}
\end{subequations}

If there exist distributions $b^{0},b',b''\in\mathcal{K}$ of some fixed dimension $n=m+z$, such that for $b^0$ of full rank,
    \begin{equation}
\label{eq:Lemma16-1App} \Lambda(b^0)=b^0,
    \end{equation}
and for $b'=(b'_1,\cdots,b'_m,0,\cdots,0)$
    \begin{equation}
\label{eq:Lemma16-3App} \Lambda(b')\equiv b''=(b''_1,\cdots,b''_m,0,\cdots,0),      
    \end{equation}
then $\Lambda$ can be decomposed as
    \begin{equation}
\Lambda\equiv\Lambda^{(m)}\oplus\Lambda^{(z)},
    \end{equation}
where $\Lambda^{(m)}$ acts only on the first $m$ elements mapping them onto the same group of elements, while $\Lambda^{(z)}$ acts similarly on the remaining $z=n-m$ elements.
\end{lemma}
  
We now prove that any probability distribution with irrational values can be arbitrarily approximated by those with all rational components. Note that a similar lemma, i.e., Lemma 15 in~\cite{BrandaoEtAl2015}, shows that this is possible for any full-rank probability distribution.
  
\begin{lemma} \label{lemma2}   
     Let $\bbar=\{\bbar_i\}_{i=1}^n$ be an $n$-dimensional probability distribution with decreasing order having $z$ zero components and $m$ non-zero components, of which some are possibly irrational,
    \begin{eqnarray}\label{eq:SetBbar1App}
     \bbar_1\ge\cdots\ge\bbar_m>\bbar_{m+1}=\cdots=\bbar_{m+z}=0.
    \end{eqnarray}
Then, $\mathcal{K}$ being the probability simplex and $b\in\mathcal{K}$, there exists a valid channel $L$,
\begin{subequations}
    \begin{empheq}{alignat=3}
\label{eq:EChannel1App}      L&:\mathcal{K}&&\rightarrow\mathcal{K}\\
\label{eq:EChannel2App}         &:b&&\mapsto L(b),
    \end{empheq}
\end{subequations}
which is composed of channels $\nu$ and $\Lambda$,
\begin{equation}\label{eq:EChannel0App}      
    L\equiv\nu\circ\Lambda,
\end{equation}  
such that 

\begin{enumerate}[label=(\roman*)]
    \item $\Lambda(\bbar)=b'\equiv\{b'_i\}_{i=1}^n$,

\begin{widetext}
\begin{equation}\label{eq:EBbarApp}
    b'_1\ge\cdots\ge b'_m> b'_\mtn{m+1}=\cdots=b'_n=0,\qquad b'_i\in\mathbb{Q},\quad \forall i=1,2,\cdots,n,
\end{equation}
\end{widetext}
and for any chosen $M>0$ 
\begin{equation}\label{eq:StatDistance1App}
    \|\bbar-b'\|<\frac{m}{M},
\end{equation}
\item $\nu(b')=q\equiv\{q_i\}_{i=1}^n$,
\begin{equation}
    q_1\ge\cdots q_m\ge q_\mtn{m+1}\ge\cdots\ge q_n>0, \qquad q_i\in\mathbb{Q}, \quad \forall i=1,2,\cdots,n, 
\end{equation}
and for any chosen $M'>f'>0$  (here and throughout this appendix, $\|\cdot\|$ denotes the unnormalized $\ell_1$ distance between probability distributions; the return error of the main text uses the factor-$\tfrac12$ convention of Eq.~\eqref{eq:return-error-def})
\begin{equation}\label{eq:StatDistanceNext1App}
    \|b'-q\|<\frac{f'}{M'}.
\end{equation}
\end{enumerate}

\end{lemma}

\begin{proof}
Choose first $M\in\mathbb{Z}^+$. Let
\begin{widetext}
\begin{equation}\label{eq:LeastIntApp}
    b'_i\equiv
    \begin{cases}
        \frac{\lceil M\bbar_i\rceil}{M},&i=1,\cdots,(m-1)\\
        &\\
        1-{\displaystyle\sum_{k=1}^{m-1}}b'_i, &i=m\\
        &\\
        0,&i=(m+1),\cdots,n
    \end{cases},
\end{equation}
\end{widetext}
where $\lceil M\bbar_i\rceil$ is the least integer greater than or equal to $M\bbar_i$. Consider the following amplitudes:

\begin{empheq}[left =Pr(i\ra j){\equiv}\empheqlbrace\,]{equation}\label{eq:TransitionProbsApp}
    \begin{alignedat}{3}
        & i\in \{1,\cdots,(m-1)\} & &\Ra \:   
        \begin{cases}
        1,&j=i\\
        0,&j\ne i
    \end{cases}\\
    &&&\\
        & i=m & &\Ra \:     \begin{cases}
        \frac{b'_j-\bbar_j}{\bbar_m},&j\in \{1,\cdots,(m-1)\}\\
        &\\
        1-{\displaystyle\sum_{k=1}^{m-1}}\frac{b'_k-\bbar_k}{\bbar_m},&j=m\\
        &\\
        0,&j\in\{(m+1),\cdots,(m+z)\}
    \end{cases}\\
    &&&\\
        & i\in \{(m+1),\cdots,(m+z)\} & &\Ra \: 
            \begin{cases}
        1,&j=i\\
        0,&j\ne i
    \end{cases}\\
    \end{alignedat}.
\end{empheq}
Transition probabilities given by \eqref{eq:TransitionProbsApp} satisfying
\begin{equation}\label{eq:ValidChannelCndApp}
    \displaystyle\sum_{j=1}^{m+z}Pr(i\ra j)=1, \qquad\forall i=1,\cdots,m+z,
\end{equation}
defines a valid channel, $\Lambda\equiv\Lambda^{(m)}\oplus\Lambda^{(z)}$, so that
\begin{equation}\label{EbbarEqQApp}
    \Lambda(\bbar)\equiv b'=\{b'_i\}_{i=1}^{n}
\end{equation}
is an $n$-dimensional probability distribution with decreasing order having $z$ zero components and  $m=n-z$ non-zero components all rational.

With these settings, the proof of \eqref{eq:StatDistance1App} is immediate: 
\begin{widetext}
\begin{subequations}
\begin{align}\label{eq:StatDistance2App}
    \|\bbar-\Lambda(\bbar)\|&\equiv\|\bbar-b'\|\\
    &=\frac{1}{2}\left[\displaystyle\sum_{i:\bbar_i>b'_i}(\bbar_i-b'_i)+\displaystyle\sum_{i:b'_i>\bbar_i}(b'_i-\bbar_i)\right]\\
    &=\frac{1}{2}\left[\displaystyle\sum_{i:\bbar_i>b'_i}(\bbar_i-b'_i)+\cancel{\displaystyle\sum_{i:b'_i>\bbar_i}(b'_i-\bbar_i)}+\displaystyle\sum_{i:\bbar_i\ge b'_i}(\bbar_i-b'_i)-\cancel{\displaystyle\sum_{i:\bbar_i\ge b'_i}(\bbar_i-b'_i)}\right]\\
    &=\displaystyle\sum_{i:\bbar_i>b'_i}(\bbar_i-b'_i)\\
\label{eq:StatDistance4App}  &\equiv \bbar_m-b'_m\\
    &=\left[\bbar_m-Pr(m\ra m)\bbar_m\right]\\
    &=\frac{1}{M}\displaystyle\sum_{k=1}^{m-1}\left(\lceil Mb_k\rceil-Mb_k\right)\le\frac{m-1}{M}<\frac{m}{M},
\end{align}
\end{subequations}
\end{widetext}
where in \eqref{eq:StatDistance4App} $m^{th}$ components are
\begin{subequations}
\begin{eqnarray}\label{eq:MinBbarApp}
    \bbar_m&=&\min_i\bbar_i,\\
        b'_m&=&\min_ib'_i.
\end{eqnarray}
\end{subequations}
Therefore, $\|\bbar-\Lambda(\bbar)\|$ can be made arbitrarily small by choosing $M$. For $\Lambda$ characterized by \eqref{eq:TransitionProbsApp}, the statistical distance between a random distribution $b$ and $\Lambda(b)$ can easily be shown to be 
\begin{widetext}
\begin{equation}\label{eq:StatDistance3App}
    \|b-\Lambda(b)\|<\frac{m}{M}\frac{b_m}{\bbar_m}<\frac{m}{M}\frac{1}{\bbar_m}.
\end{equation}
\end{widetext}
Here, $m$ being the index of the smallest component of the probability vector $\bbar$, $b_m$ is the $m^{th}$ component of $b$. Being random, $b_m$ need not be the smallest component of $b$.

Consider now an $n$-dimensional distribution, $b'$, of rank $m$ in the form of \eqref{eq:EBbarApp}, 
    \[
    \begin{aligned}
    &b'_1\ge\cdots\ge b'_m> b'_\mtn{m+1}=\cdots=b'_\mtn{m+z}=0,\\
    &b'_i\in\mathbb{Q},\quad \forall i=1,2,\cdots,(m+z).
    \end{aligned}
    \]  
Perturbing $b'$, a full-rank distribution can be obtained in $z$ steps. At each step, one component, which is next to the last and least non-zero component, is perturbed to be the next last and least non-zero component, and then the greater ones perturbed accordingly to keep the normalization. The channel $\nu$ is assumed to be composed of channels $\nu_i, i=1,2,\cdots,z$,
    \begin{equation}\label{eq:PertChannelApp}
        \nu\equiv\nu_z\circ\nu_{z-1}\circ\cdots\circ\nu_1.
    \end{equation}

The procedure is as follows:

    \begin{enumerate}
        \item[Step 1.]
            Choose $f'_1,M'_1\in\mathbb{Z}^+$ such that $f'_1<M'_1$ and
    \begin{equation}\label{eq:Proclambda1ChooseApp}
        b'_m\ge\left(1+\frac{1}{m}\right)\frac{f'_1}{M'_1}.
    \end{equation}
    Set $\nu_1(b')$ as
    \begin{equation}\label{eq:Proclambda1App}
        \nu_1(b')=\Biggl(\left(b'_1-\frac{f'_1/M'_1}{m}\right),\cdots,\left(b'_m-\frac{f'_1/M'_1}{m}\right),\frac{f'_1}{M'_1},0,\cdots,0\Biggr).
    \end{equation}
  
        \item[Step 2.]
            Choose $f'_2,M'_2\in\mathbb{Z}^+$ such that $f'_2<M'_2$ and
    \begin{equation}\label{eq:Proclambda2ChooseApp}
        \frac{f'_1}{M'_1}\ge\left(1+\frac{1}{m+1}\right)\frac{f'_2}{M'_2}.
    \end{equation}
    Set $\nu_2\circ\nu_1(b')$ as
\begin{widetext}
    \begin{align}\label{eq:Proclambda2App}
        \nu_2(\nu_1(b'))=\Biggl(&\left(b'_1-\frac{f'_1/M'_1}{m}-\frac{f'_2/M'_2}{m+1}\right),\cdots,\left(b'_m-\frac{f'_1/M'_1}{m}-\frac{f'_2/M'_2}{m+1}\right),\nonumber\\
        &\left(\frac{f'_1}{M'_1}-\frac{f'_2/M'_2}{m+1}\right),\frac{f'_2}{M'_2},0,\cdots,0\Biggr).
    \end{align}
\end{widetext}
    \item[$\vdots\;\;$]
    \item[Step $z$.] We present the result obtained by repeating similar steps:\\
    $\nu(b')=q=\{q_i\}_{i=1}^{m+z}$, where\\ 
    \begin{equation}\label{eq:ProcZApp}
        q_i=\begin{cases}
            b'_i-\sum_{k=1}^z\frac{f'_k/M'_k}{m+k-1},&i=1,\cdots,m\\
            &\\
            \frac{f'_i}{M'_i}-\sum_{k=i+1}^z\frac{f'_k/M'_k}{m+k-1},&i=(m+1),\cdots,(m+z)
        \end{cases}.
    \end{equation}
    \end{enumerate}
Finally, the statistical distance between states $b'$ and $\nu(b')$ is calculated as
\begin{subequations}
\begin{eqnarray}
    \label{eq:StatDistlambd1App}    \|b'-\nu(b')\|&\equiv&\|b'-q\|\\
    \label{eq:StatDistlambd2App}                      &=&{\displaystyle\sum_{k=1}^z}\frac{f'_k}{M'_k}\\
    \label{eq:StatDistlambd3App}                      &\equiv&\frac{f'}{M'},
\end{eqnarray}
\end{subequations}
and can be made arbitrarily small by choosing $f'$ and $M'$.
\end{proof}

\renewcommand{\theequation}{C\arabic{equation}}
\setcounter{equation}{0}  
\appsection{Trumping and Non-additive Divergence} \label{appC}

In this section, after restating the trumping condition we are going to examine how it is related to non-additive divergences.

\begin{theorem}\label{thm1}
Consider $N$-dimensional probability distributions $\tilde{p},\tilde{p}'\in(\mathbb{R^+})^N$ that do not both contain components equal to zero, and $\tilde{p}\ne\tilde{p}'$. Then $\tilde{p}$ can be trumped into $\tilde{p}'$, that is, $\tilde{p}$ can be transformed to $\tilde{p}'$ by means of catalytic thermal operations if and only if
    \begin{equation}\label{eq:TrumpingCondApp}
        f_\mtn{\alpha}(\tilde{p})>f_\mtn{\alpha}(\tilde{p}'),\qquad\forall\alpha\in\mathbb{R},
    \end{equation}
    where the functions $f_\mtn{\alpha}$ for probability vectors $x\in\mathbb{(R^+)}^N$ are given as
    \begin{equation}\label{eq:TrumpingFncsApp}
        f_\mtn{\alpha}(x)=
        \begin{cases}
            \ln\displaystyle\sum_{i=1}^N x_i^{\alpha},& 1<\alpha<+\infty\\
            \displaystyle\sum_{i=1}^N x_i\ln x_i,& 1=\alpha\\
            -\ln\displaystyle\sum_{i=1}^N x_i^{\alpha},& 0<\alpha<1\\
            -\displaystyle\sum_{i=1}^N \ln x_i,& 0=\alpha\\
            \ln\displaystyle\sum_{i=1}^N x_i^{\alpha},& -\infty<\alpha<0
        \end{cases}.
    \end{equation}
\end{theorem}
The proof of the above theorem can be found in~\cite{Klimesh2004,Turgut2007,Klimesh2007}.

\begin{proposition}\label{prop1}   
For $N$ dimensional probability distributions $\tilde{p},\tilde{p}'\in(\mathbb{R^+})^N$, the following conditions are equivalent: 
    \begin{enumerate}[label=(\roman*)]
    \label{Prop4_Cnd1} \item For arbitrary $\epsilon>0$, there exists an $N$-dimensional probability distribution $\tilde{p}'_{\epsilon}$, such that $\|\tilde{p}'-\tilde{p}'_{\epsilon}\|\le\epsilon$, and $\tilde{p}$ can be trumped into $\tilde{p}'_{\epsilon}$~\cite{Klimesh2004,Turgut2007,Klimesh2007}, that is, for all $\alpha\in(-\infty,+\infty)$,
    \begin{equation}\label{eq:TrumpingCondeApp}
        f_\mtn{\alpha}(\tilde{p})>f_\mtn{\alpha}(\tilde{p}'_{\epsilon}).
    \end{equation}
        \item $\eta_N$ being the uniform probability distribution \eqref{eq:UniformProbDistApp} of dimensions $N$, the following inequality holds for $\alpha\in \mathbb{R}$
    \begin{equation}\label{eq:TrumpingCndDivApp}
        D_\mtn{\alpha}(\tilde{p}\|\eta_{N})\geq D_\mtn{\alpha}(\tilde{p}'\|\eta_{N}).
    \end{equation}
    \end{enumerate}
\end{proposition}

\begin{proof}
Note first that if $\tilde{p}'_{\epsilon}$ is not of full rank, perturbing it with a small amount of noise, one can always obtain a full-rank distribution as proved in Lemma~\ref{lemma2}, since adding a small amount of noise is a valid thermal operation~\cite{BrandaoEtAl2015}. Hence, without loss of generality, $\tilde{p}'_{\epsilon}$ is assumed to be of full rank.

Now, let us check how $f_\mtn{\alpha}$ and $D_\mtn{\alpha}$ are related:
  
\begin{enumerate}
        \item $1<\alpha<+\infty$:
\begin{widetext}
        \begin{eqnarray}
\label{eq:Prop4_1a}   f_\mtn{\alpha}(\tilde{p})>f_\mtn{\alpha}(\tilde{p}'_{\epsilon})\equiv \ln\sum_{i=1}^N(\tilde{p}_i)^{\alpha}>\ln\sum_{i=1}^N(\tilde{p}'_{\epsilon,i})^{\alpha} &\Longleftrightarrow& \sum_{i=1}^N(\tilde{p}_i)^{\alpha}>\sum_{i=1}^N(\tilde{p}'_{\epsilon, i})^{\alpha}\\
\label{eq:Prop4_1b} &\Longleftrightarrow& D_\mtn{\alpha}(\tilde{p}\|\eta_N)> D_\mtn{\alpha}(\tilde{p}'_{\epsilon}\|\eta_N)
\end{eqnarray}
\end{widetext}
      
        \item $\alpha=1$:
\begin{widetext}
        \begin{equation}\label{eq:Prop4_2}
f_{1}(\tilde{p})>f_{1}(\tilde{p}'_{\epsilon})\equiv \sum_{i=1}^N\tilde{p}_i\ln \tilde{p}_i>\sum_{i=1}^N\tilde{p}_{\epsilon, i}\ln \tilde{p}_{\epsilon, i}\Longleftrightarrow D_{1}(\tilde{p}\|\eta_N)> D_{1}(\tilde{p}'_{\epsilon}\|\eta_N)
        \end{equation}
\end{widetext}
      
        \item $0<\alpha<1$:\
\begin{widetext}
        \begin{eqnarray}
\label{eq:Prop4_3a} f_\mtn{\alpha}(\tilde{p})>f_\mtn{\alpha}(\tilde{p}'_{\epsilon})\equiv     -\ln\sum_{i=1}^N(\tilde{p}_i)^{\alpha} > -\ln\sum_{i=1}^N (\tilde{p}_{\epsilon, i})^{\alpha} &\Longleftrightarrow&
\frac{1}{\alpha -1} \sum_{i=1}^N(\tilde{p}_i)^{\alpha} > \frac{1}{\alpha -1}\sum_{i=1}^N (\tilde{p}_{\epsilon, i})^{\alpha}\\
\label{eq:Prop4_3b} &\Longleftrightarrow& D_\mtn{\alpha}(\tilde{p}\|\eta_N)> D_\mtn{\alpha}(\tilde{p}'_{\epsilon}\|\eta_N)
        \end{eqnarray}
\end{widetext}
      
        \item $\alpha=0$:\\ 
      
        For $\tilde{p}_{\epsilon}$ being of full rank, we have
\begin{widetext}
        \begin{align}
\label{eq:Prop4_4a}   f_0(\tilde{p})>f_0(\tilde{p}'_{\epsilon}) &\equiv  &&\lim_{\alpha\rightarrow0^{+}}f_\mtn{\alpha}(\tilde{p})>\lim_{\alpha\rightarrow0^{+}}f_\mtn{\alpha}(\tilde{p}'_{\epsilon})\\
\label{eq:Prop4_4b} &\equiv &&\lim_{\alpha\rightarrow0^{+}}\left(-\log\sum_{i=1}^N(\tilde{p}_i)^{\alpha} \right)>\lim_{\alpha\rightarrow0^{+}}\left(-\log\sum_{i=1}^N(\tilde{p}_{\epsilon, i})^{\alpha} \right)\\
\label{eq:Prop4_4c} &\Leftrightarrow &&-\text{rank} (\tilde{p}) > -\text{rank} (\tilde{p}_{\epsilon})\\
\label{eq:Prop4_4d} &\equiv &&-H_0(\tilde{p}) >-H_0(\tilde{p}_{\epsilon})\\
\label{eq:Prop4_4e} &\equiv &&D_0(\tilde{p}\|\eta_N) >D_0(\tilde{p}_{\epsilon}\|\eta_N).
        \end{align}
\end{widetext}

        \item $-\infty<\alpha<0$:
\begin{widetext}
        \begin{eqnarray}
\label{eq:Prop4_5a}   f_\mtn{\alpha}(\tilde{p})>f_\mtn{\alpha}(\tilde{p}'_{\epsilon})\equiv \ln\sum_{i=1}^N(\tilde{p}_i)^{\alpha}>\ln\sum_{i=1}^N(\tilde{p}'_{\epsilon, i})^{\alpha} &\Longleftrightarrow& \sum_{i=1}^N(\tilde{p}_i)^{\alpha}>\sum_{i=1}^N(\tilde{p}'_{\epsilon, i})^{\alpha}\\
\label{eq:Prop4_5b} &\Longleftrightarrow& D_\mtn{\alpha}(\tilde{p}\|\eta_N)> D_\mtn{\alpha}(\tilde{p}'_{\epsilon}\|\eta_N)
        \end{eqnarray}
\end{widetext}
    \end{enumerate}

For arbitrary $\epsilon>0$, given $\|\tilde{p}'-\tilde{p}'_{\epsilon }\|\leq\epsilon$, we can therefore write
\begin{eqnarray}\label{eq:Prop4LastIneqApp}
    f_\mtn{\alpha}(\tilde{p})>f_\mtn{\alpha}(\tilde{p}'_{\epsilon})\Longleftrightarrow D_\mtn{\alpha}(\tilde{p}\|\eta_N)\geq D_\mtn{\alpha}(\tilde{p}'\|\eta_N),\qquad \forall\alpha\in\mathbb{R}. 
\end{eqnarray}
\end{proof}

\begin{definition}\label{it:EmbeddingChn}
Let $p=\{p_i\}_{i=1}^n$, $q=\{q_i\}_{i=1}^n$ be two probability distributions, such that 
    \begin{equation}\label{eq:EmbeddedDistrApp}
        q_i=\frac{d_i}{N},\qquad  d_i\in\mathbb{Z}^+\qquad \forall i\in\{1,\cdots,n\}. 
    \end{equation}
An \textit{embedding channel} $\Gamma$ is a mapping from the space of probability distributions to itself, 
\begin{empheq}{alignat=4}
\label{eq:EmbeddChn_1App} \Gamma &: &&\mathcal{K}&&\rightarrow &&\mathcal{K}\\
\label{eq:EmbeddChn_2App}&:(p&&;q)&&\mapsto\Gamma(&&p;q),
\end{empheq}
and is defined in~\cite{BrandaoEtAl2015} as the fine-grained $N$-dimensional probability distribution 
    \begin{equation}\label{eq:EmbeddChnApp}
        \Gamma(p;q)\coloneqq(\tilde{p};q)
    \end{equation}
with
    \begin{eqnarray} 
\label{eq:EmbeddChnPApp}
    \tilde{p}&=&\Biggl(\underbrace{\overbrace{\frac{p_1}{d_1},\cdots,\frac{p_1}{d_1}}^{d_1},\overbrace{\frac{p_2}{d_2},\cdots,\frac{p_2}{d_2}}^{d_2},\cdots\cdots,\overbrace{\frac{p_n}{d_n},\cdots,\frac{p_n}{d_n}}^{d_n}}\Biggr),\\
\label{eq:SelfFineGrainedApp}
    \Gamma(q;q)=\eta_N=\tilde{q}&=&\Biggl(\overbrace{\frac{1}{N},\cdots\cdots\cdots\cdots\cdots\cdots\cdots\cdots\cdots\cdots\cdots\cdots\cdots\cdots\cdots,\frac{1}{N}}^{N}\Biggr).
    \end{eqnarray}  
The inverse mapping is also defined and denoted as $\Gamma^*$, where
\begin{empheq}{alignat=10}
\label{eq:InverseAndEmbeddingApp}
    &\Gamma^*&\circ &\Gamma  &  &(p;q)              & = & \Gamma^*(\tilde{p};q) &= &(p;q)             &\\
\label{eq:EmbeddingAndInverseApp}
   &\Gamma  &\circ  &\Gamma^*& & (\tilde{p};q) & = & \Gamma(p;q)                &= &(\tilde p;q). &
\end{empheq}
\end{definition}

\begin{lemma}\label{lemma3} 
Let $p=\{p_i\}_{i=1}^n$ be a probability distribution with decreasing order, $p_1\ge p_2\ge\cdots\ge p_n$, and $q\coloneqq\{\frac{d_i}{N}\}_{i=1}^n=\{q_i\}_{i=1}^n$ be a full-rank probability distribution of dimensions $n$ such that $d_i\in\mathbb Z^+$. Define the following fine-grained $N$-dimensional distribution
\begin{equation}\label{eq:FineGrainedApp}
    \tilde p\coloneqq\Gamma(p;q)\coloneqq\Biggl(\underbrace{\frac{p_1}{d_1},\cdots,\frac{p_1}{d_1}}_{d_1},\underbrace{\frac{p_2}{d_2},\cdots,\frac{p_2}{d_2}}_{d_2},\cdots\cdots\underbrace{\frac{p_n}{d_n},\cdots,\frac{p_n}{d_n}}_{d_n}\Biggr).
\end{equation}
Then, for all $\alpha\in(-\infty,\infty)$, and $\eta_N$, the $N$-dimensional uniform probability distribution \eqref{eq:SelfFineGrainedApp}, we have
\begin{equation}\label{eq:Lemma13App}
    D_\mtn{\alpha}(p\|q)=D_\mtn{\alpha}(\tilde p\|\eta_N).
\end{equation}
\end{lemma}

\begin{proof}
We consider distinct intervals in our proof. Hence, 
\begin{enumerate}[label=(\roman*)]
    \item For $\alpha\in\mathbb R\setminus\{-\infty,0,1,+\infty\}$:
\begin{subequations}
\begin{eqnarray}\label{eq:Lemma13Proof_1App}
    D_\mtn{\alpha}(p\|q)&=& \frac{\sgn\alpha}{\alpha-1}\left( \sum_{i=1}^np_i^{\alpha}q_i^{1-\alpha}-1 \right)\\
                           &=& \frac{\sgn\alpha}{\alpha-1}\left( \sum_{i=1}^n d_i \left(\frac{p_i}{d_i}\right)^{\alpha}\left(\frac{1}{N}\right)^{1-\alpha}-1 \right)\\
                           &=& \frac{\sgn\alpha}{\alpha-1}\left( \sum_{i=1}^N  \left(\frac{p_i}{d_i}\right)^{\alpha}\left(\frac{1}{N}\right)^{1-\alpha}-1 \right)\\
                           &\equiv& D_\mtn{\alpha}(\tilde p\|\eta_N).
\end{eqnarray}
\end{subequations}

    \item For $\alpha=0$:
\begin{subequations}
\begin{eqnarray}\label{eq:Lemma13Proof_2App}
    D_{0}(p\|q)&=&\lim_{\alpha\rightarrow0^{+}}D_\mtn{\alpha}(p\|q)\\
                      &=& -\left(\sum_{i:p_i\neq0}q_i -1\right)\\
                      &=& -\left(\sum_{i:p_i\neq0}d_i\frac{1}{N} -1\right)\\
                      &=& -\left(\sum_{i:\tilde p_i\neq0}\frac{1}{N} -1\right)\\
                      &\equiv& D_{0}(\tilde p\|\eta_N).
\end{eqnarray}
\end{subequations}

   \item For $\alpha=1$:
\begin{subequations}
\begin{eqnarray}\label{eq:Lemma13Proof_3App}
    D_{1}(p\|q)&=&\sum_{i=1}^n p_i\ln \left(p_i\frac{N}{d_i} \right)\\
                      &=&\sum_{i=1}^n d_i \frac{p_i}{d_i}\ln \left(\frac{p_i}{d_i}N\right)\\
                      &=&\sum_{i=1}^N \tilde p_i \ln\left(\frac{\tilde p_i}{\eta_{N,i}}\right)\\
                      &\equiv& D_{1}(\tilde p\|\eta_N).
\end{eqnarray}
\end{subequations}

   \item For $\alpha=+\infty$:
\begin{subequations}
\begin{eqnarray}\label{eq:Lemma13Proof_4App}
    D_{+\infty}(p\|q)&=& \ln\max_i\frac{p_i}{q_i}\\                              
                            &\equiv& D_{+\infty}(\tilde p\|\eta_N).
\end{eqnarray}
\end{subequations}

   \item For $\alpha=-\infty$:
\begin{subequations}
\begin{eqnarray}\label{eq:Lemma13Proof_5App}
    D_{-\infty}(p\|q)&=&D_{+\infty}(q\|p)\\
                            &=&D_{+\infty}(\eta_N\|\tilde p)\equiv D_{-\infty}(\tilde p\|\eta_N).
\end{eqnarray}
\end{subequations}

\end{enumerate}
Therefore, we have
\begin{subequations}
\begin{eqnarray}
\label{eq:Lemma13-1LastApp}    D_\mtn{\alpha}\left(p\|q\right)&=&D_\mtn{\alpha}\left(\tilde{p}\|\eta_N\right)\\
\label{eq:Lemma13-2LastApp}                                 &=&D_\mtn{\alpha}\left(\Gamma(p;q)\|\Gamma(q;q)\right),
\end{eqnarray}
\end{subequations}
for all $\alpha\in\mathbb{R}$.
\end{proof}

\renewcommand{\theequation}{D\arabic{equation}}
\setcounter{equation}{0}  
\appsection{Non-additive Divergence and Catalytic \texorpdfstring{$d$}{d}-Majorization}  \label{appD}

The relation between the catalytic thermal transformations and non-additive divergence $D_\mtn{\alpha}$ is given in the following theorem.

\begin{theorem}\label{thm2} 
Given probability distributions $p,p',q,q'$ with $q,q'$ of full rank, the following conditions are equivalent:
\begin{enumerate}[label=(\Roman*)]
    \item\label{it:Thm17iApp}$D_\mtn{\alpha}(p\|q)\geq D_\mtn{\alpha}(p'\|q')$, \quad$\forall\alpha\in\mathbb{R}$,
    \item\label{it:Thm17iiApp} For any $\epsilon>0$, there exist probability distributions $r,s$ of full rank, a distribution $p'_\epsilon$ and a stochastic mapping $\Lambda$ such that
    \begin{enumerate}
        \item\label{it:Thm17IIiApp} $\Lambda(p\otimes r)=p'_{\epsilon}\otimes r$ 
        \item\label{it:Thm17IIiiApp} $\Lambda(q\otimes s)=q'\otimes s$
        \item\label{it:Thm17IIiiiApp} $\|p'-p'_{\epsilon}\|\leq\epsilon$
    \end{enumerate}
\end{enumerate}

\end{theorem}
Theorem~\ref{thm2} is the trumping characterization of Ref.~\cite{BrandaoEtAl2015} (Theorem~17 therein) in its general full-rank form; the uniform-reference specialization $s=\eta$ invoked there for catalytic thermal operations with a trivial catalyst Hamiltonian is not needed in the present setting.

\begin{proof}
We first consider the proof along the direction \ref{it:Thm17iApp}$\Ra$\ref{it:Thm17iiApp} so that we assume $D_\mtn{\alpha}$ satisfy the inequality in \ref{it:Thm17iApp}. Then, one should prove the existence of a stochastic mapping $\Lambda$ with properties endowed in \ref{it:Thm17iiApp}. To this aim,

    \begin{enumerate}
        \item[\textbf{(A)}] The case where components of $q$ and $q'$ are all rational:\\
Assume condition \ref{it:Thm17iApp} holds and $q_i, q'_i\in\mathbb Q$, that is, $q_i=\frac{d_i}{N}$, $q'_i=\frac{d'_i}{N}$,  $\forall i =1,\cdots,n$. By Lemma~\ref{lemma3}, we have \eqref{eq:Lemma13-1LastApp}, i.e.,
\[
D_\mtn{\alpha}(p\|q)=D_\mtn{\alpha}(\tilde{p}\|\eta_N),\qquad\forall\alpha\in\mathbb{R},
\]
with $\tilde{p}=\Gamma(p;q)$ given by \eqref{eq:FineGrainedApp} and $\eta_N=\Gamma(q;q)$  given by \eqref{eq:SelfFineGrainedApp}.
Due to Proposition~\ref{prop1}, for any $\epsilon>0$, there exists $\tilde{p}'_{\epsilon}$ such that $\|\tilde{p}'-\tilde{p}'_{\epsilon}\|\le\epsilon$ and $\tilde{p}'$ can be trumped into $\tilde{p}'_{\epsilon}$, that is,
\begin{equation}\label{eq:TrumpingCnddApp}
f_\mtn{\alpha}(\tilde{p})>f_\mtn{\alpha}(\tilde{p}'_{\epsilon}),\qquad\forall\alpha\in\mathbb{R},  
\end{equation}
where functions $f_\mtn{\alpha}$ are given in \eqref{eq:TrumpingFncsApp}. This means that there exists a catalyst in a state represented by the probability vector $r$ and a bi-stochastic mapping $\Phi$ such that~\cite{Klimesh2004,Turgut2007,Klimesh2007}
\begin{subequations}
\begin{eqnarray}
\label{eq:BiStochastic1MapApp}
    \Phi&:&\tilde{p}\otimes r\mapsto\Phi(\tilde{p}\otimes r)\\
\label{eq:BiStochastic2MapApp}        &:&\tilde{p}\otimes r\mapsto\tilde{p}'_{\epsilon}\otimes r.
\end{eqnarray}
\end{subequations}

As it satisfies \eqref{eq:BiStochastic2MapApp}, whether $r$ is of full rank or not becomes unimportant under a bi-stochastic mapping. Since the maximally mixed state is preserved under $\Phi$, satisfying the requirement \eqref{eq:Lemma16-1App}, it is always possible, through Lemma~\ref{lemma1}, to decompose $\Phi$ into a direct sum of two mappings, one acting only on the support of $r$, and one acting only on the rest part of $r$ which corresponds to components all equal to zero. So, without loss of generality, we can take $r$ of full rank. With the following map
\begin{equation}
    \Lambda=(\Gamma'^*\otimes I)\circ\Phi\circ(\Gamma\otimes I)
\end{equation}
(i) and (iii) in \ref{it:Thm17iiApp} are satisfied: 
\begin{subequations}
\begin{align}
    \Lambda:p\otimes r&\mapsto (\Gamma'^*\otimes I)\circ\Phi\circ(\Gamma\otimes I)(p\otimes r)\\
                      &\mapsto (\Gamma'^*\otimes I)\circ\Phi\Bigl(\Gamma(p;q)\otimes I(r)\Bigr)\\
                      &\mapsto (\Gamma'^*\otimes I)\circ\Phi(\tilde{p}\otimes r\\
                      &\mapsto (\Gamma'^*\otimes I)(\tilde{p}'_{\epsilon}\otimes r)\\
                      &\mapsto \Gamma'^*(\tilde{p}'_{\epsilon};q')\otimes I(r)\\
                      &\mapsto p'_{\epsilon}\otimes r.
\end{align}
\end{subequations}
Here, $I$ is identity, and '$'$' mappings apply to '$'$' distributions, that is,
\begin{subequations}
\begin{eqnarray}
\label{eq:InverseAndEmbeddingP1App}    \Gamma'(p';q')&=&\tilde{p}'\\
\label{eq:InverseAndEmbeddingP2App}    \Gamma'^*(\tilde{p}_{\epsilon}';q')&=&p_{\epsilon}'.
\end{eqnarray}
\end{subequations}
Let us check for (ii) in \ref{it:Thm17iiApp}:
\begin{subequations}
\begin{align}
    \Lambda:q\otimes s&\mapsto (\Gamma'^*\otimes I)\circ\Phi\circ(\Gamma\otimes I)(q\otimes s)\\
     &\mapsto (\Gamma'^*\otimes I)\circ\Phi\Bigl(\Gamma(q;q)\otimes I(s)\Bigr)\\
\label{eq:CatalystSApp}     &\mapsto (\Gamma'^*\otimes I)\circ\Phi\Bigl(\frac{\mathbb{I}}{N}\otimes s\Bigr),
\end{align}
\end{subequations}
where $\frac{\mathbb{I}}{N}=\eta_N$ is the $N$-dimensional uniform distribution. Take $s=\eta_\mtn{S}=\frac{\mathbb{I}}{S}$, a uniform distribution of dimension $S$
\begin{subequations}
\begin{align} 
\Lambda:q\otimes s&\mapsto (\Gamma'^*\otimes I)\circ\Phi\Bigl(\frac{\mathbb{I}}{N}\otimes\frac{\mathbb{I}}{S}\Bigr),\\
                  &\mapsto (\Gamma'^*\otimes I)\Bigl(\frac{\mathbb{I}}{N}\otimes\frac{\mathbb{I}}{S}\Bigr)\\
                  &\mapsto (\Gamma'^*\otimes I)\Bigl(\Gamma'(q';q')\otimes\frac{\mathbb{I}}{S}\Bigr)\\
                  &\mapsto\Biggl(\Gamma'^*\circ\Gamma'\Bigl(q';q'\Bigr)\otimes \frac{\mathbb{I}}{S}\Biggr)\\
                  &\mapsto q'\otimes s,
\end{align}
\end{subequations}
i.e., (ii) in \ref{it:Thm17iiApp} is also satisfied.

\item[\textbf{(B)}] The case where some components of $q$ and $q'$ are irrational:\\
Assume the condition \ref{it:Thm17iApp} holds and that, $\exists i,j\in\{1,\cdots,n\}$ such that $q_i,q_j'\in\mathbb{R}\setminus\mathbb{Q}$. Lemma~\ref{lemma2} lets us define channels $L,L'$ such that  
\begin{subequations}
\begin{empheq}{alignat=4}
\label{eq:ChannelQApp}            L&:q&&\mapsto L(q)&&=\qbar,\qquad\qquad\quad&&\qbar_i\in\mathbb{Q},\\
\label{eq:ChannelQPApp}           L'&:q'&&\mapsto L'(q')&&=\qbar',            &&\qbar_i'\in\mathbb{Q},\qquad\forall i\in\{1,\cdots,n\}.  
\end{empheq}
\end{subequations}

Let $q_n,q'_n$ be the smallest components of $q,q'$, respectively. Then, the statistical distances between the distributions and mapped ones found in Lemma~\ref{lemma2} for $M,M'\in\mathbb{Z}^+$ satisfy
\begin{subequations}
\begin{empheq}{alignat=6}
\label{eq:StatDistances1App}
    \|q  - L(q) \|  &\equiv\|q  -\qbar \|&&<\frac{n}{M}\\
\label{eq:StatDistances2App}
    \|p  - L(p)\|   &\equiv\|p  -\pbar \|&&<\frac{n}{M}\frac{1}{q_n}\\
\label{eq:StatDistances3App}
    \|q' - L'(q')\| &\equiv\|q' -\qbar'\|&&<\frac{n}{M'}\\
\label{eq:StatDistances4App}
    \|p' - L'(p')\| &\equiv\|p' -\pbar'\|&&<\frac{n}{M'}\frac{1}{q'_n}.
\end{empheq}
\end{subequations}
For $q,q'$ of full-rank, $q_n\ne0$ and $q'_n\ne0$. Since $M,M'$ are arbitrary, they can be chosen arbitrarily large. Hence, for $0<n,q_n,q'_n<\infty$ we have

\begin{subequations}
\begin{empheq}{alignat=2}
\label{eq:MLimits1App}
        \lim_{M\ra\infty}D_\mtn{\alpha}(p\|q) &\ra D_\mtn{\alpha}(\pbar\|\qbar)\\
\label{eq:MLimits2App}
        \lim_{M'\ra\infty}D_\mtn{\alpha}(p'\|q') &\ra D_\mtn{\alpha}(\pbar'\|\qbar').
\end{empheq}
\end{subequations}
\end{enumerate}

Therefore, the inequality in condition \ref{it:Thm17iApp} is valid under mappings $L$ and $L'$, 
\begin{equation}\label{eq:it:Thm17iModifiedApp}
    D_\mtn{\alpha}(\pbar\|\qbar)\ge D_\mtn{\alpha}(\pbar'\|\qbar'),
\end{equation}
so that the direction \ref{it:Thm17iApp}$\Rightarrow$\ref{it:Thm17iiApp} is proved.

We now consider the reverse direction and aim to show the existence of a stochastic mapping $\Lambda$ with properties \ref{it:Thm17iiApp} necessitates that $D_\mtn{\alpha}$ satisfy the inequality in \ref{it:Thm17iApp}.

In order to obtain a maximally mixed distribution preserved under a bi-stochastic mapping, we chose $s$ in \eqref{eq:CatalystSApp} as an $S$-dimensional uniform probability distribution. Since it is arbitrary, choose $S$ so that the support of $s$ includes the support of $r$ in \ref{it:Thm17iiApp}, that is, $s\equiv\eta_\mtn{S}$ is the uniform distribution onto the support of $r$. Under these conditions, due to the monotonicity of the divergences, $D_\mtn{\alpha}$ can be written as the following inequality;
\begin{subequations}
\begin{empheq}{alignat=2}
\label{eq:MonotonicityD0App}       D_\mtn{\alpha}\Bigl(p\otimes r\|q\otimes s\Bigr)&\ge D_\mtn{\alpha}\Bigl(\Lambda(p\otimes r)\|\Lambda(q\otimes s)\Bigr)\\
\label{eq:MonotonicityDApp}        D_\mtn{\alpha}\Bigl(p\otimes r\|q\otimes s\Bigr)&\ge D_\mtn{\alpha}\Bigl(p'_{\epsilon}\otimes r\|q'\otimes s\Bigr), \quad\qquad\forall\alpha\in\mathbb{R}.
\end{empheq}
\end{subequations}

\noindent Moreover, since the non-additive divergence $D_\mtn{\alpha}$ is pseudo-additive~\cite{FuruichiYanagiKuriyama2004}, \eqref{eq:MonotonicityDApp} can further be written as

\begin{subequations}
\begin{empheq}{alignat=2}
\label{eq:MonotonicityDT1App}
        D_\mtn{\alpha}\Bigl(p\otimes r\|q\otimes s\Bigr)&\ge 
        D_\mtn{\alpha}\Bigl(p'_{\epsilon}\otimes r\|q'\otimes s\Bigr)\\
\label{eq:MonotonicityDT2App}        D_\mtn{\alpha}\Bigl(p\|q\Bigr)+D_\mtn{\alpha}\Bigl(r\|s\Bigr)&+\sgn{\alpha}(\alpha-1)D_\mtn{\alpha}\Bigl(p\|q\Bigr)D_\mtn{\alpha}\Bigl(r\|s\Bigr)\notag\\
&\ge \\ 
    D_\mtn{\alpha}\Bigl(p'_{\epsilon}\|q'\Bigr)+D_\mtn{\alpha}\Bigl(r\|s\Bigr)&+\sgn{\alpha}(\alpha-1)D_\mtn{\alpha}\Bigl(p'_{\epsilon}\|q'\Bigr)D_\mtn{\alpha}\Bigl(r\|s\Bigr).\notag
\end{empheq}
\end{subequations}
The finiteness of divergences whenever both arguments are of full rank---here $q$, $q'$ and $s$ in the second slot, and the approximant $p'_{\epsilon}$ and the reference $r$ in the first, the former full-rank by construction and the latter by Theorem~\ref{thm2}---allows us to write 
\begin{widetext}
\begin{subequations}
\begin{alignat}{4}
    & \Biggl\{1+\sgn{\alpha}(\alpha-1) &&D_\mtn{\alpha}\Bigl(r\|s\Bigr)\Biggr\}
    &&\times\quad\Biggl[D_\mtn{\alpha}\Bigl(p\|q\Bigr)
    &&-D_\mtn{\alpha}\Bigl(p'_{\epsilon}\|q'\Bigr)\Biggr]\ge 0, \label{eq:MonotonicityDT3_0App}\\
    \alpha\in(-\infty,0)&: &&\ge0 &&\Ra\quad &&\ge0,\label{eq:MonotonicityDT3_1App}\\
    \alpha\in[0,1)&: &&\in[0,(1-\alpha)^{-1}] &&\Ra\quad &&\ge0,\label{eq:MonotonicityDT3_2App}\\
    \alpha=1&: &&=1 &&\Ra\quad &&\ge0,\label{eq:MonotonicityDT3_3App}\\
    \alpha\in(1,+\infty)&: &&\ge1 &&\Ra\quad &&\ge0.\label{eq:MonotonicityDT3_4App}
\end{alignat}
\end{subequations}
\end{widetext}
Therefore,
    \begin{equation}\label{eq:MonotonicityDT3App}
        D_\mtn{\alpha}\Bigl(p\|q\Bigr)\ge D_\mtn{\alpha}\Bigl(p'_{\epsilon}\|q'\Bigr),\qquad\forall\alpha\in\mathbb{R}.
    \end{equation}
If $p$ itself is not of full rank, then $D_\mtn{\alpha}(p\|q)=+\infty$ for
$\alpha<0$ and inequality~\eqref{eq:MonotonicityDT3_0App} holds trivially in that
range.

\noindent Note that \eqref{eq:MonotonicityDT3_2App} can be seen to hold, since
    \begin{empheq}{alignat=4}
                & && &&\sum_{i=1}^np_i^{\alpha}q_i^{1-\alpha}&&\ge0\notag\\
\alpha\in[0,1)\Ra\quad& &&\quad\frac{1}{\alpha-1}&&\sum_{i=1}^np_i^{\alpha}q_i^{1-\alpha}&&\le0\notag\\
& &&\quad\frac{1}{\alpha-1}&&\Bigl(\sum_{i=1}^np_i^{\alpha}q_i^{1-\alpha}-1\Bigr)&&\le\Bigl(\frac{1}{1-\alpha}\Bigr)\notag\\
\label{eq:AnIntervalForDTAppD}D_\mtn{\alpha}(p\|q)\ge0\Ra\quad&0\:&&\le&&D_\mtn{\alpha}(p\|q)&&\le\Bigl(\frac{1}{1-\alpha}\Bigr).
    \end{empheq}

\noindent Note also that $D_\mtn{\alpha}$ is continuous with respect to $p'$ and $q'$ for $\alpha>0$. For $\alpha<0$, if $p'$ is of full rank, then $D_\mtn{\alpha}$ is continuous with respect to $p'$ and $q'$. If $p'$ is not of full rank, then $D_\mtn{\alpha}\ra\infty$ by preserving the previously listed inequalities and \eqref{eq:MonotonicityDT3App} in particular. For $\alpha=0$, we have $\lim_{\alpha\ra0^+}D_\mtn{\alpha}(p'\|q')=D_0(p'\|q')$. 

Therefore,
\begin{equation}\label{eq:Thm17EpsilonZeroLimitApp}
    \lim_{p'_{\epsilon}\ra p'}   D_\mtn{\alpha}(p'_{\epsilon}\|q')=D_\mtn{\alpha}(p'\|q'),\qquad\forall\alpha\in\mathbb{R},
\end{equation}
and hence 
\begin{equation}\label{eq:MonotonicityDRTApp}
        D_\mtn{\alpha}(p\|q)\ge D_\mtn{\alpha}(p'\|q'),\qquad\forall\alpha\in\mathbb{R}.
\end{equation}
\end{proof}
\textbf{Note:}\label{it:NoteForTheorem18App} The following statement, taken from Ref.~\cite{BrandaoEtAl2015}, is used in the proof below; we quote it without proof.
  
\textit{Consider a system that has undergone a process under catalytic thermal operations. If the input state of the system is block-diagonal in the energy eigenbasis, then the diagonal elements of the output state of the system depend only on the diagonal elements of the catalyst, not on the non-diagonal elements of the state of the catalyst. Each block has the same energy and has a degeneracy of the side size of the block.
}

\renewcommand{\theequation}{E\arabic{equation}}
\setcounter{equation}{0}  
\appsection{Proof of Theorem \texorpdfstring{\ref{thm3}}{1}}  \label{appE}

This appendix provides the proof of Theorem~\ref{thm3}, stated in Sec.~\ref{sec::T2ndLaw}. Throughout, $\rho_\mtn{S}^\mtn{\beta}=e^{-\beta H_\mtn{S}}/Z$ denotes the thermal state corresponding to the Hamiltonian $H_\mtn{S}$, with eigenvalue equation $H_\mtn{S}|E_{i,g}\rangle=E_{i}|E_{i,g}\rangle$ for a $g$-fold degenerate eigenstate $|E_{i,g}\rangle$ of energy $E_{i}$, and $Z=\sum_{i,g}e^{-\beta E_i}$ is the partition function.

\begin{proof}[Proof of Theorem~\ref{thm3}]
\begin{enumerate} 
        \item[\textbf{(A)}]  Assuming that
\begin{equation}\label{eq:Thm18ProofA1App}
    D_\mtn{\alpha}(\rho^i_\mtn{S}\|\rho_\mtn{S}^\mtn{\beta})\geq D_\mtn{\alpha}(\rho^f_\mtn{S}\|\rho_\mtn{S}^\mtn{\beta}),
\end{equation}
so that, by Theorem~\ref{thm2}, there exist a channel $\Lambda$, a catalyst state $\rho^i_\mtn{C}$, and full-rank reference states $\rho^\mtn{\beta}_{\text{ancilla}}$ and $\rho^\mtn{\beta}_\mtn{C}$ such that
\begin{equation}\label{eq:Thm18ProofA2App}
    \Lambda:(\rho^i_\mtn{S}\otimes\rho^i_\mtn{C})\mapsto(\rho^{out}_\mtn{S}\otimes\rho^i_\mtn{C})
\end{equation}
with
\begin{equation}\label{eq:Thm18ProofA3App}
    \|\rho^f_\mtn{S}-\rho^{out}_\mtn{S}\|\le\epsilon
\end{equation}
for arbitrary $\epsilon>0$, and
\begin{equation}\label{eq:Thm18ProofA4App}
    \Lambda:(\rho^\mtn{\beta}_\mtn{S}\otimes\rho^\mtn{\beta}_{\text{ancilla}}\otimes\rho^\mtn{\beta}_\mtn{C})\mapsto(\rho^\mtn{\beta}_\mtn{S}\otimes\rho^\mtn{\beta}_{\text{ancilla}}\otimes\rho^\mtn{\beta}_\mtn{C}).
\end{equation}
Here, $\rho_\mtn{C}^i$ and $\rho^\mtn{\beta}_\mtn{C}$ are the initial state and the thermal state of the catalyst, respectively, and $\rho^\mtn{\beta}_{\text{ancilla}}$ is the Gibbs state of the ancilla. Since any full-rank distribution is the Gibbs state of a suitably chosen Hamiltonian, the catalyst reference provided by Theorem~\ref{thm2} is identified with the thermal state $\rho^\mtn{\beta}_\mtn{C}$. For $H_{\text{ancilla}}$ and $H_\mtn{C}$ being the Hamiltonian of the ancillary system and the catalyst, respectively, if
\begin{equation}\label{eq:Thm18ProofA5App}
    [\rho^i_\mtn{S}\otimes\rho^i_{\text{ancilla}},H_\mtn{S}+H_{\text{ancilla}}]=0
\end{equation}
then, it would be enough for the catalyst to consider only states $\rho^i_\mtn{C}$ that satisfy
\begin{equation}\label{eq:Thm18ProofA6App}
    [\rho^i_\mtn{S}\otimes\rho^i_{\text{ancilla}}\otimes\rho^i_\mtn{C},H_\mtn{S}+H_{\text{ancilla}}+H_\mtn{C}]=0.
\end{equation}

It is worth noting that the dimension of the catalyst-side reference supplied by Theorem~\ref{thm2} is not fixed in advance; it is part of the construction, and may be large. The ancilla provides precisely this freedom: since its dimension and Hamiltonian can be chosen at will, the reference state $\rho^\mtn{\beta}_{\text{ancilla}}\otimes\rho^\mtn{\beta}_\mtn{C}$ preserved under the mapping $\Lambda$ \eqref{eq:Thm18ProofA4App} realizes a full-rank thermal state of whatever dimension the construction requires.

\item[\textbf{(B)}] Assume the existence of a channel $\Lambda$ that catalytically transforms $\rho_\mtn{S}^i$ to $\rho_\mtn{S}^{out}$ under the catalysis of $C$ with Hamiltonian $H_\mtn{C}$ such that
\begin{subequations}
\begin{empheq}{alignat=3}
\label{eq:Thm18ProofB1App}    &\Lambda:(\rho^i_\mtn{S}\otimes\rho^i_\mtn{C})\mapsto(\rho^{out}_\mtn{S}\otimes\rho^i_\mtn{C}),\\
\label{eq:Thm18ProofB2App}    &\Lambda:(\rho^\mtn{\beta}_\mtn{S}\otimes\rho^\mtn{\beta}_\mtn{C})\mapsto(\rho^\mtn{\beta}_\mtn{S}\otimes\rho^\mtn{\beta}_\mtn{C}),\\
\label{eq:Thm18ProofB3App}    &\|\rho^f_\mtn{S}-\rho^{out}_\mtn{S}\|\le\epsilon,
\end{empheq}
\end{subequations}
for arbitrary $\epsilon>0$. Here, $\rho^\mtn{\beta}_\mtn{S}$ and $\rho^\mtn{\beta}_\mtn{C}$ are the thermal states of the system $S$ and the catalyst $C$, respectively. Then, with Theorem~\ref{thm2}, we have the following:

\begin{equation}\label{eq:Thm18ProofB4App}
    D_\mtn{\alpha}(\rho^i_\mtn{S}\|\rho_\mtn{S}^\mtn{\beta})\geq D_\mtn{\alpha}(\rho^f_\mtn{S}\|\rho_\mtn{S}^\mtn{\beta}),
\end{equation}
for the states $\rho_\mtn{S}^i$, $\rho_\mtn{S}^f$ that always commute with the system's Hamiltonian.
\end{enumerate}
\end{proof}

\renewcommand{\theequation}{F\arabic{equation}}
\setcounter{equation}{0}  
\appsection{Work Bit: Proof of Theorem \texorpdfstring{\ref{thm5}}{3}}  \label{appF}

This appendix provides the proof of Theorem~\ref{thm5}, stated in Sec.~\ref{sec::T2ndLaw}.

\begin{proof}[Proof of Theorem~\ref{thm5}]
\begin{enumerate} 
        \item[\textbf{(A)}]  Assuming that
\begin{equation}\label{eq:Thm20ProofA1App}
    D_\mtn{\alpha}(\rho^i_\mtn{S}\otimes\rho^i_\mtn{Q}\|\rho_\mtn{S}^\mtn{\beta}\otimes\rho_\mtn{Q}^\mtn{\beta})\geq D_\mtn{\alpha}(\rho^f_\mtn{S}\otimes\rho^f_\mtn{Q}\|\rho_\mtn{S}^\mtn{\beta}\otimes\rho_\mtn{Q}^\mtn{\beta}), \qquad\forall\alpha\in\mathbb{R}
\end{equation}
so that, by Theorem~\ref{thm2}, there exist a channel $\Lambda$, a catalyst state $\rho^i_\mtn{C}$, and full-rank reference states $\rho^\mtn{\beta}_{\text{ancilla}}$ and $\rho^\mtn{\beta}_\mtn{C}$ such that
\begin{equation}\label{eq:Thm20ProofA2App}
    \Lambda:(\rho^i_\mtn{S}\otimes\rho^i_\mtn{C}\otimes\rho^i_\mtn{Q})\mapsto(\rho^{out}_\mtn{S}\otimes\rho^i_\mtn{C}\otimes\rho^{out}_\mtn{Q})
\end{equation}
with
\begin{equation}\label{eq:Thm20ProofA3_1App} 
    \|\rho^f_\mtn{S}\otimes\rho^f_\mtn{Q}-\rho^{out}_\mtn{S}\otimes\rho^{out}_\mtn{Q}\|\le\epsilon  
\end{equation}
for arbitrary $\epsilon>0$, and
\begin{equation}\label{eq:Thm20ProofA4App}
    \Lambda:(\rho^\mtn{\beta}_\mtn{S}\otimes\rho^\mtn{\beta}_{\text{ancilla}}\otimes\rho^\mtn{\beta}_\mtn{C}\otimes\rho^\mtn{\beta}_\mtn{Q})\mapsto(\rho^\mtn{\beta}_\mtn{S}\otimes\rho^\mtn{\beta}_{\text{ancilla}}\otimes\rho^\mtn{\beta}_\mtn{C}\otimes\rho^\mtn{\beta}_\mtn{Q}).
\end{equation}
Here, $\rho_\mtn{C}^i$ and $\rho^\mtn{\beta}_\mtn{C}$ are the initial and the thermal states of the catalyst, respectively, and $\rho^\mtn{\beta}_{\text{ancilla}}$ is the Gibbs state of the ancilla. As in Appendix~\ref{appE}, the catalyst reference provided by Theorem~\ref{thm2} is identified with the thermal state $\rho^\mtn{\beta}_\mtn{C}$. $H_{\text{ancilla}}$ and $H_\mtn{C}$ being the Hamiltonian of the ancillary system and the catalyst, respectively, if
\begin{equation}\label{eq:Thm20ProofA5App}
    [\rho^i_\mtn{S}\otimes\rho^i_{\text{ancilla}}\otimes\rho^i_\mtn{Q},H_\mtn{S}+H_{\text{ancilla}}+H_\mtn{Q}]=0
\end{equation}
then, it would be enough for the catalyst to consider only states $\rho^i_\mtn{C}$ that satisfy
\begin{equation}\label{eq:Thm20ProofA6App}
    [\rho^i_\mtn{S}\otimes\rho^i_{\text{ancilla}}\otimes\rho^i_\mtn{C}\otimes\rho^i_\mtn{Q},H_\mtn{S}+H_{\text{ancilla}}+H_\mtn{C}+H_\mtn{Q}]=0.
\end{equation}

Again, it is worth noting that the possibility of an ancilla ensures that the reference state $\rho^\mtn{\beta}_{\text{ancilla}}\otimes\rho^\mtn{\beta}_\mtn{C}$ preserved under the mapping $\Lambda$ \eqref{eq:Thm20ProofA4App} is of any dimension.

\item[\textbf{(B)}] Assume now the existence of a channel $\Lambda$ that catalytically transforms $\rho_\mtn{S}^i\otimes\rho^i_\mtn{Q}$ to $\rho_\mtn{S}^{out}\otimes\rho^{out}_\mtn{Q}$ under the catalysis of $C$ with Hamiltonian $H_\mtn{C}$ such that
\begin{subequations}
\begin{empheq}{alignat=3}
\label{eq:Thm20ProofB1App}    &\Lambda:(\rho^i_\mtn{S}\otimes\rho^i_\mtn{C}\otimes\rho^i_\mtn{Q})\mapsto(\rho^{out}_\mtn{S}\otimes\rho^i_\mtn{C}\otimes\rho^{out}_\mtn{Q}),\\
\label{eq:Thm20ProofB2App}    &\Lambda:(\rho^\mtn{\beta}_\mtn{S}\otimes\rho^\mtn{\beta}_\mtn{C}\otimes\rho^\mtn{\beta}_\mtn{Q})\mapsto(\rho^\mtn{\beta}_\mtn{S}\otimes\rho^\mtn{\beta}_\mtn{C}\otimes\rho^\mtn{\beta}_\mtn{Q}),\\
\label{eq:Thm20ProofB3App}    &\|\rho^f_\mtn{S}\otimes\rho^f_\mtn{Q}-\rho^{out}_\mtn{S}\otimes\rho^{out}_\mtn{Q}\|\le\epsilon,
\end{empheq}
\end{subequations}
for arbitrary $\epsilon>0$. Here, $\rho^\mtn{\beta}_\mtn{S}$, $\rho^\mtn{\beta}_\mtn{C}$ and $\rho^\mtn{\beta}_\mtn{Q}$ are the thermal states of the system $S$, the catalyst $C$, and the qubit $Q$, respectively. Theorem~\ref{thm2} allows us to write
\begin{widetext}
\begin{subequations}
\begin{eqnarray}
\label{eq:Thm20ProofB4_1App}
            D_\mtn{\alpha}(\rho^i_\mtn{S}\otimes\rho^i_\mtn{C}\otimes\rho^i_\mtn{Q}\|\rho_\mtn{S}^\mtn{\beta}\otimes\rho_\mtn{C}^\mtn{\beta}\otimes\rho_\mtn{Q}^\mtn{\beta})
    &\geq&  D_\mtn{\alpha}(\Lambda(\rho^i_\mtn{S}\otimes\rho^i_\mtn{C}\otimes\rho^i_\mtn{Q})\|\Lambda(\rho_\mtn{S}^\mtn{\beta}\otimes\rho_\mtn{C}^\mtn{\beta}\otimes\rho_\mtn{Q}^\mtn{\beta}))\\
\label{eq:Thm20ProofB4_2App}
    &=&     D_\mtn{\alpha}(\rho^{out}_\mtn{S}\otimes\rho^i_\mtn{C}\otimes\rho^{out}_\mtn{Q}\|\rho_\mtn{S}^\mtn{\beta}\otimes\rho_\mtn{C}^\mtn{\beta}\otimes\rho_\mtn{Q}^\mtn{\beta})\qquad\forall\alpha\in\mathbb{R}.
\end{eqnarray}
\end{subequations}
\end{widetext}

Using pseudo-additivity of divergence~\cite{FuruichiYanagiKuriyama2004} this inequality becomes, $\forall\alpha\in\mathbb{R}$,

\begin{empheq}{alignat=4}\label{eq:MonotonicityDT3_00App}
    &\, \Biggl\{1+\sgn{\alpha}(\alpha-1) &&D_\mtn{\alpha}\Bigl(\rho^i_\mtn{C}\|\rho^\mtn{\beta}_\mtn{C}\Bigr)\Biggr\}&&\times\Biggl[D_\mtn{\alpha}\Bigl(\rho^i_\mtn{S}\otimes\rho^i_\mtn{Q}\|\rho^\mtn{\beta}_\mtn{S}\otimes\rho^\mtn{\beta}_\mtn{Q}\Bigr)&&-D_\mtn{\alpha}\Bigl(\rho^{out}_\mtn{S}\otimes\rho^{out}_\mtn{Q}\|\rho^\mtn{\beta}_\mtn{S}\otimes\rho^\mtn{\beta}_\mtn{Q})\Biggr]\ge 0.\\
    \alpha\in[0,+\infty)&: &&>0 &&\Ra &&\ge0\nonumber
\end{empheq}

Since $\rho^i_\mtn{Q}$ is not of full-rank, for $\alpha\in(-\infty,0)$ the first term in the square brackets $[\cdot]$ in \eqref{eq:MonotonicityDT3_00App} goes to infinity and this makes it impossible to compare it with the second term.
For \(\alpha\in[0,+\infty)\), the braced prefactor is strictly positive, as follows directly from the definition of \(D_\alpha\) and the full rank of \(\rho_C^\beta\); thus
\begin{widetext}
    \begin{equation}\label{eq:MonotonicityDT33App}
        D_\mtn{\alpha}\Bigl(\rho^i_\mtn{S}\otimes\rho^i_\mtn{Q}\|\rho^\mtn{\beta}_\mtn{S}\otimes\rho^\mtn{\beta}_\mtn{Q}\Bigr)\ge D_\mtn{\alpha}\Bigl(\rho^{out}_\mtn{S}\otimes\rho^{out}_\mtn{Q}\|\rho^\mtn{\beta}_\mtn{S}\otimes\rho^\mtn{\beta}_\mtn{Q}\Bigr),\qquad\forall\alpha\in[0,+\infty).
    \end{equation}
\end{widetext}

\noindent Hence, 
\begin{widetext}
\begin{equation}\label{eq:Thm20EpsilonZeroLimitApp}
    \lim_{\rho^{out}_\mtn{S}\otimes\rho^{out}_\mtn{Q}\ra \rho^f_\mtn{S}\otimes\rho^f_\mtn{Q}}   D_\mtn{\alpha}\Bigl(\rho^{out}_\mtn{S}\otimes\rho^{out}_\mtn{Q}\|\rho^\mtn{\beta}_\mtn{S}\otimes\rho^\mtn{\beta}_\mtn{Q}\Bigr)=D_\mtn{\alpha}\Bigl(\rho^f_\mtn{S}\otimes\rho^f_\mtn{Q}\|\rho^\mtn{\beta}_\mtn{S}\otimes\rho^\mtn{\beta}_\mtn{Q}\Bigr),\qquad\forall\alpha\in[0,+\infty),
\end{equation}
\end{widetext}

\noindent so that
\begin{widetext}
\begin{equation}\label{eq:Thm20MonotonicityDRTApp}
        D_\mtn{\alpha}\Bigl(\rho^i_\mtn{S}\otimes\rho^i_\mtn{Q}\|\rho^\mtn{\beta}_\mtn{S}\otimes\rho^\mtn{\beta}_\mtn{Q}\Bigr)\ge D_\mtn{\alpha}\Bigl(\rho^f_\mtn{S}\otimes\rho^f_\mtn{Q}\|\rho^\mtn{\beta}_\mtn{S}\otimes\rho^\mtn{\beta}_\mtn{Q}\Bigr),\qquad\forall\alpha\in[0,+\infty).
\end{equation}
\end{widetext}

\end{enumerate}

\end{proof}

\renewcommand{\theequation}{G\arabic{equation}}
\setcounter{equation}{0}  
\appsection{Approximate Catalysis and Work Distance}  \label{appG}

Now we will consider the cases where in the process the output state of the catalyst differs from its input state
\begin{equation}\label{eq:RhoSRhoC}
    \rho^i_\mtn{S}\otimes\rho^i_\mtn{C}\longmapsto\rho^f_\mtn{S}\otimes\rho^f_\mtn{C}.
\end{equation}
If $\rho^f_\mtn{C}$ differs from $\rho^i_\mtn{C}$ by a sufficiently small amount so that with the assistance of the qubit, $Q$, the transition 
\begin{equation}\label{eq:CatalystTransition}
    \rho^f_\mtn{S}\otimes\rho^f_\mtn{C}\otimes\rho^i_\mtn{Q}\longmapsto\rho^f_\mtn{S}\otimes\rho^i_\mtn{C}\otimes\rho^f_\mtn{Q}
\end{equation}
is possible, then it would be possible to recover the second laws for all $\alpha\in[0,+\infty)$ expressed in \eqref{eq:QSecondLawsQubit}. In transition \eqref{eq:CatalystTransition}, let $\rho^i_\mtn{Q}=|0\rangle\langle0|$ with energy $E_\mtn{Q}^0=0$, and $\rho^f_\mtn{Q}=|\omega\rangle\langle\omega|$ with energy $E_\mtn{Q}^\omega=\hbar\omega$ where $\hbar$ is the reduced Planck constant. The energy difference $\Delta E_\mtn{Q}=\hbar\omega$ is positive if the work bit is initially in the ground state and negative if not.
The condition for this transition to occur is that the corresponding free energy would not increase, namely,
\begin{widetext}
\begin{equation}\label{eq:CatalystQubitCond_F}
    F_\mtn{\alpha}(\rho^f_\mtn{C}\otimes\rho^i_\mtn{Q},H_\mtn{C}+H_\mtn{Q})\ge F_\mtn{\alpha}(\rho^i_\mtn{C}\otimes\rho^f_\mtn{Q},H_\mtn{C}+H_\mtn{Q}),\qquad\forall\alpha\in[0,+\infty). 
\end{equation}
\end{widetext}
For the desired transition of the system $S$
\begin{equation}\label{eq:DesiredTransition}    \rho^i_\mtn{S}\otimes\rho^i_\mtn{C}\otimes\rho^i_\mtn{Q}\longmapsto\rho^f_\mtn{S}\otimes\rho^i_\mtn{C}\otimes\rho^f_\mtn{Q},
\end{equation}
the condition can be expressed as
\begin{widetext}
\begin{equation}\label{eq:SystemQubitCond_F}
    F_\mtn{\alpha}(\rho^i_\mtn{S}\otimes\rho^i_\mtn{C}\otimes\rho^i_\mtn{Q},H_\mtn{S}+H_\mtn{C}+H_\mtn{Q})\ge 
    F_\mtn{\alpha}(\rho^f_\mtn{S}\otimes\rho^i_\mtn{C}\otimes\rho^f_\mtn{Q},H_\mtn{S}+H_\mtn{C}+H_\mtn{Q}),\qquad\forall\alpha\in[0,+\infty). 
\end{equation}
\end{widetext}
The two conditions for the transition \eqref{eq:DesiredTransition} to occur can be expressed in terms of non-additive divergences: for \eqref{eq:CatalystQubitCond_F}, 
\begin{widetext}
\begin{equation}\label{eq:CatalystQubitCond_D}
    D_\mtn{\alpha}(\rho^f_\mtn{C}\otimes\rho^i_\mtn{Q}\|\rho^\mtn{\beta}_\mtn{C}\otimes\rho^\mtn{\beta}_\mtn{Q})\ge 
    D_\mtn{\alpha}(\rho^i_\mtn{C}\otimes\rho^f_\mtn{Q}\|\rho^\mtn{\beta}_\mtn{C}\otimes\rho^\mtn{\beta}_\mtn{Q}),\qquad\forall\alpha\in[0,+\infty), 
\end{equation}
\end{widetext}
and for \eqref{eq:SystemQubitCond_F},
\begin{widetext}
\begin{equation}\label{eq:SystemtQubitCond_D}
    D_\mtn{\alpha}(\rho^i_\mtn{S}\otimes\rho^i_\mtn{C}\otimes\rho^i_\mtn{Q}\|\rho^\mtn{\beta}_\mtn{S}\otimes\rho^\mtn{\beta}_\mtn{C}\otimes\rho^\mtn{\beta}_\mtn{Q})\ge
    D_\mtn{\alpha}(\rho^f_\mtn{S}\otimes\rho^i_\mtn{C}\otimes\rho^f_\mtn{Q}\|\rho^\mtn{\beta}_\mtn{S}\otimes\rho^\mtn{\beta}_\mtn{C}\otimes\rho^\mtn{\beta}_\mtn{Q}),\qquad\forall\alpha\in[0,+\infty),
\end{equation}
\end{widetext}
where $\rho^\mtn{\beta}$ are thermal states. Inequality \eqref{eq:CatalystQubitCond_D} is the condition for the recovery of the quantum second laws for inexact catalysis as expressed in \eqref{eq:SystemtQubitCond_D}. Using pseudo-additivity~\cite{FuruichiYanagiKuriyama2004} of $D_\mtn{\alpha}$, we can simplify \eqref{eq:SystemtQubitCond_D} as, $\forall\alpha\in[0,+\infty)$

\begin{empheq}{alignat=4}\label{eq:MonotonicityDT30_00App}
    &\, \Biggl\{1+(\alpha-1) &&D_\mtn{\alpha}\Bigl(\rho^i_\mtn{C}\|\rho^\mtn{\beta}_\mtn{C}\Bigr)\Biggr\}&&\times\Biggl[D_\mtn{\alpha}\Bigl(\rho^i_\mtn{S}\otimes\rho^i_\mtn{Q}\|\rho^\mtn{\beta}_\mtn{S}\otimes\rho^\mtn{\beta}_\mtn{Q}\Bigr)&&-D_\mtn{\alpha}\Bigl(\rho^{f}_\mtn{S}\otimes\rho^{f}_\mtn{Q}\|\rho^\mtn{\beta}_\mtn{S}\otimes\rho^\mtn{\beta}_\mtn{Q})\Biggr]\ge 0.\\
    \alpha\in[0,+\infty)&: &&>0 &&\Ra &&\ge0\nonumber
\end{empheq}
The expression in square brackets can be written for $\alpha\ge0$, the braced prefactor being strictly positive there since $\rho^\mtn{\beta}_\mtn{C}$ is of full rank, as 
\begin{align}\label{eq:ProofREqT1}
D_\mtn{\alpha}\Bigl(\rho^i_\mtn{S}\|\rho^\mtn{\beta}_\mtn{S}\Bigr)\Bigl[1+(\alpha-1)D_\mtn{\alpha}\Bigl(\rho^i_\mtn{Q}\|\rho^\mtn{\beta}_\mtn{Q}\Bigr)\Bigr]+ D_\mtn{\alpha}\Bigl(\rho^i_\mtn{Q}\|\rho^\mtn{\beta}_\mtn{Q}\Bigr) \ge{}&\nonumber\\
    D_\mtn{\alpha}\Bigl(\rho^f_\mtn{S}\|\rho^\mtn{\beta}_\mtn{S}\Bigr)\Bigl[1+(\alpha-1)D_\mtn{\alpha}\Bigl(\rho^f_\mtn{Q}\|\rho^\mtn{\beta}_\mtn{Q}\Bigr)\Bigr]+ D_\mtn{\alpha}\Bigl(\rho^f_\mtn{Q}\|\rho^\mtn{\beta}_\mtn{Q}\Bigr).&
\end{align}
Substituting qubit states
\begin{align}\label{eq:ProofREqT2}
D_\mtn{\alpha}\Bigl(\rho^i_\mtn{S}\|\rho^\mtn{\beta}_\mtn{S}\Bigr)\Bigl[1+e^{-\beta\hbar\omega}\Bigr]^{\alpha-1}+ \frac{1}{\alpha-1}\left(\Bigl[1+e^{-\beta\hbar\omega}\Bigr]^{\alpha-1}-1\right)\ge{}&\nonumber\\
D_\mtn{\alpha}\Bigl(\rho^f_\mtn{S}\|\rho^\mtn{\beta}_\mtn{S}\Bigr)
\Bigl[e^{\beta\hbar\omega}+1\Bigr]^{\alpha-1}+\frac{1}{\alpha-1}\left(\Bigl[e^{\beta\hbar\omega}+1\Bigr]^{\alpha-1}-1\right),&
\end{align}
and rearranging gives
\begin{equation}\label{eq:ProofREqT3}
    \frac{(\alpha-1)D_\mtn{\alpha}\Bigl(\rho^i_\mtn{S}\|\rho^\mtn{\beta}_\mtn{S}\Bigr)+1}{(\alpha-1)D_\mtn{\alpha}\Bigl(\rho^f_\mtn{S}\|\rho^\mtn{\beta}_\mtn{S}\Bigr)+1}\ge e^{\beta\hbar\omega(\alpha-1)}.
\end{equation}
Using \eqref{eq:RenyiDTsallisD2}, we can write \eqref{eq:ProofREqT3} as 
\begin{widetext}
\begin{equation}\label{eq:ProofREqT4}
    D^\mtn{R}_\mtn{\alpha}\Bigl(\rho^i_\mtn{S}\|\rho^\mtn{\beta}_\mtn{S}\Bigr)-D^\mtn{R}_\mtn{\alpha}\Bigl(\rho^f_\mtn{S}\|\rho^\mtn{\beta}_\mtn{S}\Bigr)\ge \beta\hbar\omega\equiv\beta\Delta E_\mtn{Q},\qquad\forall\alpha\in[0,+\infty).
\end{equation}
\end{widetext}
for condition \eqref{eq:SystemtQubitCond_D}. It is then straightforward to obtain 
\begin{widetext}
\begin{equation}\label{eq:ProofREqT5}
    D^\mtn{R}_\mtn{\alpha}\Bigl(\rho^f_\mtn{C}\|\rho^\mtn{\beta}_\mtn{C}\Bigr)-D^\mtn{R}_\mtn{\alpha}\Bigl(\rho^i_\mtn{C}\|\rho^\mtn{\beta}_\mtn{C}\Bigr)\ge \beta\hbar\omega\equiv\beta\Delta E_\mtn{Q},\qquad\forall\alpha\in[0,+\infty).
\end{equation}
\end{widetext}
for condition \eqref{eq:CatalystQubitCond_D}.

This is rather surprising, namely, the conditions for the transitions \eqref{eq:RhoSRhoC} and \eqref{eq:CatalystTransition} to occur, is the same whether one adopts the additive R\'enyi $D_\mtn{\alpha}^\mtn{R}$ or the non-additive $D_\mtn{\alpha}$ divergences.  
That is, 
\begin{widetext}
\begin{equation}\label{eq:WBitEnergyDiff0}
\Bigl[D^\mtn{R}_\mtn{\alpha}(\rho^{input}\|\rho^\mtn{\beta})-D^\mtn{R}_\mtn{\alpha}(\rho^{output}\|\rho^\mtn{\beta})\Bigr]\ge\beta\Delta E_\mtn{Q},\qquad\forall\alpha\in[0,+\infty).
\end{equation}
\end{widetext}

Note that in inequality \eqref{eq:CatalystQubitCond_D}, we have $\rho^{input}=\rho^f_\mtn{C}$ and $\rho^{output}=\rho^i_\mtn{C}$, while $\rho^{input}=\rho^i_\mtn{S}$ and $\rho^{output}=\rho^f_\mtn{S}$ in inequality \eqref{eq:SystemtQubitCond_D}. \\

Now, consider \eqref{eq:WBitEnergyDiff0} as
\begin{equation}
\label{eq:CatalystQubitDR2}  
    D^\mtn{R}_\mtn{\alpha}(\rho^f_\mtn{C}\|\rho^\mtn{\beta}_\mtn{C})\ge
    D^\mtn{R}_\mtn{\alpha}(\rho^i_\mtn{C}\|\rho^\mtn{\beta}_\mtn{C})+\beta\Delta E_\mtn{Q},\qquad\forall\alpha\in[0,+\infty).
\end{equation}

$\bullet$ If $\Delta E_\mtn{Q}>0$, then energy is transferred from the system to the qubit, and
\begin{widetext}
\begin{equation}\label{eq:WBitEnergyDiff1}
    \Bigl[D^\mtn{R}_\mtn{\alpha}(\rho^{input}\|\rho^\mtn{\beta})-D^\mtn{R}_\mtn{\alpha}(\rho^{output}\|\rho^\mtn{\beta})\Bigr]\ge\beta\Delta E_\mtn{Q}>0,\qquad\forall\alpha\in[0,+\infty).
\end{equation}
\end{widetext}

Then, the maximum amount of energy extractable from a system in contact with a heat reservoir at temperature $T$ in the process $\rho^i_\mtn{S}\otimes|0\rangle\langle0|\longmapsto\rho^f_\mtn{S}\otimes|\omega\rangle\langle\omega|$ is given by 

\begin{widetext}
\begin{equation}\label{eq:WBitEnergyDiffP}
    k_\mtn{B}T\,\inf_{\alpha\ge0}[D^\mtn{R}_\mtn{\alpha}(\rho^{input}\|\rho^\mtn{\beta})-D^\mtn{R}_\mtn{\alpha}(\rho^{output}\|\rho^\mtn{\beta})]\equiv W_{\text{extract}}(\rho_\mtn{S}^{input}\mapsto\rho_\mtn{S}^{output}).
\end{equation}
\end{widetext}

If we take $\rho_\mtn{S}^f=\rho_\mtn{S}^\mtn{\beta}$, that is, $\rho_\mtn{S}^{output}=\rho_\mtn{S}^\mtn{\beta}$, then $W_{\text{extract}}$ would yield \textit{the work distance} from $\rho^i_\mtn{S}$ to $\rho_\mtn{S}^\mtn{\beta}$:

\begin{equation}\label{eq:WorkDistOut0}
\rho_\mtn{S}^i\otimes|0\rangle\langle0|\xmapsto{\quad W_{\text{extract}}\quad}\rho_\mtn{S}^\mtn{\beta}\otimes|\omega\rangle\langle\omega|
\end{equation}

\begin{equation}\label{eq:WorkDistOut}
    W_{\text{extract}}(\rho_\mtn{S}^i\mapsto\rho_\mtn{S}^\mtn{\beta})=\mathcal{D}_{work}(\rho_\mtn{S}^i\succ\rho_\mtn{S}^\mtn{\beta}).
\end{equation}

$\bullet$ If $\Delta E_\mtn{Q}<0$, then energy is transferred from qubit to the system, and 
\begin{widetext}
\begin{equation}\label{eq:WBitEnergyDiff2}
    0<-\Bigl[D^\mtn{R}_\mtn{\alpha}(\rho^{input}\|\rho^\mtn{\beta})-D^\mtn{R}_\mtn{\alpha}(\rho^{output}\|\rho^\mtn{\beta})\Bigr]\le-\beta\Delta E_\mtn{Q},\qquad\forall\alpha\in[0,+\infty).
\end{equation}
\end{widetext}

Then, the minimum amount of energy to form a state $\rho^f_\mtn{S}$ of a system in contact with a heat reservoir at temperature $T$ in the process $\rho_\mtn{S}^i\otimes|0\rangle\langle0|\longmapsto\rho_\mtn{S}^f\otimes|\omega\rangle\langle\omega|$ reads

\begin{widetext}
\begin{equation}\label{eq:WBitEnergyDiffN}
    k_\mtn{B}T\,\sup_{\alpha\ge0}[D^\mtn{R}_\mtn{\alpha}(\rho^{output}\|\rho^\mtn{\beta})-D^\mtn{R}_\mtn{\alpha}(\rho^{input}\|\rho^\mtn{\beta})]\equiv W_{\text{cost}}(\rho_\mtn{S}^{input}\mapsto\rho_\mtn{S}^{output}).
\end{equation}
\end{widetext}

If we take $\rho^i_\mtn{S}=\rho_\mtn{S}^\mtn{\beta}$, that is, $\rho_\mtn{S}^{input}=\rho_\mtn{S}^\mtn{\beta}$, then $W_{\text{cost}}$ would yield the work distance from $\rho_\mtn{S}^\mtn{\beta}$ to $\rho_\mtn{S}^f$:

\begin{equation}\label{eq:WorkDistIn0}
\rho_\mtn{S}^\mtn{\beta}\otimes|\omega\rangle\langle\omega|\xmapsto{\quad W_{\text{cost}}\quad}\rho_\mtn{S}^f\otimes|0\rangle\langle0|
\end{equation}

\begin{equation}\label{eq:WorkDistIn}
    W_{\text{cost}}(\rho_\mtn{S}^\mtn{\beta}\mapsto\rho_\mtn{S}^f)=-\mathcal{D}_{\text{work}}(\rho_\mtn{S}^\mtn{\beta}\succ\rho_\mtn{S}^f).
\end{equation}

Considering the inequality \eqref{eq:CatalystQubitDR2}, we see that for $\Delta E_\mtn{Q}>0$, $C$ transfers the excess energy to $Q$, and for $\Delta E_\mtn{Q}<0$, $C$ takes from $Q$ the energy that it has lost. That is, the transition \eqref{eq:CatalystTransition} occurs, and the $2^{nd}$ laws of quantum thermodynamics would be recovered for the transitions with inexact catalysis.

\renewcommand{\theequation}{H\arabic{equation}}
\setcounter{equation}{0}
\appsection{Relation to order-theoretic characterizations of catalytic transformations}  \label{appCC}

The purpose of this appendix is to place the non-additive second-law inequalities developed in this work within the general order-theoretic framework of catalytic transformations, as formalized in the resource-theoretic literature.

\begin{corollary}
Let $m,n\in\mathbb{N}$, $p,q\in\mathcal{I}_n$ and $p',q'\in\mathcal{I}_m$ with the assumption that at least one of $p$ or $q$ lies in the interior of $\mathcal{I}_n$\footnote{
This technical assumption follows the order-theoretic framework for catalytic convertibility developed by Gour; see Ref.~\cite{Gour2021}, in particular the characterization underlying Theorem~21.} (i.e.\ has full support). Then the following statements are equivalent:
\begin{itemize}
    \item[(a)] 
    $(p,q)$ catalytically relatively majorizes $(p',q')$, denoted by 
    \begin{equation}\label{eq:CatMajor}
      (p,q) \succeq_c (p',q')
    \end{equation}
  
    \item[(b)]
    $
    D_\mtn{\alpha}(p\|q)\ge D_\mtn{\alpha}(p'\|q')
    $
    and
    $
    D_\mtn{\alpha}(q\|p)\ge D_\mtn{\alpha}(q'\|p')
    $
    for every $\alpha\ge \tfrac12$.
\end{itemize}
\end{corollary}

\begin{proof}
By Theorem~21 of Ref.~\cite{Gour2021}, relative majorization satisfies
\begin{widetext}
\begin{equation}
\label{eq:GourRenyiChar}
(p,q) \succeq_c (p',q')
\quad\Longleftrightarrow\quad
\begin{cases}
D_\mtn{\alpha}^{\mtn{R}}(p\|q) \ge D_\mtn{\alpha}^{\mtn{R}}(p'\|q'),\\[4pt]
D_\mtn{\alpha}^{\mtn{R}}(q\|p) \ge D_\mtn{\alpha}^{\mtn{R}}(q'\|p'),
\end{cases}
\qquad \text{for every } \alpha \ge \tfrac12,
\end{equation}
\end{widetext}
where $D_\mtn{\alpha}^{\mtn{R}}$ denotes the R\'enyi divergence of order~$\alpha$.\\
For each fixed $\alpha\neq 1$, non-additive and the R\'enyi divergences are related by
Eqs.~\eqref{eq:RenyiDTsallisD1} and~\eqref{eq:RenyiDTsallisD2}, namely
\begin{equation}
\label{eq:TsallisRenyiMonotone}
D_\mtn{\alpha}(p\|q)
=
\frac{e^{(\alpha-1)D_\mtn{\alpha}^{\mtn{R}}(p\|q)} - 1}{\alpha - 1},
\end{equation}
which is a strictly increasing function of $D_\mtn{\alpha}^{\mtn{R}}(p\|q)$ for $\alpha\neq 1$.
Hence, for any two pairs $(p,q)$ and $(p',q')$,
\begin{equation}
\label{eq:RenyiTsallisOrderEquiv}
D_\mtn{\alpha}^{\mtn{R}}(p\|q) \ge D_\mtn{\alpha}^{\mtn{R}}(p'\|q')
\quad\Longleftrightarrow\quad
D_\mtn{\alpha}(p\|q) \ge D_\mtn{\alpha}(p'\|q').
\end{equation}
Applying this equivalence pointwise in $\alpha$ to the R\'enyi characterization
in Eq.~\eqref{eq:GourRenyiChar} yields
\begin{widetext}
\begin{equation}
\label{eq:TsallisChar}
(p,q) \succeq_c (p',q')
\quad\Longleftrightarrow\quad
\begin{cases}
D_\mtn{\alpha}(p\|q) \ge D_\mtn{\alpha}(p'\|q'),\\[4pt]
D_\mtn{\alpha}(q\|p) \ge D_\mtn{\alpha}(q'\|p'),
\end{cases}
\qquad \text{for every } \alpha\ge \tfrac12,
\end{equation}
\end{widetext}
which proves the claim.
\end{proof}

It is worth noting that Refs.~\cite{Farooq2024,Verhagen2025} provide a strictly stronger derivation of Theorem~21 in \cite{Gour2021}, obtained within a general matrix majorization framework. In that setting, relative majorization emerges as a special case of a broader operator ordering, and the characterization holds without auxiliary assumptions or regularization procedures on $p'$ and $q'$. While this matrix-based approach offers a deeper structural understanding of the ordering, its technical scope and level of abstraction go beyond the needs of the present work. For our purposes, the R\'enyi-based formulation of \cite{Gour2021} provides a natural and sufficiently direct route to the non-additive characterization established above.

\renewcommand{\theequation}{I\arabic{equation}}
\setcounter{equation}{0}  
\appsection{Pseudo-additivity and finite-size catalytic constraints}  \label{app:pseudo-additivity-finite-size}

This appendix develops the finite-size content of the pseudo-additive framework underlying Sec.~\ref{sec::PseudoAdd}. Sec.~\ref{app:exact-uncorr-cat} records the exact factorization of the joint free-energy change for uncorrelated catalysis; Sec.~\ref{app:approx-uncorr-cat} derives continuity bounds and the transition-dependent trade-off when the catalyst is returned only approximately; Sec.~\ref{app:finite-benchmark-examples} works out benchmark examples in the trivial-Hamiltonian setting; and Sec.~\ref{app:gapped-catalyst} lifts that benchmark and establishes the gapped-catalyst identity of Eq.~\eqref{eq:DeltaF-general-identity}.

\subsection{Exact uncorrelated catalysis} \label{app:exact-uncorr-cat}

For uncorrelated catalytic transformations of the form
\begin{equation}
  \rho_\mtn{S} \otimes \sigma_\mtn{M} \;\longmapsto\; \rho'_\mtn{S} \otimes \sigma_\mtn{M} ,
\end{equation}
the change in total non-additive free energy takes the factorized form of Eq.~\eqref{eq:DeltaF-exact-factorized}:
\[
  \Delta F_\mtn{\alpha}
  =
  \Delta F_\mtn{\alpha}^{\mtn{(S)}}
  \Bigl[
    1 + \mathrm{sgn}(\alpha)(\alpha-1)D_\mtn{\alpha}(\sigma_\mtn{M}\|\gamma_\mtn{M})
  \Bigr],
\]
where $\Delta F_\mtn{\alpha}^{\mtn{(S)}}=F_\mtn{\alpha}(\rho_\mtn{S}')-F_\mtn{\alpha}(\rho_\mtn{S})$ is the free-energy change of the system alone.
Unlike additive formulations, the catalyst contribution does not cancel identically even in the uncorrelated case. Instead, the catalyst's distance from equilibrium enters multiplicatively, making its thermodynamic properties operationally relevant.

For $0<\alpha<1$, assuming $\Delta F_\mtn{\alpha}^{\mtn{(S)}}\le 0$, monotonicity implies
\begin{equation}
  1-(1-\alpha)D_\mtn{\alpha}(\sigma_\mtn{M}\|\gamma_\mtn{M})\ge 0,
\end{equation}
and hence
\begin{equation}
  D_\mtn{\alpha}(\sigma_\mtn{M}\|\gamma_\mtn{M})\le \frac{1}{1-\alpha}.
\end{equation}
This inequality should not be overinterpreted as a sharp lower bound on the catalyst dimension by itself. Rather, it shows that for $0<\alpha<1$ the pseudo-additive correction prevents arbitrarily large catalytic athermality from becoming completely irrelevant: the catalyst contribution remains explicitly visible in the monotonicity condition, in contrast to additive R\'enyi-based formulations.

To clarify the scope of this observation, let us consider the simplest case in which the catalytic Hamiltonian is trivial,
\begin{equation}
  H_\mtn{M}=0,
\end{equation}
so that the thermal state is maximally mixed,
\begin{equation}
  \gamma_\mtn{M} = \mathrm{diag}(\vec{\mu}_{d_\mtn{M}}),
  \qquad
  \vec{\mu}_{d_\mtn{M}} = (1/d_\mtn{M},\dots,1/d_\mtn{M}).
\end{equation}
For a diagonal catalyst $\sigma_\mtn{M}=\mathrm{diag}(\vec p)$, the above condition reduces to
\begin{equation}
  d_\mtn{M}^{\alpha-1}\sum_i p_i^\alpha \ge 0,
\end{equation}
which is trivially satisfied for any probability distribution $\vec p$. Thus, for exact uncorrelated catalysis with a trivial catalyst Hamiltonian, no nontrivial finite-size restriction follows from this bound alone. This is precisely why the approximate setting is more informative: once the catalyst is returned only up to a controlled error, the pseudo-additive structure acquires nontrivial finite-size content.

\subsection{Approximate uncorrelated catalysis} \label{app:approx-uncorr-cat}

Assume now that the catalyst is returned only approximately. Denoting by $\sigma_\mtn{M}$ and $\sigma'_\mtn{M}$ the initial and final states of $M$, respectively, we impose the trace-distance constraint
\begin{equation}
    \frac{1}{2}\|\sigma'_\mtn{M} - \sigma_\mtn{M}\|_1 \le \varepsilon,
    \qquad
    \text{equivalently } \|\vec{q} - \vec{p} \|_1 \le 2\varepsilon,
\label{eq:trace-distance}
\end{equation}
where $\sigma'_\mtn{M} = \mathrm{diag}(\vec{q})$, $\sigma_\mtn{M} = \mathrm{diag}(\vec{p})$, and $\varepsilon \ge 0$ quantifies the degree of catalytic approximation. For uncorrelated catalytic transformations of the form
\begin{equation}
  \rho_\mtn{S} \otimes \sigma_\mtn{M} \longmapsto \rho^\prime_\mtn{S} \otimes \sigma^\prime_\mtn{M},
\end{equation}
the change in total non-additive free energy can be written as
\begin{widetext}
\begin{equation}
    \Delta F_\mtn{\alpha} = \Delta F_\mtn{\alpha}^{\mtn{(S)}} + k_BT\!\left(1+(\alpha-1)\ln Z_\mtn{S}\right)(x-y) + (\alpha-1)\!\left[x\,F_\mtn{\alpha}(\rho_\mtn{S}')-y\,F_\mtn{\alpha}(\rho_\mtn{S})\right],
\label{eq:DeltaFalpha-start}
\end{equation}
\end{widetext}
where the auxiliary quantities $x$ and $y$ quantify the non-additive athermalities (deviation from thermal equilibrium $\gamma_\mtn{M}$) of the final and initial catalyst states
\begin{widetext}
\begin{equation}
    x \coloneqq D_\mtn{\alpha}(\sigma^\prime_\mtn{M} \| \gamma_\mtn{M}) = \frac{d_\mtn{M}^{\alpha-1}}{\alpha-1}\sum_i q_i^\alpha - \frac{1}{\alpha-1},
    \qquad
    y \coloneqq D_\mtn{\alpha}(\sigma_\mtn{M} \| \gamma_\mtn{M}) = \frac{d_\mtn{M}^{\alpha-1}}{\alpha-1}\sum_i p_i^\alpha - \frac{1}{\alpha-1}.
\label{eq:x-y-def}
\end{equation}
\end{widetext}
The first equalities in Eq.~\eqref{eq:x-y-def}, defining $x$ and $y$ as the divergences of the catalyst states from the reference, are general; the second, evaluated forms, and everything that follows in this appendix, hold in the trivial-Hamiltonian benchmark introduced in Sec.~\ref{app:exact-uncorr-cat}, in which $\gamma_\mtn{M}$ is maximally mixed on a $d_\mtn{M}$-dimensional space.

In terms of the spectral power sums $P_\mtn{\alpha}$, $Q_\mtn{\alpha}$, and the shorthand $A_\mtn{\alpha}$ defined in Eq.~\eqref{eq:PQA-shorthand},
\[
Q_\mtn{\alpha} \coloneqq \sum_i q_i^\alpha,
\qquad
P_\mtn{\alpha} \coloneqq \sum_i p_i^\alpha,
\qquad
A_\mtn{\alpha} \coloneqq \frac{k_BT\!\left(1+(\alpha-1)\ln Z_\mtn{S}\right)}{\alpha-1},
\]
Eq.~\eqref{eq:DeltaFalpha-start} becomes
\begin{equation}
\Delta F_\mtn{\alpha}
=
d_\mtn{M}^{\alpha-1}
\left[
A_\mtn{\alpha} (Q_\mtn{\alpha}-P_\mtn{\alpha})
+
Q_\mtn{\alpha} F_\mtn{\alpha}(\rho^\prime_\mtn{S})
-
P_\mtn{\alpha} F_\mtn{\alpha}(\rho_\mtn{S})
\right] ,
\label{eq:DeltaFalpha-expanded}
\end{equation}
where the last two terms can be rearranged as
\begin{equation}
Q_\mtn{\alpha} F_\mtn{\alpha}(\rho^\prime_\mtn{S})-P_\mtn{\alpha} F_\mtn{\alpha}(\rho_\mtn{S})
=
P_\mtn{\alpha} \, \Delta F_\mtn{\alpha}^{\mtn{(S)}}
+
(Q_\mtn{\alpha}-P_\mtn{\alpha})F_\mtn{\alpha}(\rho^\prime_\mtn{S}),
\end{equation}
yielding
\begin{equation}
\Delta F_\mtn{\alpha}
=
d_\mtn{M}^{\alpha-1}
\left[
P_\mtn{\alpha} \, \Delta F_\mtn{\alpha}^{\mtn{(S)}}
+
(Q_\mtn{\alpha}-P_\mtn{\alpha})\!\left(F_\mtn{\alpha}(\rho^\prime_\mtn{S})+A_\mtn{\alpha}\right)
\right].
\label{eq:DeltaFalpha-limit-form}
\end{equation}

\subsubsection{Continuity bounds}

\paragraph*{Case \texorpdfstring{\(\alpha \ge 1\)}{α ≥ 1}.}
For $\alpha\ge 1$, the function $t^\alpha$ is Lipschitz continuous on $[0,1]$, implying
\begin{equation}
|Q_\mtn{\alpha}-P_\mtn{\alpha}|
=
\Big|\sum_i(q_i^\alpha-p_i^\alpha)\Big|
\le
\alpha \sum_i |q_i-p_i|
\le
2\alpha\,\varepsilon.
\label{eq:Qalpha-bound-alpha-ge-1}
\end{equation}
Hence, the correction proportional to $Q_\mtn{\alpha}-P_\mtn{\alpha}$ in Eq.~\eqref{eq:DeltaFalpha-limit-form} is at most $O(\varepsilon)$.

\paragraph*{Case \texorpdfstring{\(0 < \alpha < 1\)}{0 < α < 1}.}
For $0<\alpha<1$, the function $t^\alpha$ is H\"older continuous, leading to
\begin{equation}
|Q_\mtn{\alpha}-P_\mtn{\alpha}|
\le
\sum_i |q_i^\alpha-p_i^\alpha|
\le
\sum_i |q_i-p_i|^\alpha
\le
d_\mtn{M}^{1-\alpha}(2\varepsilon)^\alpha.
\label{eq:Qalpha-bound-alpha-lt-1}
\end{equation}
In this regime, the correction scales as $O(\varepsilon^\alpha)$.

\paragraph*{Limit \texorpdfstring{\(\varepsilon\to 0\)}{ε → 0}.}
In both regimes, $Q_\mtn{\alpha}-P_\mtn{\alpha}\to 0$ as $\varepsilon\to 0$. Therefore,
\begin{equation}
\lim_{\varepsilon\to 0}\Delta F_\mtn{\alpha}
=
d_\mtn{M}^{\alpha-1} P_\mtn{\alpha} \, \Delta F_\mtn{\alpha}^{\mtn{(S)}} ,
\label{eq:DeltaFalpha-limit}
\end{equation}
which is consistent with the exact-return expression.

\subsubsection{A transition-dependent trade-off}

Eq.~\eqref{eq:DeltaFalpha-limit-form} yields an explicit finite-size trade-off between the target transformation on $S$ and the allowed deviation of the catalyst. Imposing the monotonicity condition $\Delta F_\mtn{\alpha} \le 0$ and rearranging gives
\begin{equation}
P_\mtn{\alpha} \,\Delta F_\mtn{\alpha}^{\mtn{(S)}}
\;\le\;
-(Q_\mtn{\alpha}-P_\mtn{\alpha})\bigl(F_\mtn{\alpha}(\rho_\mtn{S}')+A_\mtn{\alpha}\bigr).
\label{eq:tradeoff-start}
\end{equation}
Using $-(Q_\mtn{\alpha}-P_\mtn{\alpha})\le |Q_\mtn{\alpha}-P_\mtn{\alpha}|$, one obtains the sufficient condition of Eq.~\eqref{eq:tradeoff-sufficient},
\[
P_\mtn{\alpha}\,\bigl(-\Delta F_\mtn{\alpha}^{\mtn{(S)}}\bigr)
\;\ge\;
|Q_\mtn{\alpha}-P_\mtn{\alpha}|\;\bigl|F_\mtn{\alpha}(\rho_\mtn{S}')+A_\mtn{\alpha}\bigr|.
\]

\paragraph*{Explicit \(\varepsilon\)-bounds.}
For $\alpha \ge 1$, Eq.~\eqref{eq:Qalpha-bound-alpha-ge-1} implies $|Q_\mtn{\alpha}-P_\mtn{\alpha}|\le 2\alpha\varepsilon$,
and Eq.~\eqref{eq:tradeoff-sufficient} yields the bound of Eq.~\eqref{eq:eps-bound-alpha-ge-1},
\[
\varepsilon
\;\le\;
\frac{P_\mtn{\alpha}\,\bigl(-\Delta F_\mtn{\alpha}^{\mtn{(S)}}\bigr)}{2\alpha\,\bigl|F_\mtn{\alpha}(\rho_\mtn{S}')+A_\mtn{\alpha}\bigr|}.
\]

For $0 < \alpha < 1$, Eq.~\eqref{eq:Qalpha-bound-alpha-lt-1} gives
$|Q_\mtn{\alpha}-P_\mtn{\alpha}|\le d_\mtn{M}^{1-\alpha}(2\varepsilon)^\alpha$, which leads to the bound of Eq.~\eqref{eq:eps-bound-alpha-lt-1},
\[
\varepsilon
\;\le\;
\frac{1}{2}\left[
\frac{P_\mtn{\alpha}\,\bigl(-\Delta F_\mtn{\alpha}^{\mtn{(S)}}\bigr)}{d_\mtn{M}^{1-\alpha}\,\bigl|F_\mtn{\alpha}(\rho_\mtn{S}')+A_\mtn{\alpha}\bigr|}
\right]^{\!1/\alpha}.
\]

\subsection{Simple finite-dimensional benchmark examples} \label{app:finite-benchmark-examples}

To make this finite-size mechanism explicit, let us consider the simplest benchmark in which the catalyst Hamiltonian is trivial,
\begin{equation}
  H_\mtn{M}=0,
\end{equation}
so that the reference thermal state is uniform:
\begin{equation}
  \gamma_\mtn{M} = \mathrm{diag}(1/d_\mtn{M},\dots,1/d_\mtn{M}).
\end{equation}
Instead of taking the \textit{initial} catalyst to be uniform, it is often more natural in approximate catalysis to regard the \textit{returned} catalyst as the reference object and impose
\begin{equation}
  q_i = \frac{1}{d_\mtn{M}}.
\end{equation}
Then
\begin{equation}
  Q_\mtn{\alpha} = \sum_i q_i^\alpha = d_\mtn{M}^{1-\alpha},
\end{equation}
and all finite-size structure is encoded in the initial profile $\vec p$ through $P_\mtn{\alpha}$ and $Q_\mtn{\alpha}-P_\mtn{\alpha}$.

\paragraph*{Example 1: distributed perturbation around a uniform returned catalyst.}

Assume $d_\mtn{M}$ is even and define
\begin{equation}
  q_i = \frac{1}{d_\mtn{M}},
\end{equation}
and
\begin{equation}
  p_i =
  \begin{cases}
    \dfrac{1}{d_\mtn{M}} + \dfrac{2\varepsilon}{d_\mtn{M}}, & i\in \mathcal E, \\[6pt]
    \dfrac{1}{d_\mtn{M}} - \dfrac{2\varepsilon}{d_\mtn{M}}, & i\in \mathcal O,
  \end{cases}
\end{equation}
where $|\mathcal E|=|\mathcal O|=d_\mtn{M}/2$. Then
\begin{equation}
  \sum_i p_i = 1,
  \qquad
  \frac12\|\vec q-\vec p\|_1 = \varepsilon,
\end{equation}
so the approximate-return condition is exactly saturated.

In this case,
\begin{equation} \label{eq:P1}
  P_\mtn{\alpha}
  =
  d_\mtn{M}^{1-\alpha}
  \frac{(1+2\varepsilon)^\alpha + (1-2\varepsilon)^\alpha}{2},
\end{equation}
and therefore
\begin{equation}
  Q_\mtn{\alpha} - P_\mtn{\alpha}
  =
  d_\mtn{M}^{1-\alpha}
  \left[
    1 -
    \frac{(1+2\varepsilon)^\alpha + (1-2\varepsilon)^\alpha}{2}
  \right].
\end{equation}
For small $\varepsilon$, a Taylor expansion yields
\begin{equation}\label{eq:QmPq1}
  Q_\mtn{\alpha} - P_\mtn{\alpha}
  \approx
  -2\alpha(\alpha-1)\, d_\mtn{M}^{1-\alpha}\varepsilon^2.
\end{equation}
Substituting the small-$\varepsilon$ expansion in Eq.~\eqref{eq:QmPq1} into Eq.~\eqref{eq:DeltaFalpha-limit-form}, we obtain, to leading nontrivial order in $\varepsilon$,
\begin{equation}
\Delta F_\mtn{\alpha}
\approx
d_\mtn{M}^{\alpha-1}P_\mtn{\alpha} \,\Delta F_\mtn{\alpha}^{\mtn{(S)}}
-
2\alpha(\alpha-1)\varepsilon^2\bigl(F_\mtn{\alpha}(\rho_\mtn{S}')+A_\mtn{\alpha}\bigr).
\label{eq:I32_sub_I15a}
\end{equation}
Using Eq.~\eqref{eq:P1}, one further has
\begin{equation}
d_\mtn{M}^{\alpha-1}P_\mtn{\alpha}
=
\frac{(1+2\varepsilon)^\alpha+(1-2\varepsilon)^\alpha}{2}
\approx
1+2\alpha(\alpha-1)\varepsilon^2,
\label{eq:I32_sub_I15b}
\end{equation}
so that we arrive at Eq.~\eqref{eq:I32_sub_I15c},
\[
\Delta F_\mtn{\alpha}
\approx
\Delta F_\mtn{\alpha}^{\mtn{(S)}}
-
2\,\varepsilon^2\,\alpha(\alpha-1)\bigl(F_\mtn{\alpha}(\rho_\mtn{S})+A_\mtn{\alpha}\bigr).
\]

Thus, within this particular approximate-return construction, the leading pseudo-additive correction is quadratic in $\varepsilon$ and remains finite even for large $d_\mtn{M}$. In particular, increasing the catalyst dimension alone does not suppress the correction at this order. For $\alpha>1$, this means that any reduced-state violation $\Delta F_\mtn{\alpha}^{\mtn{(S)}}>0$ that can be compensated within this ansatz is controlled at leading order by an $O(\varepsilon^2)$ correction.

\paragraph*{Example 2: concentrated two-level perturbation around a uniform returned catalyst.}

A sharper benchmark is obtained by concentrating the same trace-distance error into only two levels:
\begin{equation}
  q_i = \frac{1}{d_\mtn{M}},
\end{equation}
and
\begin{equation}
  p_1 = \frac{1}{d_\mtn{M}}+\varepsilon,
  \qquad
  p_2 = \frac{1}{d_\mtn{M}}-\varepsilon,
  \qquad
  p_j = \frac{1}{d_\mtn{M}}\quad (j\ge 3),
\end{equation}
with positivity requiring $\varepsilon\le 1/d_\mtn{M}$. Again,
\begin{equation}
  \sum_i p_i = 1,
  \qquad
  \frac12\|\vec q-\vec p\|_1 = \varepsilon.
\end{equation}
Now
\begin{equation}\label{eq:P2}
  P_\mtn{\alpha}
  =
  \left(\frac{1}{d_\mtn{M}}+\varepsilon\right)^\alpha
  +
  \left(\frac{1}{d_\mtn{M}}-\varepsilon\right)^\alpha
  +
  (d_\mtn{M}-2)\left(\frac{1}{d_\mtn{M}}\right)^\alpha,
\end{equation}
so that
\begin{equation}
  Q_\mtn{\alpha} - P_\mtn{\alpha}
  =
  2\left(\frac{1}{d_\mtn{M}}\right)^\alpha
  -
  \left(\frac{1}{d_\mtn{M}}+\varepsilon\right)^\alpha
  -
  \left(\frac{1}{d_\mtn{M}}-\varepsilon\right)^\alpha.
\end{equation}
Since $t^\alpha$ is concave for $0<\alpha<1$ and convex for $\alpha>1$, it follows that
\begin{equation}
  Q_\mtn{\alpha}-P_\mtn{\alpha}
  \ge 0 \quad (0<\alpha<1),
  \qquad
  Q_\mtn{\alpha}-P_\mtn{\alpha}
  \le 0 \quad (\alpha>1).
\end{equation}
For small $\varepsilon$, setting $a=1/d_\mtn{M}$, one finds
\begin{equation} \label{eq:QmPq2}
  Q_\mtn{\alpha}-P_\mtn{\alpha}
  \approx
  -\alpha(\alpha-1)a^{\alpha-2}\varepsilon^2
  =
  -\alpha(\alpha-1)d_\mtn{M}^{2-\alpha}\varepsilon^2.
\end{equation}
Substituting Eq.~\eqref{eq:QmPq2} into Eq.~\eqref{eq:DeltaFalpha-limit-form}, we obtain, to leading nontrivial order in $\varepsilon$,
\begin{equation}
\Delta F_\mtn{\alpha}
\approx
d_\mtn{M}^{\alpha-1}P_\mtn{\alpha} \,\Delta F_\mtn{\alpha}^{\mtn{(S)}}
-
d_\mtn{M}\,\varepsilon^2\,\alpha(\alpha-1)\bigl(F_\mtn{\alpha}(\rho_\mtn{S}')+A_\mtn{\alpha}\bigr).
\label{eq:I43_sub_I15a}
\end{equation}
Moreover, expanding Eq.~\eqref{eq:P2} for small $\varepsilon$ gives
\begin{equation}
d_\mtn{M}^{\alpha-1}P_\mtn{\alpha}
\approx
1+d_\mtn{M}\,\varepsilon^2\,\alpha(\alpha-1),
\label{eq:I43_sub_I15b}
\end{equation}
so that we arrive at Eq.~\eqref{eq:I43_sub_I15c},
\[
\Delta F_\mtn{\alpha}
\approx
\Delta F_\mtn{\alpha}^{\mtn{(S)}}
-
d_\mtn{M}\,\varepsilon^2\,\alpha(\alpha-1)\bigl(F_\mtn{\alpha}(\rho_\mtn{S})+A_\mtn{\alpha}\bigr).
\]

Thus, in this concentrated two-level perturbation ansatz, the leading pseudo-additive correction scales as $d_\mtn{M}\varepsilon^2$, in contrast with the $2\varepsilon^2$ scaling of the distributed perturbation considered above. This shows that, within the same fixed-$d_\mtn{M}$ comparison, concentrating the catalytic deviation into a small number of levels produces a parametrically stronger leading-order correction than distributing it uniformly. At the same time, the positivity condition $\varepsilon \le 1/d_\mtn{M}$---equivalently, $d_\mtn{M}\varepsilon^2 \le \varepsilon$---constrains this enhancement in the present ansatz, so the large-$d_\mtn{M}$ behavior must be interpreted together with the shrinking admissible range of $\varepsilon$. The main lesson is therefore not unrestricted growth with catalyst size, but rather that the pseudo-additive correction is sensitive not only to the magnitude of the catalytic deviation, but also to its spectral profile.

\paragraph*{Physical interpretation.}

These two examples sharpen the main finite-size message of the pseudo-additive framework. In approximate uncorrelated catalysis, the correction to the monotonicity condition is not determined solely by the trace-distance return error $\varepsilon$; it is also sensitive to the \textit{shape} of the catalytic perturbation.

This is precisely where the finite-size content of the pseudo-additive framework becomes operational. The continuity bounds provide worst-case control, but structured finite-dimensional ans\"atze reveal something stronger: for a fixed target transformation $\rho_\mtn{S}\mapsto \rho'_\mtn{S}$, there may exist dimensions for which no admissible catalyst profile within the allowed trace-distance budget can satisfy the non-additive monotonicity condition, whereas larger dimensions or different spectral profiles may restore feasibility. Thus, rather than yielding a universal sharp dimension bound in closed form, the non-additive formalism naturally induces a \textit{finite-size feasibility landscape}.

\paragraph*{Comparison with the additive R\'enyi balance.}

For comparison, an additive R\'enyi-based treatment of approximate uncorrelated catalysis with a trivial catalyst Hamiltonian separates the total generalized free-energy change into a system contribution and a distinct catalyst contribution. For the same two finite-dimensional ans\"atze considered above, this additive catalyst term exhibits the same qualitative scaling contrast as in the non-additive case: the distributed perturbation yields a leading correction of order $\varepsilon^2$, whereas the concentrated two-level perturbation produces an amplified contribution of order $d_\mtn{M}\varepsilon^2$ within its admissible regime. This shows that the profile dependence exhibited above is not specific to the non-additive formulation: the additive R\'enyi divergences of the catalyst likewise distinguish between different spectral distributions of the same trace-distance deviation. The distinguishing feature of the present non-additive formulation is instead structural: this dependence enters directly into the generalized monotonicity condition through the pseudo-additive term, where it couples to the system contribution. In contrast, within additive formulations, the same dependence appears only through the catalyst divergence itself and does not alter the structure of the system-level free-energy balance.

\subsection{Beyond the trivial-Hamiltonian benchmark: a gapped catalyst} \label{app:gapped-catalyst}

We now lift the trivial-Hamiltonian assumption of the preceding subsections. Let the catalyst have a nondegenerate Hamiltonian with full-rank thermal reference $\gamma_\mtn{M}=\mathrm{diag}(g_1,\dots,g_{d_\mtn{M}})$ at the bath inverse temperature $\beta$, and let $\sigma_\mtn{M}=\mathrm{diag}(\vec{p})$ and $\sigma'_\mtn{M}=\mathrm{diag}(\vec{q})$ be diagonal in the same basis, with $\tfrac12\|\vec q - \vec p\|_1\le\varepsilon$. We first close the joint free-energy change into an exact identity, then derive the gap-sensitive continuity bounds and the resulting $\varepsilon$-budgets, and finally illustrate the mechanism for a two-level catalyst with an energy gap.

\paragraph*{Exact identity.}
For $\alpha>0$, write $D_\mtn{S}\coloneqq D_\mtn{\alpha}(\rho_\mtn{S}\|\gamma_\mtn{S})$ and introduce
\begin{equation}
W_\mtn{\alpha}(\vec{v})\coloneqq\sum_i v_i^{\alpha}g_i^{1-\alpha}
=1+(\alpha-1)D_\mtn{\alpha}(\mathrm{diag}(\vec v)\|\gamma_\mtn{M}) .
\label{eq:W-def}
\end{equation}
Pseudo-additivity gives $D_\mtn{\alpha}(\rho_\mtn{S}\otimes\sigma_\mtn{M}\|\gamma_\mtn{S}\otimes\gamma_\mtn{M})=D_\mtn{S}\,W_\mtn{\alpha}(\vec p)+[W_\mtn{\alpha}(\vec p)-1]/(\alpha-1)$, so that, using $k_BT\,D_\mtn{S}=F_\mtn{\alpha}(\rho_\mtn{S})+k_BT\ln Z_\mtn{S}$ and $A_\mtn{\alpha}=k_BT\ln Z_\mtn{S}+k_BT/(\alpha-1)$,
\begin{widetext}
\begin{equation}
F_\mtn{\alpha}(\rho_\mtn{S}\otimes\sigma_\mtn{M},H_\mtn{S}+H_\mtn{M})
=
W_\mtn{\alpha}(\vec p)\bigl[F_\mtn{\alpha}(\rho_\mtn{S})+A_\mtn{\alpha}\bigr]-A_\mtn{\alpha}-k_BT\ln Z_\mtn{M}.
\label{eq:Fjoint-closed}
\end{equation}
\end{widetext}
Taking the difference of Eq.~\eqref{eq:Fjoint-closed} between $(\rho'_\mtn{S},\vec q)$ and $(\rho_\mtn{S},\vec p)$ and rearranging yields the exact identity of Eq.~\eqref{eq:DeltaF-general-identity}. For uniform $\gamma_\mtn{M}$, $W_\mtn{\alpha}=d_\mtn{M}^{\alpha-1}P_\mtn{\alpha}$ and Eq.~\eqref{eq:DeltaF-general-identity} reduces to Eq.~\eqref{eq:DeltaFalpha-limit-form}.

\paragraph*{Continuity bounds and \(\varepsilon\)-budgets.}
Since $\Delta F_\mtn{\alpha}\le W_\mtn{\alpha}(\vec p)\,\Delta F_\mtn{\alpha}^{\mtn{(S)}}+|W_\mtn{\alpha}(\vec q)-W_\mtn{\alpha}(\vec p)|\,|F_\mtn{\alpha}(\rho'_\mtn{S})+A_\mtn{\alpha}|$, 

monotonicity is guaranteed by the $W$-analog of the sufficient condition~\eqref{eq:tradeoff-sufficient},
\[
|W_\mtn{\alpha}(\vec{q})-W_\mtn{\alpha}(\vec{p})|\;\bigl|F_\mtn{\alpha}(\rho'_\mtn{S})+A_\mtn{\alpha}\bigr|
\;\le\;
W_\mtn{\alpha}(\vec{p})\,\bigl(-\Delta F_\mtn{\alpha}^{\mtn{(S)}}\bigr).
\]
For $\alpha\ge 1$, the mean-value estimate $|q_i^{\alpha}-p_i^{\alpha}|\le\alpha|q_i-p_i|$ on $[0,1]$ gives
\begin{equation}
\bigl|W_\mtn{\alpha}(\vec q)-W_\mtn{\alpha}(\vec p)\bigr|
\le
\alpha\Bigl(\max_i g_i^{1-\alpha}\Bigr)\sum_i|q_i-p_i|
\le
2\alpha\varepsilon\, g_{\mathrm{min}}^{\,1-\alpha},
\label{eq:W-continuity-ge1}
\end{equation}
which yields Eq.~\eqref{eq:eps-bound-gapped}. For $0<\alpha<1$, the subadditivity estimate $|q_i^{\alpha}-p_i^{\alpha}|\le|q_i-p_i|^{\alpha}$ combined with the H\"older step used for Eq.~\eqref{eq:eps-bound-alpha-lt-1} gives $|W_\mtn{\alpha}(\vec q)-W_\mtn{\alpha}(\vec p)|\le g_{\mathrm{max}}^{\,1-\alpha}\,d_\mtn{M}^{1-\alpha}(2\varepsilon)^{\alpha}$ with $g_{\mathrm{max}}\coloneqq\max_i g_i$, and hence
\begin{equation}
\varepsilon
\;\le\;
\frac12\left[
\frac{W_\mtn{\alpha}(\vec p)\,\bigl(-\Delta F_\mtn{\alpha}^{\mtn{(S)}}\bigr)}{g_{\mathrm{max}}^{\,1-\alpha}\,d_\mtn{M}^{1-\alpha}\,\bigl|F_\mtn{\alpha}(\rho_\mtn{S}')+A_\mtn{\alpha}\bigr|}
\right]^{1/\alpha}.
\label{eq:eps-bound-gapped-lt1}
\end{equation}
For uniform $\gamma_\mtn{M}$ ($g_{\mathrm{min}}=g_{\mathrm{max}}=1/d_\mtn{M}$), Eqs.~\eqref{eq:eps-bound-gapped} and~\eqref{eq:eps-bound-gapped-lt1} reduce exactly to Eqs.~\eqref{eq:eps-bound-alpha-ge-1} and~\eqref{eq:eps-bound-alpha-lt-1}, respectively.

\paragraph*{Two-level catalyst with an energy gap.}
For $d_\mtn{M}=2$ with gap $E_\mtn{M}>0$, the thermal weights are $g_0=1/(1+e^{-\beta E_\mtn{M}})$ and $g_1=e^{-\beta E_\mtn{M}}/(1+e^{-\beta E_\mtn{M}})$, so $g_{\mathrm{min}}=g_1$ and
\begin{equation}
g_{\mathrm{min}}^{\,1-\alpha}
=
e^{(\alpha-1)\beta E_\mtn{M}}\bigl(1+e^{-\beta E_\mtn{M}}\bigr)^{\alpha-1}
\;\xrightarrow{\;\beta E_\mtn{M}\gg 1\;}\;
e^{(\alpha-1)\beta E_\mtn{M}} .
\label{eq:gmin-gap}
\end{equation}
Consequently, for $\alpha>1$ the admissible return error of Eq.~\eqref{eq:eps-bound-gapped} contracts exponentially in the gap, $\varepsilon_{\mathrm{max}}\propto e^{-(\alpha-1)\beta E_\mtn{M}}$. If, in addition, the catalyst is prepared thermally, $\vec p=(g_0,g_1)$, then $W_\mtn{\alpha}(\vec p)=1$ and the bound takes the particularly simple form $\varepsilon\le(-\Delta F_\mtn{\alpha}^{\mtn{(S)}})/[2\alpha\,g_{\mathrm{min}}^{\,1-\alpha}|F_\mtn{\alpha}(\rho'_\mtn{S})+A_\mtn{\alpha}|]$. In contrast with the benchmark examples of Sec.~\ref{app:finite-benchmark-examples}, where the correction is governed by dimension and spectral profile at zero gap, here the gap enters through the Boltzmann factor of the catalyst spectrum: at fixed dimension, a stiffer catalyst leaves an exponentially smaller budget for its approximate return.

\renewcommand{\theequation}{J\arabic{equation}}
\setcounter{equation}{0}  
\appsection{Correlated catalytic thermal transformations beyond reduced-state monotones}  \label{app:corrcatal}

We now turn to correlated catalytic thermal transformations, where the catalyst is preserved only marginally and is allowed to become correlated with the system during the transition. In the correlated-catalytic setting, one considers transformations of the form
\begin{equation}
\rho_\mtn{S} \otimes \sigma_\mtn{M} \;\longmapsto\; \rho'_\mtn{SM},
\end{equation}
such that the final marginals satisfy
\begin{equation}
\mathrm{tr}_\mtn{M}[\rho'_\mtn{SM}] = \rho'_\mtn{S},
\qquad
\mathrm{tr}_\mtn{S}[\rho'_\mtn{SM}] = \sigma_\mtn{M}.
\end{equation}

Thus, the catalyst can be reused locally, even though the final joint state need not be a product state. In the asymptotic regime of negligibly small correlations, the infinite family of generalized second laws effectively collapses to the Helmholtz free-energy inequality~\cite{2015_PRL_Lostaglio_Muller, 2017_Entropy_CorrelatedCatal, PhysRevX.8.041051, 2022_PRL_CorrelatedCatal, 2025_PRL_CorrCatalAndFreeE, 2025_arXiv_CorrCatalAndFreeE, 2025_EPL_Shiraishi_CorrelatedCatalyzerAndUniqueF}. Our interest here is different: we focus on finite-dimensional regimes in which the generated system--catalyst correlations are not negligible and can materially affect the thermo-majorization verdict for the joint transition.

The key question is whether any family of second laws defined only in terms of reduced-state quantities can remain complete in this regime. The examples below show that the answer is negative. More precisely, we exhibit explicit finite-dimensional transitions in which the initial state, the thermal reference state, and the final system and catalyst marginals are all held fixed, yet the thermo-majorization accessibility of the \textit{joint} transformation changes when only the correlation structure of the final state is varied. In this sense, the relevant thermodynamic information is not fully encoded in the reduced state of the system, nor even in the pair of reduced states alone.

\subsection{Common initial product state and thermal reference}

In all examples below, the initial state is taken to be a product of two local thermal states,
\begin{equation}
\rho_\mtn{SM}^\mtn{\mathrm{uc}}=\rho_\mtn{S}^\mtn{\beta_2}\otimes \rho_\mtn{M}^\mtn{\beta_1} ,
\label{eq:rhoP-product}
\end{equation}
where the subsystem $S$ is initially in thermal equilibrium at inverse temperature $\beta_2$, and the subsystem $M$ (the catalyst) at $\beta_1$.  Both subsystems are two-level systems governed by the common local Hamiltonian of Eq.~\eqref{eq:qubit-hamiltonian}:
\[
H_\mtn{S} = H_\mtn{M} = E_g |g\rangle \langle g| + E_e |e\rangle \langle e | .
\]

In the ordered product basis $\{|gg\rangle,|ge\rangle,|eg\rangle,|ee\rangle\}$, with $S$ the first tensor factor, the initial state reads
\begin{widetext}
\begin{equation}
\rho_\mtn{SM}^\mtn{\mathrm{uc}} =
\frac{1}{(e^{E_e\beta_1}+e^{E_g\beta_1})(e^{E_e\beta_2}+e^{E_g\beta_2})}
\begin{pmatrix}
e^{E_e(\beta_1+\beta_2)} & 0 & 0 & 0\\[2pt]
0 & e^{E_e\beta_2+E_g\beta_1} & 0 & 0\\[2pt]
0 & 0 & e^{E_g\beta_2+E_e\beta_1} & 0\\[2pt]
0 & 0 & 0 & e^{E_g(\beta_1+\beta_2)}
\end{pmatrix}.
\label{eq:rhoP-explicit}
\end{equation}
\end{widetext}

For the final states, we choose the reduced states of $M$ and $S$ to be fixed thermal states at inverse temperatures $\beta_1$ and $\beta_3$, respectively:
\begin{equation}
\rho^\prime_\mtn{M} =
\begin{pmatrix}
p_\mtn{\beta_1} & 0\\
0 & 1-p_\mtn{\beta_1}
\end{pmatrix},
\qquad
\rho^\prime_\mtn{S} =
\begin{pmatrix}
p_\mtn{\beta_3} & 0\\
0 & 1-p_\mtn{\beta_3}
\end{pmatrix},
\label{eq:fixed-marginals}
\end{equation}
with the ground-state populations
\begin{equation}
p_\mtn{\beta_1} = \frac{e^{-E_g\beta_1}}{e^{-E_e\beta_1}+e^{-E_g\beta_1}},
\qquad
p_\mtn{\beta_3} = \frac{e^{-E_g\beta_3}}{e^{-E_e\beta_3}+e^{-E_g\beta_3}}.
\label{eq:pApB-def}
\end{equation}
Thus, in all cases considered below, the catalyst is locally returned to its initial thermal state, while the system is transformed to a different local thermal state. Note that $\beta_1$,  $\beta_2$, and $\beta_3$ characterize the chosen local marginals, whereas $\beta_b$ is the inverse temperature of the bath defining the thermal operations. The numerical values used in the examples are those of Eq.~\eqref{eq:numerical-parameters}:
\[
E_g = 0,
\qquad
E_e = 2,
\qquad
\beta_1 = 0.1,
\qquad
\beta_2 = 0.2,
\qquad
\beta_3 = 1,
\qquad
\beta_b = 2.
\]

With these choices, the initial state $\rho_\mtn{SM}^\mtn{\mathrm{uc}}$, the final marginals $\rho^\prime_\mtn{M},\rho^\prime_\mtn{S}$, and the thermal reference state for the joint thermo-majorization test are fixed throughout. Only the \textit{correlation structure} of the final joint state is varied.

At the reduced-state level, this immediately implies that all the generalized free energies of the system based on R\'enyi (additive) and non-additive divergences depend only on the fixed transition $\rho_\mtn{S}^\mtn{\beta_2}\mapsto \rho_\mtn{S}^\mtn{\beta_3}$ and are therefore completely insensitive to the detailed final correlations. In particular, for the above parameters, the changes in the system R\'enyi divergences are
\begin{widetext}
\begin{equation}
\Delta D^\mtn{R}_{0} = 0,
\qquad
\Delta D^\mtn{R}_{1} = -41.01 \times 10^{-2},
\qquad
\Delta D^\mtn{R}_{\infty} = -60.70 \times 10^{-2} .
\label{eq:renyi-system-values}
\end{equation}
\end{widetext}
For non-additive divergences, the cases $\alpha=0$ and $\alpha=1$  coincide with the corresponding reduced-state verdicts, while the $\alpha\to\infty$ expression becomes ill-defined in this parameterization. The essential point is that all such reduced-state quantities are independent of the final correlation parameters introduced below.

\subsection{Example 1: classical correlations with fixed marginals}

Our first family of final states is the purely classically correlated family $\rho_\mtn{SM}^\mtn{\mathrm{cc}}(\chi)$ of Eq.~\eqref{eq:rhoC},
\[
\rho_\mtn{SM}^\mtn{\mathrm{cc}}(\chi) =
\begin{pmatrix}
p_\mtn{\beta_3} p_\mtn{\beta_1} + \chi & 0 & 0 & 0\\
0 & p_\mtn{\beta_3}(1-p_\mtn{\beta_1})-\chi & 0 & 0\\
0 & 0 & (1-p_\mtn{\beta_3})p_\mtn{\beta_1}-\chi & 0\\
0 & 0 & 0 & (1-p_\mtn{\beta_3})(1-p_\mtn{\beta_1})+\chi
\end{pmatrix},
\]

where positivity requires
\begin{equation}
-5.37 \times 10^{-2} < \chi < 6.55 \times 10^{-2}
\label{eq:chi-range}
\end{equation}
for the parameter values in Eq.~\eqref{eq:numerical-parameters}. By construction,
\begin{equation}
\mathrm{tr}_\mtn{S}[\rho_\mtn{SM}^\mtn{\mathrm{cc}}(\chi)] = \rho^\prime_\mtn{M},
\qquad
\mathrm{tr}_\mtn{M}[\rho_\mtn{SM}^\mtn{\mathrm{cc}}(\chi)] = \rho^\prime_\mtn{S},
\end{equation}
so the reduced states are independent of $\chi$.

This is already enough to establish the first key point: any family of generalized second laws built only from reduced-state functionals of $S$ (or even from the pair of reduced states $(\rho^\prime_\mtn{M},\rho^\prime_\mtn{S})$) is completely blind to the parameter $\chi$, even though $\chi$ changes the joint state and, as we now show, can change the thermo-majorization verdict.

For $\chi = 5 \times 10^{-2}$, the mutual information of $\rho_\mtn{SM}^\mtn{\mathrm{cc}}$ is
\begin{equation}
I(S:M) = H_\mtn{S} + H_\mtn{M} - H_\mtn{SM} = \text{$5.17 \times 10^{-2}$},
\end{equation}
and the transition $\rho_\mtn{SM}^\mtn{\mathrm{uc}} \mapsto \rho_\mtn{SM}^\mtn{\mathrm{cc}}$ is \textit{allowed} by thermo-majorization. For $\chi = 6.5 \times 10^{-2}$, the mutual information increases to
\begin{equation}
I(S:M) = \text{$10.14 \times 10^{-2}$},
\end{equation}
but the transition $\rho_\mtn{SM}^\mtn{\mathrm{uc}} \mapsto \rho_\mtn{SM}^\mtn{\mathrm{cc}}$ becomes \textit{forbidden}. Since the initial state, the final marginals, and the thermal reference state are unchanged, the only difference between these two cases is the amount of classical correlation in the final joint state. Therefore, thermo-majorization distinguishes between two transformations that are indistinguishable at the reduced-state level.

This establishes that, in finite-dimensional correlated catalysis, no family of second laws defined solely in terms of reduced-state monotones can provide a complete characterization of state accessibility. The thermo-majorization verdict depends on genuinely joint-state information that is not recoverable from the marginals alone.

\subsection{Example 2: same marginals, same mutual information, different correlation structure}

The previous example shows that the \textit{amount} of correlation affects accessibility even when the reduced states are fixed. We now show that even the amount of correlation is not, by itself, sufficient to determine accessibility. To this end, consider the discordant family $\rho_\mtn{SM}^\mtn{\mathrm{qc}}(\lambda)$ of Eq.~\eqref{eq:rhoD},
\[
\rho_\mtn{SM}^\mtn{\mathrm{qc}}(\lambda) =
\begin{pmatrix}
p_\mtn{\beta_3} p_\mtn{\beta_1} & 0 & 0 & 0\\
0 & p_\mtn{\beta_3}(1-p_\mtn{\beta_1}) & \lambda & 0\\
0 & \lambda & (1-p_\mtn{\beta_3})p_\mtn{\beta_1} & 0\\
0 & 0 & 0 & (1-p_\mtn{\beta_3})(1-p_\mtn{\beta_1})
\end{pmatrix},
\]

which has the same marginals as $\rho_\mtn{SM}^\mtn{\mathrm{cc}}(\chi)$:
\begin{equation}
\mathrm{tr}_\mtn{S}[\rho_\mtn{SM}^\mtn{\mathrm{qc}}(\lambda)] = \rho^\prime_\mtn{M},
\qquad
\mathrm{tr}_\mtn{M}[\rho_\mtn{SM}^\mtn{\mathrm{qc}}(\lambda)] = \rho^\prime_\mtn{S}.
\end{equation}
Thus, the reduced-state generalized free energies are again identical to those of the previous example. 

Choosing
\begin{equation}
\lambda = 9.47 \times 10^{-2},
\end{equation}
one obtains
\begin{equation}
I(S:M) = \text{$5.17 \times 10^{-2}$},
\end{equation}
which is numerically equal to the mutual information of the classically correlated state $\rho_\mtn{SM}^\mtn{\mathrm{cc}}$ at $\chi=5 \times 10^{-2}$, for which the transition $\rho_\mtn{SM}^\mtn{\mathrm{uc}} \mapsto \rho_\mtn{SM}^\mtn{\mathrm{cc}}$ is allowed. However, the transition
\begin{equation}
\rho_\mtn{SM}^\mtn{\mathrm{uc}} \mapsto \rho_\mtn{SM}^\mtn{\mathrm{qc}}
\end{equation}
is \textit{forbidden} by thermo-majorization.

This comparison sharpens the previous conclusion. The states $\rho_\mtn{SM}^\mtn{\mathrm{cc}}(\chi=5 \times 10^{-2})$ and $\rho_\mtn{SM}^\mtn{\mathrm{qc}}(\lambda=9.47 \times 10^{-2})$ have identical reduced states and essentially the same mutual information; they differ only in their internal correlation structure---one purely classical, the other discordant. Yet thermo-majorization assigns different accessibility verdicts to the corresponding joint transitions. Therefore, even supplementing reduced-state data with a scalar correlation measure such as the mutual information is not sufficient, in general. Thermodynamic accessibility can depend not only on the amount of correlation, but also on its nature.

\section*{Data Availability}
There are no publicly available research data or software supporting this manuscript. Requests for further information or data should be sent to the authors.

\bibliographystyle{apsrev4-2}

%

\end{document}